\documentclass[11pt,reqno]{amsart}

\pdfoutput = 1

\usepackage[english]{babel}
\usepackage{amsmath,amsfonts,amssymb,amsthm}
\usepackage{amsfonts}
\usepackage{amsxtra}
\usepackage[numbers]{natbib}
\usepackage{mathabx,epsfig}

\usepackage{xcolor}
\usepackage{array,dcolumn}
\usepackage{graphicx}
\usepackage{caption}
\usepackage{float}
\usepackage[hyperfootnotes=true,
        pdffitwindow=true,
        plainpages=false,
        pdfpagelabels=true,
        pdfpagemode=UseOutlines,
        pdfpagelayout=SinglePage,
        pagebackref,
        hyperindex,hidelinks]{hyperref}

\usepackage[T1]{fontenc}
\usepackage[utf8]{inputenc}
\usepackage[activate={true,nocompatibility},spacing,kerning]{microtype}
\usepackage[sc,osf]{mathpazo}
\microtypecontext{spacing=nonfrench}

 \calclayout

\newcommand{\mathsym}[1]{{}}
\newcommand{\unicode}[1]{{}}

\theoremstyle{plain}
\newtheorem{theorem}{Theorem}
\newtheorem{lemma}{Lemma}

\newtheorem{proposition}{Proposition}

\theoremstyle{definition}
\newtheorem{definition}{Definition}
\newtheorem{example}{Example}

\theoremstyle{remark}

 \DeclareMathOperator{\Tr}{Tr}
 
\allowdisplaybreaks

\newcommand{\+}{\!+\!}

\renewcommand{\leq}{\leqslant}
\renewcommand{\geq}{\geqslant}

\newcommand{\Det}{\operatorname{Det}}
\def\mean#1{\left\langle #1 \right\rangle}

\begin{document}
\title[Antisymmetrised and Hermitised matrix product ensembles]{Combinatorics and loop equations for antisymmetrised and Hermitised matrix product ensembles}
\author{Stephane Dartois}
\address{Laboratoire Borelais de Recherche en Informatique, 
Universit\'e de Bordeaux, 
Bordeaux, France}
\email{stephane.dartois@labri.fr}
\author{Anas A. Rahman}
\address{Department of Mathematics, 
The University of Hong Kong,
Hong Kong}
\email{aarahman@hku.hk}

\begin{abstract}
The order $m$ antisymmetrised matrix product ensemble is represented by the product $X_1^T\cdots X_m^TJX_m\cdots X_1$, where $X_1,\ldots,X_m$ are independent real Ginibre matrices and $J$ is the elementary antisymmetric matrix, while the order $m$ Hermitised matrix product ensemble is represented by $X_1^\dagger\cdots X_m^\dagger HX_1\cdots X_m$, where $X_1,\ldots,X_m$ are now independent complex Ginibre matrices and $H$ is a Hermitian matrix drawn from the Gaussian unitary ensemble. These ensembles have recently been shown to be related to certain Muttalib--Borodin ensembles and integrals of Harish-Chandra--Itzykson--Zuber type, thereby motivating further investigation into their eigenvalue statistics. In this work, we construct ribbon graphs and constellations that are enumerated by the mixed cumulants of these ensembles and give loop equation characterisations for the generating functions of said cumulants when $m=1$.
\end{abstract}
\maketitle

\section{Introduction}\label{s1}
Products of random matrices have been studied since the 1950s, with many fundamental results and techniques being developed in the works \citep{GN50}, \citep{HC57}, \citep{FK60}, \citep{Os68}, \citep{Ra79}, \citep{New86} (among others) throughout the second half of the twentieth century; see \citep{Hel78} (group theory) and \citep{CPV93} (statistical physics) for textbook treatments. In addition to their interesting structures, research into random matrix products has been partially motivated by applications in the study of wireless communications \citep{Mu02}, \citep{TV04}, quantum transport \citep{Be97}, the stability and chaoticity of large dynamical systems \citep{CPV93}, \citep{IF18}, finance \citep{BLMP07}, and the stability of neural networks \citep{PSG17}, \citep{HN20}. There are also interesting connections to Fuss--Catalan and Raney numbers \citep{PZ11}, \citep{FL15}. The last decade has seen a remarkable surge in interest in matrix product ensembles owing to the breakthrough derivation of correlation kernels encapsulating detailed knowledge of the eigenvalue and singular value statistics of certain matrix product ensembles, where the products contain a finite number of factors and each factor is a finite-sized matrix --- most results from before this era pertain to regimes where either the matrix sizes or number of factors are taken to infinity. See \citep{AI15}, \citep{Ips15} for reviews.

In the present work, we use the Isserlis--Wick theorem to develop combinatorial characterisations of the mixed moments and cumulants of the Hermitised and antisymmetrised matrix product ensembles and then, specialising to the smallest non-trivial products, we derive loop equations that govern them. These ensembles are defined as follows.
\begin{definition} \label{def1}
Let $m,N_0\in\mathbb{N}$, fix $\nu_0:=0,\nu_1,\ldots,\nu_m\in\mathbb{N}$, and define $N_i:=N_0+\nu_i$ with $N:=N_m$. For $i\in[m]:=\{1,\ldots,m\}$, let $X_i$ be drawn independently from the (scaled) $N_{i-1}\times N_i$ complex Ginibre ensemble with probability density function (p.d.f.)
\begin{equation} \label{eq1.1}
P^{(cGin)}(X_i)\propto\exp\left[-\sqrt{N_iN_{i-1}}\Tr X_i^\dagger X_i\right]
\end{equation}
and let $H$ be drawn from the $N_0\times N_0$ Gaussian unitary ensemble (GUE) with p.d.f.
\begin{equation*}
P^{(GUE)}(H)\propto\exp\left[-\frac{N_0}{2}\Tr H^2\right].
\end{equation*}
Then, the $N\times N$ product
\begin{equation} \label{eq1.2}
\mathcal{H}_m:=X_m^\dagger\cdots X_1^\dagger HX_1\cdots X_m
\end{equation}
represents the $(N_0,\ldots,N_{m-1},N)$ \textit{Hermitised matrix product ensemble} \citep{FIL18}. When $m=1$, we alternatively say that $\mathcal{H}_1$ represents the $(N_0,N)$ \textit{Hermitised Laguerre ensemble}.
\end{definition}
\begin{definition} \label{def2}
Let $m,N_0/2\in\mathbb{N}$, fix $\nu_0:=0,\nu_1,\ldots,\nu_m\in\mathbb{N}$, and define $N_i:=N_0+2\nu_i$ with $N:=N_m$. For $i\in[m]$, let $X_i$ be drawn independently from the $N_{i-1}\times N_i$ real Ginibre ensemble with p.d.f.
\begin{equation} \label{eq1.3}
P^{(rGin)}(X_i)\propto\exp\left[-\frac{\sqrt{N_iN_{i-1}}}{2}\Tr X_i^TX_i\right]
\end{equation}
and let
\begin{equation*}
J:=I_{N_0/2}\otimes\begin{bmatrix}0&-1\\1&0\end{bmatrix}
\end{equation*}
be the elementary antisymmetric $N_0\times N_0$ block-diagonal matrix with non-zero blocks given by the $2\times2$ matrix on the right-hand side of the above. Then, the $N\times N$ product
\begin{equation} \label{eq1.4}
\mathcal{J}_m:=X_m^T\cdots X_1^TJX_1\cdots X_m
\end{equation}
represents the $(N_0,\ldots,N_{m-1},N)$ \textit{antisymmetrised matrix product ensemble} \citep{FILZ19}. When $m=1$, we alternatively say that $\mathcal{J}_1$ represents the $(N_0,N)$ \textit{antisymmetrised Laguerre ensemble}.
\end{definition}
We think of these ensembles as perturbations of the more widely studied Wishart product ensembles \citep{AI15}, \citep{Ips15}, which are defined as follows.
\begin{definition} \label{def3}
As in Definition \ref{def1}, let $m,N_0\in\mathbb{N}$, fix $\nu_0:=0,\nu_1,\ldots,\nu_m\in\mathbb{N}$, and define $N_i:=N_0+\nu_i$ with $N:=N_m$. For $i\in[m]$, let $X_i$ be drawn independently from the $N_{i-1}\times N_i$ real or complex Ginibre ensemble, as specified by the p.d.f.s \eqref{eq1.1}, \eqref{eq1.3} above (fixing the field across all $X_i$). Then, the $N\times N$ product
\begin{equation*}
\mathcal{W}_m:=X_m^\dagger\cdots X_1^\dagger X_1\cdots X_m
\end{equation*}
represents the $(N_0,\ldots,N_{m-1},N)$ \textit{real or complex Wishart product ensemble}, respectively. Setting $m=1$ recovers the well known Laguerre ensemble.
\end{definition}

Our motivation for studying the ensembles of definitions \ref{def1} and \ref{def2} is twofold. Firstly, we would like to see how the recent loop equation analysis of the first present author and Forrester \citep{DF20} for the $(N,N,N)$ complex Wishart product ensemble extends to the ensembles at hand, thereby improving our understanding of the ramifications of perturbing the Wishart product ensembles in the corresponding ways (inserting a GUE matrix $H$ or elementary antisymmetric matrix $J$ in the middle of the relevant products). Secondly, the Hermitised and antisymmetrised matrix product ensembles have recently been shown to be related to certain Muttalib--Borodin ensembles and a generalisation, respectively analogue, of the Harish-Chandra--Itzykson--Zuber (HCIZ) integral \citep{FIL18}, \citep{FILZ19}.

\subsection{Related literature}
The Laguerre and Hermite Muttalib--Borodin (biorthogonal) ensembles are systems of real eigenvalues $\{\lambda_i\}_{i=1}^N$ distributed according to the joint p.d.f.s (j.p.d.f.s)
\begin{align*}
p^{(L,\theta)}(\lambda_1,\ldots,\lambda_N;a)&=\frac{1}{\mathcal{N}^{(L)}_{N,\theta}}\prod_{i=1}^N\lambda_i^ae^{-\lambda_i}\chi_{\lambda_i>0}\prod_{1\le j<k\le N}(\lambda_k-\lambda_j)(\lambda_k^\theta-\lambda_j^\theta),
\\p^{(H,\theta)}(\lambda_1,\ldots,\lambda_N;a)&=\frac{1}{\mathcal{N}^{(H)}_{N,\theta}}\prod_{i=1}^N|\lambda_i|^ae^{-\lambda_i}\prod_{1\le j<k\le N}(\lambda_k-\lambda_j)((\mathrm{sgn}\,\lambda_k)|\lambda_k|^\theta-(\mathrm{sgn}\,\lambda_j)|\lambda_j|^\theta),
\end{align*}
respectively. Here, $\theta>0$, $\alpha>-1$ are real parameters, $\chi_{\lambda>0}$ is the indicator function equal to one when $\lambda>0$ and zero otherwise, $\mathrm{sgn}\,\lambda$ is the sign of $\lambda$, and $\mathcal{N}^{(L)}_{N,\theta},\mathcal{N}^{(H)}_{N,\theta}$ are normalisation constants. These ensembles were first introduced by Borodin \citep{Bo98} as extensions of a model of Muttalib \citep{Mut95} that was itself proposed as an exactly solvable simplification of the solution to the so-called DMPK equation \citep{Dor82}, \citep{MPK88}, \citep{BR93}. They have since drawn considerable attention \citep{FL15}, \citep{FW17}, \citep{FI18}, \citep{Che18}, \citep{CGS19}; see also \citep{CR14}, \citep{KS14}, \citep{KM19}, \citep{CGS19}, \citep{CLM21}, \citep{Mol21}, \citep{WZ21}, \citep{CZ26} for various generalisations and \citep{LSZ06}, \citep{Miller26} for some applications in physics. Our interest in these ensembles stems from the fact that at the large eigenvalue limit, the j.p.d.f.s of the eigenvalues of $\mathcal{H}_m$ and positive eigenvalues of $\mathrm{i}\mathcal{J}_m$ have the following approximations, respectively \citep{FLZ15}, \citep{Ips15}, \citep{FIL18}, \citep{FILZ19}:
\begin{align*}
p^{(\mathcal{H}_m)}\left(\tfrac{\gamma_1^{2m+1}}{\sqrt{2}},\ldots,\tfrac{\gamma_{N_0}^{2m+1}}{\sqrt{2}}\right)&\approx \prod_{i=1}^{N_0}\frac{(m+1/2)^{m-1/2}}{\sqrt{2}\lambda_i^{2m}}\,p^{(H,2m+1)}\left(\lambda_1,\ldots,\lambda_{N_0};\,2\sum_{i=1}^m\nu_i+m\right),
\\ p^{(\mathrm{i}\mathcal{J}_m)}\left(\eta_1^m,\ldots,\eta_{N_0/2}^m\right)&\approx \prod_{i=1}^{N_0/2}\left(\tfrac{\lambda_i}{m}\right)^{1-m}\,p^{(L,2m)}\left(\lambda_1,\ldots,\lambda_{N_0/2};\,2\sum_{i=1}^m\nu_i+\tfrac{m-1}{2}\right).
\end{align*}
Here, $\gamma_i=N_0N_1\cdots N_{m-1}\sqrt{N_m/(2m+1)}\lambda_i$, $\eta_i=N_1\cdots N_{m-1}\sqrt{N_0N_m}/(2^mm)\lambda_i$, and we note that the latter approximation is exact when $m=1$ \citep[Cor.~4.4]{FILZ19}.

The Harish-Chandra--Itzykson--Zuber (HCIZ) integral formula is the evaluation
\begin{equation*}
\int_{U(N)}e^{\Tr(AUBU^\dagger)}\,\mathrm{d}\mu_{\mathrm{Haar}}(U)=\prod_{i=1}^{N-1}i!\,\frac{\Det[e^{a_jb_k}]_{j,k=1}^N}{\Delta_N(a)\Delta_N(b)},
\end{equation*}
where $A$ and $B$ are $N\times N$ complex Hermitian matrices with eigenvalues $\{a_j\}_{j=1}^N$ and $\{b_j\}_{j=1}^N$, $\Delta_N(\lambda)=\prod_{1\le i<j\le N}(\lambda_j-\lambda_i)$ is the Vandermonde determinant, and $\mathrm{d}\mu_{\mathrm{Haar}}(U)$ is the Haar probability measure on $U(N)$. This formula is so named due to it being a special case of the general group integral studied by Harish-Chandra in the 1957 work \citep{HC57}, and for its independent derivation and introduction to random matrix theory by Itzykson and Zuber in the 1980 work \citep{IZ80}. The HCIZ integral and its many generalisations and analogues have proven to be a powerful tool in random matrix theory (see, e.g., \citep{Joh01}, \citep{KT04}, \citep{FR05}, \citep{BK05}, \citep{BBP05}, \citep{SMM06}, \citep{AKW13}, \citep{AIK13}, \citep{AS16}, \citep{KKS16}, \citep{Li18}).

Of interest to us are the orthogonal Harish-Chandra integral \citep{FEFZ07} and the hyperbolic HCIZ integral \citep{Fyo02}, \citep{FS02}. The former, defined for even $N$, is given by
\begin{equation*}
\int_{O(N)}e^{\tfrac{1}{2}\Tr(AUBU^T)}\,\mathrm{d}\mu_{\mathrm{Haar}}(U)=2^{-N/2}\prod_{i=1}^{N/2-1}(2i)!\,\frac{\Det[2\cosh(a_jb_k)]_{j,k=1}^{N/2}}{\prod_{1\leq j<k\leq N/2}(a_k^2-a_j^2)(b_k^2-b_j^2)},
\end{equation*}
where $A$ and $B$ are $N\times N$ real antisymmetric matrices with imaginary eigenvalues $\{\pm\mathrm{i}a_j\}_{j=1}^{N/2}$ and $\{\pm\mathrm{i}b_j\}_{j=1}^{N/2}$, while $\mathrm{d}\mu_{\mathrm{Haar}}(U)$ now denotes the Haar probability measure on $O(N)$. For the latter, fixing $M,N\in\mathbb{N}$ such that $0\leq M\leq N$, letting $\eta=\mathrm{diag}(-I_M,I_{N-M})$ be a pseudo-metric tensor, and taking $A$ and $B$ to be $N\times N$ complex Hermitian matrices with eigenvalues
\begin{align*}
&a_1<\cdots<a_M<0<a_{M+1}<\cdots<a_N,
\\ &b_1<\cdots<b_M<0<b_{M+1}<\cdots<b_N,
\end{align*}
the hyperbolic HCIZ integral satisfies the relation
\begin{equation*}
\int_{U(\eta)/U(1)^N}e^{-\Tr(AUBU^{-1})}\,\mathrm{d}\tilde{\mu}_{\mathrm{Haar}}(U)\propto\frac{\Det[e^{-a_jb_k}]_{j,k=1}^M\Det[e^{-a_jb_k}]_{j,k=M+1}^N}{\Delta_N(a)\Delta_N(b)},
\end{equation*}
where $U(\eta):=\{U\in GL_N(\mathbb{C})\,|\,U^\dagger\eta U=U\eta U^\dagger=\eta\}$ is the pseudo-unitary group with metric $\eta$ and $\mathrm{d}\tilde{\mu}_{\mathrm{Haar}}(U)$ is the Haar probability measure on the quotient space $U(\eta)/U(1)^N$.

The orthogonal Harish-Chandra integral was used in \citep{FILZ19} to study the eigenvalue statistics of the antisymmetrised matrix product ensembles of Definition \ref{def2} and it was shown in \citep{FIL18} that the eigenvalue j.p.d.f. of the Hermitised matrix product ensemble of Definition~\ref{def1} with $m\in\mathbb{N}$ and $\nu_1=\cdots=\nu_m=0$ can be obtained via repeated applications of the relation hyperbolic HCIZ integral. This marked the first instance of an HCIZ-type integral with non-compact domain of integration being successfully applied to the study of matrix product ensembles.

\subsection{Outline and summary of results} \label{s1.2}
The eigenvalues of the Hermitised matrix product ensembles are supported on the real line while those of the antisymmetrised matrix product ensembles are supported on the imaginary axis. Thus, they are fully characterised by their mixed moments
\begin{align}
m_{k_1,\ldots,k_n}^{(\mathcal{H}_m)}&:=\left\langle\prod_{s=1}^n\Tr \mathcal{H}_m^{k_s}\right\rangle, \label{eq1.5}
\\ m_{k_1,\ldots,k_n}^{(\mathcal{J}_m)}&:=\left\langle\prod_{s=1}^n\Tr \mathcal{J}_m^{k_s}\right\rangle, \label{eq1.6}
\end{align}
where $k_1,\ldots,k_n\in\mathbb{N}$ and $\langle\,\cdot\,\rangle$ denotes an average with respect to the appropriate p.d.f. (one simply takes a product of the p.d.f.s in Definition \ref{def1}, respectively \ref{def2}, and then multiplies by a Dirac delta enforcing equation \eqref{eq1.2}, respectively \eqref{eq1.4}). Our interest lies in these mixed moments, the corresponding mixed cumulants, and their respective generating functions. The mixed cumulants are defined implicitly by the moment-cumulants relation \citep[Ch.~2]{McC87}
\begin{equation} \label{eq1.7}
m_{k_1,\ldots,k_n}=\sum_{K\vdash\{k_1,\ldots,k_n\}}\prod_{\kappa_i\in K}c_{\kappa_i},
\end{equation}
where $K\vdash\{k_1,\ldots,k_n\}$ means that $K$ is a partition of $\{k_1,\ldots,k_n\}$, i.e., $K=\{\kappa_i\}_{i=1}^m$ for some $m\in[n]$ such that the disjoint union $\kappa_1\sqcup\cdots\sqcup\kappa_m$ is equal to $\{k_1,\ldots,k_n\}$. The corresponding generating functions are
\begin{align}
U_n(x_1,\ldots,x_n)&:=\sum_{k_1,\ldots,k_n=1}^\infty\frac{m_{k_1,\ldots,k_n}}{x_1^{k_1+1}\cdots x_n^{k_n+1}}. \label{eq1.8}
\\ W_n(x_1,\ldots,x_n)&:=\sum_{k_1,\ldots,k_n=1}^\infty\frac{c_{k_1,\ldots,k_n}}{x_1^{k_1+1}\cdots x_n^{k_n+1}}. \label{eq1.9}
\end{align}
We refer to them as unconnected and connected $n$-point correlators, respectively.

As the matrices $X_i$ in definitions \ref{def1} and \ref{def2} are Ginibre, their entries are independently and identically distributed (i.i.d.) Gaussian variables. Thus, the mixed moments and cumulants of the Hermitised and antisymmetrised matrix product ensembles may be computed via the Isserlis--Wick theorem \citep{Iss18}, \citep{Wic50}:
\begin{theorem}[Isserlis--Wick]
For $k\in\mathbb{N}$, let $x_1,\ldots,x_k$ be centred normal variables. Then, if $k$ is odd, the expectation $\mean{x_1\cdots x_k}$ is equal to zero, while for $k$ even, it decomposes as
\begin{equation*}
\mean{x_1\cdots x_k}=\sum_{\sigma\in\mathcal{P}_k}\prod_{\{u,v\}\in\sigma}\langle x_ux_v\rangle,
\end{equation*}
where $\mathcal{P}_k$ is the set of pairings of $[k]:=\{1,\ldots,k\}$.
\end{theorem}
In sections \ref{s2} and \ref{s3}, we use this theorem to show that the mixed moments and cumulants of $\mathcal{H}_m,\mathcal{J}_m$ have (terminating) genus expansions with coefficients being enumerations of certain $(m+1)$-coloured ribbon graphs, equivalently constellations. These genus expansions carry over to the $n$-point correlators and we have in particular that there exist \textit{correlator expansion coefficients} independent of $N$ such that
\begin{equation} \label{eq1.10}
W_n(x_1,\ldots,x_n)=N^{2-n}\sum_{l=0}^\infty\frac{W_n^l(x_1,\ldots,x_n)}{N^l};
\end{equation}
note that the equivalent expansion for $U_n(x_1,\ldots,x_n)$ is of order $N^2$, rather than $N^{2-n}$.

In sections \ref{s4} and \ref{s5}, we transition to a loop equation analysis of the objects of interest, specialising to $m=1$. The loop equation formalism is a storied tool that is, at its core, essentially an advanced application of integration by parts that relates the quantities defined above to products of similar objects with smaller index or variable sets. Loop equations are closely related to the concepts of Ward identities, Virasoro constraints, Schwinger--Dyson equations, Tutte equations, and topological recursion (see, e.g., \citep{Mig83}, \citep{ACKM93}, \citep{EO09}, \citep{Oxf15}).

In this work, we begin with loop equations on the mixed moments $m_{k_1,\ldots,k_n}$ of the ensembles at hand. Although these relations connect complicated moments to simpler ones, they do not result in closed recursions because the number of independent relations is always less than the number of indeterminate moments. Hence, we proceed by converting said loop equations into equivalents for the unconnected correlators $U_n(x_1,\ldots,x_n)$ and then, through an analogue of the moment-cumulants relation, loop equations for the connected correlators $W_n(x_1,\ldots,x_n)$. Then, inserting the genus expansion \eqref{eq1.10} and comparing terms of equal order in $N$ yields loop equations on the coefficients $W_n^l(x_1,\ldots,x_n)$ that are finally closed, meaning that they can be solved. We conclude each of sections \ref{s4} and \ref{s5} by solving the simplest of these loop equations and making explicit the asymptotic moment and covariance generating functions $W_1^0(x_1)$ and $W_2^0(x_1,x_2)$.

\setcounter{equation}{0}
\section{Genus expansions for Hermitised matrix product ensembles} \label{s2}
It is known that the mixed moments and cumulants of the Laguerre unitary ensemble (LUE) and GUE are enumerations of ribbon graphs weighted by powers of $N$. In the following, we first show how the same quantities for the Hermitised matrix product ensembles enumerate amalgamations of the LUE and GUE ribbon graphs, thereby proving that they have terminating genus expansions. Then, in \S\ref{s2.2}, we give a formulation of our results in terms of constellations \citep{LZ04}.

\subsection{Ribbon graphs for Hermitised matrix product ensembles} \label{s2.1}
We begin with the elementary observation that since the matrices $H,X_1,\ldots,X_m$ within Definition \ref{def1} are independent, writing out the mixed moment $m_{k_1,\ldots,k_n}^{(\mathcal{H}_m)}$ in terms of the entries of $\mathcal{H}_m$ allows one to decouple the average \eqref{eq1.5} into a sum of a product of $m+1$ distinct averages over each of $H,X_1,\ldots,X_m$.

\begin{lemma} \label{Lemma1}
Fix $m,n,k_1,\ldots,k_n\in\mathbb{N}$ and let $\mathcal{H}_m$ be as in Definition \ref{def1}. For $s\in[n]$ and $t\in[k_s-1]$, set $i_{k_s}^{(s;-m-1)}:=i_1^{(s;m+1)}$ and $i_t^{(s;-m-1)}:= i_{t+1}^{(s;m+1)}$. Then, we have that
\begin{equation} \label{eq2.1}
m_{k_1,\ldots,k_n}^{(\mathcal{H}_m)}=\sum_{i_b^{(a;\pm c)}=1}^{N_{c-1}}\mean{\prod_{s=1}^n\prod_{t=1}^{k_s}H_{i_t^{(s;1)}i_t^{(s;-1)}}}\prod_{u=1}^m\mean{\prod_{s=1}^n\prod_{t=1}^{k_s}(X^\dagger_u)_{i_t^{(s;u+1)}i_t^{(s;u)}}(X_u)_{i_t^{(s;-u)}i_t^{(s;-u-1)}}},
\end{equation}
where the sum is taken over $c\in[m+1]$, $a\in[n]$, and $b\in[k_a]$.
\end{lemma}
\begin{proof}
Taking equation \eqref{eq1.5} as the definition of $m_{k_1,\ldots,k_n}^{(\mathcal{H}_m)}$ and inserting the expansion
\begin{equation*}
\Tr \mathcal{H}_m^{k_s}=\sum_{i_1^{(s)},\ldots,i_{k_s}^{(s)}=1}^N(\mathcal{H}_m)_{i_1^{(s)}i_2^{(s)}}(\mathcal{H}_m)_{i_2^{(s)}i_3^{(s)}}\cdots(\mathcal{H}_m)_{i_{k_s}^{(s)}i_1^{(s)}}
\end{equation*}
shows that
\begin{equation*}
m_{k_1,\ldots,k_n}^{(\mathcal{H}_m)}=\mean{\prod_{s=1}^n\Tr \mathcal{H}_m^{k_s}}=\left(\prod_{a=1}^n\prod_{b=1}^{k_a}\sum_{i_b^{(a)}=1}^N\right)\mean{\prod_{s=1}^n\prod_{t=1}^{k_s}(\mathcal{H}_m)_{i_t^{(s)}i_{t+1}^{(s)}}},
\end{equation*}
where we have set $i_{k_s+1}^{(s)}:=i_1^{(s)}$ for all $s\in[n]$. Making the replacement $(\mathcal{H}_m)_{i_t^{(s)}i_{t+1}^{(s)}}\mapsto (\mathcal{H}_m)_{i_t^{(s;m+1)}i_{t+1}^{(s;m+1)}}=(\mathcal{H}_m)_{i_t^{(s;m+1)}i_t^{(s;-m-1)}}$ and using the definition of $\mathcal{H}_m$ to write
\begin{multline} \label{eq2.2}
(\mathcal{H}_m)_{i_t^{(s;m+1)}i_t^{(s;-m-1)}}=(X_m^{\dagger})_{i_t^{(s;m+1)}i_t^{(s;m)}}(X_{m-1}^{\dagger})_{i_t^{(s;m)}i_t^{(s;m-1)}}\cdots
\\\cdots(X_1^{\dagger})_{i_t^{(s;2)}i_t^{(s;1)}}H_{i_t^{(s;1)}i_t^{(s;-1)}}(X_1)_{i_t^{(s;-1)}i_t^{(s;-2)}}\cdots(X_m)_{i_t^{(s;-m)}i_t^{(s;-m-1)}}
\end{multline}
then shows that
\begin{align*}
m_{k_1,\ldots,k_n}^{(\mathcal{H}_m)}&=\sum_{i_b^{(a;\pm c)}=1}^{N_{c-1}}\left\langle\prod_{s=1}^n\prod_{t=1}^{k_s}(X_m^{\dagger})_{i_t^{(s;m+1)}i_t^{(s;m)}}(X_{m-1}^{\dagger})_{i_t^{(s;m)}i_t^{(s;m-1)}}\cdots\right.
\\&\hspace{10em}\left.\cdots(X_1^{\dagger})_{i_t^{(s;2)}i_t^{(s;1)}}H_{i_t^{(s;1)}i_t^{(s;-1)}}(X_1)_{i_t^{(s;-1)}i_t^{(s;-2)}}\cdots(X_m)_{i_t^{(s;-m)}i_t^{(s;-m-1)}}\right\rangle .
\end{align*}
As the random matrices $H,X_1,\ldots,X_m$ are assumed to be (mutually) independent, the average in the above can be split into the product of averages shown on the right-hand side of equation \eqref{eq2.1}, thereby proving the latter formula.
\end{proof}

Each average on the right-hand side of equation \eqref{eq2.1} contains a product of centred normal variables, so we may apply the Isserlis--Wick theorem to deconstruct them into sums of products of covariances when $k_1+\cdots+k_n$ is even (otherwise, the average over $H$ vanishes). Such covariances are well known to represent ribbons of orientable ribbon graphs (see, e.g., \citep{LZ04}, \citep[Sec.~3.3]{Rah22}, \citep{RGM26}), with the summation indices dictating how the ribbon edges are coloured and how the ribbons must be glued together. In this way, one is simply required to appropriately glue together LUE and GUE ribbon graphs.

To visualise each summand of equation \eqref{eq2.1} as a collection of ribbon graphs, we first draw a $k_s(2m+1)$-gon for each $s\in[n]$ to represent $\Tr\mathcal{H}_m^{k_s}$. The edges of these polygons represent entries of $X_m^\dagger,\ldots,X_1^\dagger,H,X_1,\ldots,X_m$ in that cyclic order and the vertices are labelled by the summation indices $i_t^{(s;\pm u)}$ in accordance with equation \eqref{eq2.2}. Then, the Isserlis--Wick theorem says that the product of averages on the right-hand side of equation \eqref{eq2.1} is given by a sum of a product of covariances, where the sum is taken over all possible ways of joining together pairs of polygon edges with ribbons such that edges representing entries of $X_i$ are identified with those representing $X_i^\dagger$ and edges representing entries of $H$ are paired among themselves. This latter constraint can be satisfied by assigning, for each $u\in[m+1]$, a unique colour $\phi_u$ to all polygon vertices labelled by indices of the form $i_t^{(s;\pm u)}$, with $s\in[n]$ and $t\in[k_s]$, and then only allowing ribbon graphs whose boundaries have well-defined colours inherited from the vertices they pass through; see Figure \ref{fig1} below for an example.

\begin{definition} \label{def4}
Fix $m,n,k_1,\ldots,k_n\in\mathbb{N}$ such that $\overline{k}_n:=k_1+\cdots+k_n$ is even and let $\phi_1,\ldots,\phi_{m+1}$ be unique colours. Define $\mathcal{P}_{k_1,\ldots,k_n}^{(m)}$ to be the set of $(m+1)$-coloured orientable ribbon graphs that can be constructed as follows:
\begin{enumerate}
\item For each $s\in[n]$, draw an oriented $k_s(2m+1)$-gon with a marked vertex (to be labelled $i_1^{(s;m+1)}$).
\item Colour the vertices of each polygon by $\phi_{m+1},\phi_m,\ldots,\phi_1,\phi_1,\ldots,\phi_m$ repeated in cyclic order, starting at the marked vertex.
\item Glue untwisted ribbons (i.e., ribbons that do not identify edges representing entries of $X_u$ to other such edges, and likewise for entries of $X_u^\dagger$) to pairs of polygon edges such that the boundaries of the resulting ribbon graphs inherit well-defined colours from the vertices they pass through.
\end{enumerate}
Let $\mathcal{P}_{k_1,\ldots,k_n}^{(m),c}\subseteq\mathcal{P}_{k_1,\ldots,k_n}^{(m)}$ denote the subset of connected ribbon graphs.
\end{definition}

\begin{figure}[H]
        \centering
\captionsetup{width=.9\linewidth}
        \includegraphics[width=0.8\textwidth]{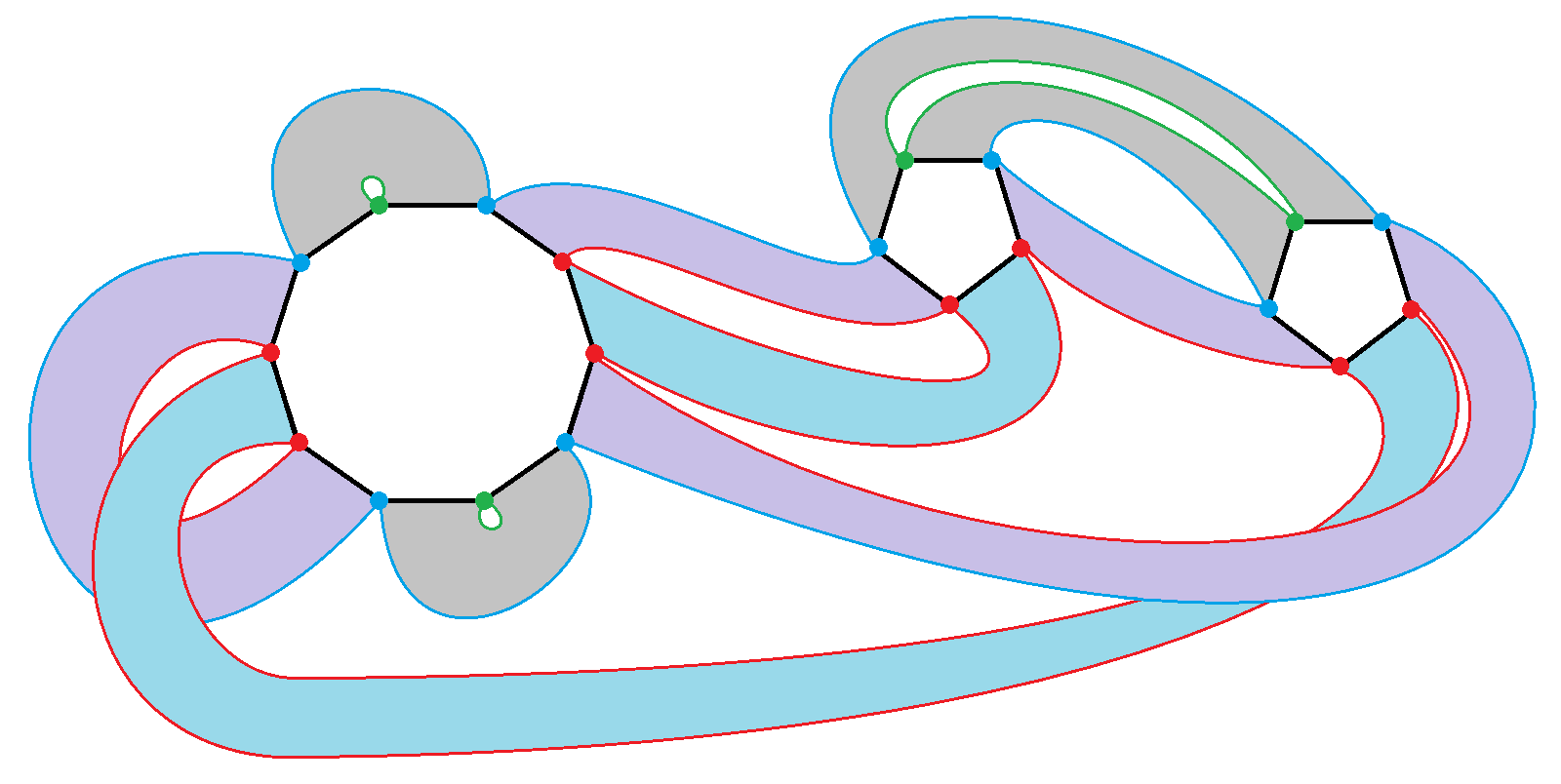}
        \caption{In computing $m_{2,1,1}^{(\mathcal{H}_2)}=\langle\Tr(\mathcal{H}_2^2)\Tr(\mathcal{H}_2)^2\rangle$, one first draws a decagon, representing $\Tr(\mathcal{H}_2^2)$, and two pentagons, each representing a factor of $\Tr(\mathcal{H}_2)$, with their vertices coloured as shown (the colours $\phi_1,\phi_2,\phi_3$ are respectively red, blue, and green). Pictured is a ribbon graph representing a pairwise identification of polygon edges using untwisted ribbons whose sides inherit well-defined colours from the vertices. The interior colours are purely for visual clarity and correspond to covariances involving $H$ (blue), $X_1,X_1^\dagger$ (purple), and $X_2,X_2^\dagger$ (grey).} \label{fig1}
\end{figure}

As the entries of $H,X_1,\ldots,X_m$ are independent (up to symmetry in the case of $H$), the covariances represented by the ribbons described above are given by
\begin{align}
\left\langle H_{ij}H_{kl}\right\rangle &= \frac{1}{N_0}\chi_{k=j}\chi_{l=i},
\\\left\langle (X_u^\dagger)_{ij}(X_u)_{kl}\right\rangle&=\frac{1}{\sqrt{N_uN_{u-1}}}\chi_{k=j}\chi_{l=i}, \label{eq2.4}
\end{align}
where $u\in[m]$ and we recall that $\chi_{k=j},\chi_{l=i}$ are indicator functions. Thus, the entire product of covariances corresponding to a ribbon graph either equals $(N_0N_1\cdots N_{m-1}\sqrt{N_m})^{-\overline{k}_n}$ or vanishes, depending on how the indices in the sum \eqref{eq2.1} are constrained. It turns out that said sum is non-vanishing only when summation indices are identified according to how boundaries of the ribbon graph prescribe identifications of polygon vertices. This prescription is the key benefit of the ribbon graph formulation.

\begin{lemma} \label{Lemma2}
Fix $m,n,k_1,\ldots,k_n\in\mathbb{N}$. If $\overline{k}_n=k_1+\cdots+k_n$ is odd,
\begin{equation*}
m_{k_1,\ldots,k_n}^{(\mathcal{H}_m)}=c_{k_1,\ldots,k_n}^{(\mathcal{H}_m)}=0,
\end{equation*}
else
\begin{align}
m_{k_1,\ldots,k_n}^{(\mathcal{H}_m)}&=(N_0N_1\cdots N_{m-1}\sqrt{N_m})^{-\overline{k}_n}\sum_{\Gamma\in\mathcal{P}^{(m)}_{k_1,\ldots,k_n}}N_0^{V_0(\Gamma)}N_1^{V_1(\Gamma)}\cdots N_m^{V_m(\Gamma)}, \label{eq2.5}
\\ c_{k_1,\ldots,k_n}^{(\mathcal{H}_m)}&=(N_0N_1\cdots N_{m-1}\sqrt{N_m})^{-\overline{k}_n}\sum_{\Gamma\in\mathcal{P}^{(m),c}_{k_1,\ldots,k_n}}N_0^{V_0(\Gamma)}N_1^{V_1(\Gamma)}\cdots N_m^{V_m(\Gamma)}, \label{eq2.6}
\end{align}
where, for $0\leq u\leq m$, $V_u(\Gamma)$ is the number of boundaries of $\Gamma$ that are of the colour $\phi_{u+1}$.
\end{lemma}
\begin{proof}
We first consider the mixed moments. Applying the Isserlis--Wick theorem to the averages in equation \eqref{eq2.1} shows that the first average over $H$ vanishes when $\overline{k}_n$ is odd. For even $\overline{k}_n$, interchanging the order of summation shows that
\begin{equation*}
m_{k_1,\ldots,k_n}^{(\mathcal{H}_m)}=\sum_{\Gamma\in\mathcal{P}_{k_1,\ldots,k_n}^{(m)}}\sum_{i_b^{(a;\pm c)}=1}^{N_{c-1}}(N_0N_1\cdots N_{m-1}\sqrt{N_m})^{-\overline{k}_n}\chi_{\Gamma},
\end{equation*}
where $\chi_{\Gamma}$ is the indicator function that equals one when, upon labelling the polygon vertices of $\Gamma$ by the indices $i_b^{(a;\pm c)}$ according to equation \eqref{eq2.2}, said indices are identified when the corresponding vertices belong to the same boundary of $\Gamma$. After said identification, we have for each $u\in[m+1]$ that the number of free, distinct indices of the form $i_b^{(a;\pm u)}$ is equal to $V_{u-1}(\Gamma)$. Equation \eqref{eq2.5} follows upon observing that
\begin{equation*}
\sum_{i_b^{(a;\pm1)}=1}^{N_0}\cdots\sum_{i_b^{(a;\pm m+1)}=1}^{N_m}\chi_{\Gamma}=N_0^{V_0(\Gamma)}N_1^{V_1(\Gamma)}\cdots N_m^{V_m(\Gamma)}.
\end{equation*}

To obtain the equivalent result for the mixed cumulants, we appeal to the well known fact that when mixed moments $m_{k_1,\ldots,k_n}$ count graphs or graph-like structures on $n$ nodes, the moment-cumulants relation requires that the corresponding mixed cumulants count the subset of graphs that are connected \citep{Smi95}, \citep[Ch.~5]{Sta99}, \citep{WF15}, \citep{Rah22}.
\end{proof}

In Definition \ref{def1}, we require that $N_1,\ldots,N_m$ differ from $N_0$ by fixed parameters $\nu_1,\ldots,\nu_m$. If we further assume that each $\nu_i$ is a constant multiple of $N_0$, we are able to reformulate the above lemma as an elegant genus expansion.
\begin{proposition} \label{prop1}
Let $m,N_0,\ldots,N_m\in\mathbb{N}$ be as in Definition \ref{def1} and suppose further that for each $i\in[m]$, there exists a non-negative constant $\hat{\nu}_i:=N_i/N_0$. Then, for $n,k_1,\ldots,k_n\in\mathbb{N}$ such that $\overline{k}_n=k_1+\cdots+k_n$ is even,
\begin{multline} \label{eq2.7}
c_{k_1,\ldots,k_n}^{(\mathcal{H}_m)}=(\hat{\nu}_1\cdots\hat{\nu}_{m-1}\sqrt{\hat{\nu}_m})^{-\overline{k}_n}N_0^{2-n}\sum_{g=0}^{\tfrac{1}{2}[(m+1/2)\overline{k}_n+1-n-m]}\frac{1}{N_0^{2g}}\sum_{p_1,\ldots,p_m=1}^{\overline{k}_n}\hat{\nu}_1^{p_1}\cdots\hat{\nu}_m^{p_m}
\\\times\#\{\Gamma\in\mathcal{P}_{k_1,\ldots,k_n}^{(m),c}\,:\,g(\Gamma)=g\textrm{ and }V_i(\Gamma)=p_i\textrm{ for each }i\in[m]\},
\end{multline}
where the set $\mathcal{P}^{(m),c}_{k_1,\ldots,k_n}$ and the functions $V_1(\Gamma),\ldots,V_m(\Gamma)$ are as in Lemma \ref{Lemma2} above, $g(\Gamma)$ is the genus of the ribbon graph $\Gamma$, and $\# S$ denotes the cardinality of the set $S$.
\end{proposition}
\begin{proof}
Inserting $N_i=\hat{\nu}_iN_0$ in equation \eqref{eq2.6} shows that
\begin{equation*}
c_{k_1,\ldots,k_n}^{(\mathcal{H}_m)}=(\hat{\nu}_1\cdots \hat{\nu}_{m-1}\sqrt{\hat{\nu}_m})^{-\overline{k}_n}N_0^{-E(\Gamma)}\sum_{\Gamma\in\mathcal{P}^{(m),c}_{k_1,\ldots,k_n}}N_0^{V(\Gamma)}\hat{\nu}_1^{V_1(\Gamma)}\cdots \hat{\nu}_m^{V_m(\Gamma)},
\end{equation*}
where $E(\Gamma)=(m+1/2)\overline{k}_n$ is the total number of ribbons of a given ribbon graph $\Gamma$ and $V(\Gamma)=\sum_{i=0}^m V_i(\Gamma)$ is the total number of boundaries. Grading $\mathcal{P}_{k_1,\ldots,k_n}^{(m),c}$ by genus and number of boundaries of given colours then allows us to replace the sum over $\mathcal{P}_{k_1,\ldots,k_n}^{(m),c}$ by the displayed sum over $g,p_1,\ldots,p_m$. To see the connection with genus, note that each ribbon graph is homeomorphic to some surface obtained by gluing together the polygons mentioned in Definition \ref{def4}. After identification, there are $n$ distinct polygon faces, $E(\Gamma)$ distinct polygon edges, and $V(\Gamma)$ distinct polygon vertices. Thus, by Euler's formula for the Euler characteristic, we have that
\begin{equation*}
2-2g(\Gamma)=n-E(\Gamma)+V(\Gamma)\implies N_0^{V(\Gamma)-E(\Gamma)}=N_0^{2-n-2g(\Gamma)}.
\end{equation*}
This explains the powers of $N_0$ seen in equation \eqref{eq2.7}. The upper terminals of the sums therein follow from consideration of the extremal cases where either no vertices of a chosen colour $\phi_{i+1}$ are identified (so $p_i=\overline{k}_n$) or all vertices of each colour are identified (so $V(\Gamma)=m+1$).
\end{proof}

\subsection{Constellations for Hermitised matrix product ensembles} \label{s2.2}
In this subsection, we give an algebraic formulation of Proposition \ref{prop1}. This is done through the theory of constellations \citep{LZ04}, which are tuples of permutations satisfying certain conditions. We encode the ribbon graphs of the previous subsection according to the following prescription:
\begin{enumerate}
\item Recall that the ribbon graphs drawn from $\mathcal{P}_{k_1,\ldots,k_n}^{(m),c}$ are built from $n$ oriented polygons, each having a marked vertex. Starting at the marked vertex of the first polygon, traverse the boundary of the polygon in the direction of its orientation and assign the label $1$ to the first $(2m+1)$ encountered edges, $2$ to the next $(2m+1)$ edges, and so on. Then, repeat this exercise for the remaining polygons, but with the first $(2m+1)$ edges of the $s$-th polygon being assigned the label $\overline{k}_{s-1}+1$; we write $\overline{k}_t:=k_1+\cdots+k_t$. In the notation of Lemma \ref{Lemma1}, the polygon edges representing entries of the form $H_{i_t^{(s;1)}i_t^{(s;-1)}},(X^\dagger_u)_{i_t^{(s;u+1)}i_t^{(s;u)}},(X_u)_{i_t^{(s;-u)}i_t^{(s;-u-1)}}$ are labelled $\overline{k}_{s-1}+t$.
\item Define $\gamma:=(1,\ldots,k_1)(k_1+1,\ldots,\overline{k}_2)\cdots(\overline{k}_{n-1}+1,\ldots,\overline{k}_n)\in\mathcal{S}_{\overline{k}_n}$ to be a permutation on $[\overline{k}_n]$ that describes traversal around the polygons at hand.
\item We say a polygon edge is of type $M$ if it represents an entry of the matrix $M$. Define $\sigma_0\in\mathcal{S}_{\overline{k}_n}$ to be a fixed-point free involution such that $\sigma_0(a)=b$ if there is a ribbon connecting an edge of type $H$ and label $a$ to one of type $H$ and label $b$.
\item For $i\in[m]$, define $\sigma_i\in\mathcal{S}_{\overline{k}_n}$ to be a permutation such that $\sigma_i(a)=b$ if there is a ribbon connecting an edge of type $X_i$ and label $a$ to one of type $X_i^\dagger$ and label $b$. Note that $\sigma_1,\ldots,\sigma_m$ are not necessarily involutions, nor fixed-point free.
\end{enumerate}
The permutations $\sigma_0,\ldots,\sigma_m$ defined above are enough to uniquely specify each ribbon graph in $\mathcal{P}_{k_1,\ldots,k_n}^{(m),c}$. Moreover, it is straightforward to check that, so long as $\sigma_0$ is a fixed-point free involution and the group $\mean{\gamma,\sigma_0,\ldots,\sigma_m}$ is transitive (transitivity of this group ensures that the associated ribbon graph is connected), any choice of $\sigma_0,\ldots,\sigma_m\in\mathcal{S}_{\overline{k}_n}$ describes a ribbon graph in $\mathcal{P}_{k_1,\ldots,k_n}^{(m),c}$. Thus, defining $\mathcal{S}_{k_1,\ldots,k_n}^{(m),c}$ to be the set of tuples $(\sigma_0,\ldots,\sigma_m)\in\mathcal{S}_{\overline{k}_n}^{m+1}$ such that $\sigma_0$ is a fixed-point free involution and $\mean{\gamma,\sigma_0,\ldots,\sigma_m}$ is transitive, we have that
\begin{equation*}
\mathcal{P}_{k_1,\ldots,k_n}^{(m),c}\simeq \mathcal{S}_{k_1,\ldots,k_n}^{(m),c}.
\end{equation*}

For our reformulation of Proposition \ref{prop1}, it remains to determine $V_0(\Gamma),\ldots,V_m(\Gamma)$ for the ribbon graph $\Gamma\in\mathcal{P}_{k_1,\ldots,k_n}^{(m),c}$ corresponding to a given tuple $(\sigma_0,\ldots,\sigma_m)\in\mathcal{S}_{k_1,\ldots,k_n}^{(m),c}$. Fixing such a tuple and its corresponding ribbon graph $\Gamma$, consider first the set $v_0(\Gamma)$ of vertices of colour $\phi_1$ that sit between polygon edges of type $H$ and $X_1$. The ribbon graph boundaries of colour $\phi_1$ each contain at least one element of $v_0(\Gamma)$ and, moreover, said boundaries act as equivalence classes on $v_0(\Gamma)$. These equivalence classes are seen to be the cycles of $\sigma_1^{-1}\sigma_0$ upon labelling the elements of $v_0(\Gamma)$ by the labels of the polygon edges that they are incident to. Thus, $V_0(\Gamma)=\#\sigma_1^{-1}\sigma_0$, where $\#\tau$ denotes the number of cycles of a permutation $\tau$. Likewise, defining $v_i(\Gamma)$, for $i\in[m-1]$, to be the set of vertices of colour $\phi_{i+1}$ sitting between polygon edges of type $X_i$ and $X_{i+1}$, the above reasoning shows that the boundaries of colour $\phi_{i+1}$ define equivalence classes on $v_i(\Gamma)$ and we consequently have $V_i(\Gamma)=\#\sigma_{i+1}^{-1}\sigma_i$. Finally, the cycles of $\gamma^{-1}\sigma_m$ describe tours of the boundaries of colour $\phi_{m+1}$ upon assigning each vertex of colour $\phi_{m+1}$ the label of the edge of type $X$ that it is incident to. Hence, $V_m(\Gamma)=\#\gamma^{-1}\sigma_m$. Combining all of this together yields the sought algebraic reformulation of Proposition \ref{prop1}.
\begin{proposition} \label{prop2}
Let $m,N_0,\ldots,N_m\in\mathbb{N}$ be as in Definition \ref{def1} and suppose that for each $i\in[m]$, there exists a non-negative constant $\hat{\nu}_i=N_i/N_0$. Then, for $n,k_1,\ldots,k_n\in\mathbb{N}$ such that $\overline{k}_n$ is even,
\begin{multline*}
c_{k_1,\ldots,k_n}^{(\mathcal{H}_m)}=(\hat{\nu}_1\cdots\hat{\nu}_{m-1}\sqrt{\hat{\nu}_m})^{-\overline{k}_n}N_0^{2-n}
\\\times\sum_{g=0}^{\tfrac{1}{2}[(m+1/2)\overline{k}_n+1-n-m]}\frac{1}{N_0^{2g}}\sum_{\substack{(\sigma_0,\ldots,\sigma_m)\\ \in\mathcal{S}_{k_1,\ldots,k_n}^{(m),c}(g)}}\hat{\nu}_1^{\#\sigma_2^{-1}\sigma_1}\cdots\hat{\nu}_{m-1}^{\#\sigma_m^{-1}\sigma_{m-1}}\hat{\nu}_m^{\#\gamma^{-1}\sigma_m},
\end{multline*}
where, for $\mathcal{S}_{k_1,\ldots,k_n}^{(m),c}$ as defined above, we further define
\begin{equation*}
\mathcal{S}_{k_1,\ldots,k_n}^{(m),c}(g):=\Big\{(\sigma_0,\ldots,\sigma_m)\in\mathcal{S}_{k_1,\ldots,k_n}^{(m),c}\,:\,\#\gamma^{-1}\sigma_m+\sum_{i=0}^{m-1}\#\sigma_{i+1}^{-1}\sigma_i=(m+\tfrac{1}{2})\overline{k}_n+2-n-2g\Big\}
\end{equation*}
to be the subset of tuples in $\mathcal{S}_{k_1,\ldots,k_n}^{(m),c}$ that are of genus $g$. Here, the genus $g$ of a tuple $(\sigma_0,\ldots,\sigma_m)$ is defined by the condition in the definition of $\mathcal{S}_{k_1,\ldots,k_n}^{(m),c}(g)$; this condition descends canonically from the definition of the genus of a ribbon graph.
\end{proposition}

To conclude, let us elucidate the connection to the constellations of \citep{LZ04}. Defining $\mathcal{C}_{k_1,\ldots,k_n}^{(m),c}$ to be the set of tuples $(\gamma,\alpha_m,\ldots,\alpha_0)\in\mathcal{S}_{\overline{k}_n}^{m+2}$ such that $\alpha_0$ is a fixed-point free involution and $\gamma=(1,\ldots,k_1)(k_1+1,\ldots,\overline{k}_2)\cdots(\overline{k}_{n-1}+1,\ldots,\overline{k}_n)\in\mathcal{S}_{\overline{k}_n}$, we see that the mappings $\alpha_0=\sigma_0$, $\alpha_i=\sigma_i^{-1}\sigma_{i-1}$ for $i\in[m]$, and $\alpha_m=\gamma^{-1}\sigma_m$ define a bijection $\mathcal{S}_{k_1,\ldots,k_n}^{(m),c}\simeq \mathcal{C}_{k_1,\ldots,k_n}^{(m),c}$. Moreover, we have by the conditions on $\sigma_0,\ldots,\sigma_m$ that $\gamma\alpha_m\cdots\alpha_0=\mathrm{id}$ and $\mean{\gamma,\alpha_m,\ldots,\alpha_0}$ is transitive. These are precisely the properties defining constellations, so $\mathcal{C}_{k_1,\ldots,k_n}^{(m),c}$ is the set of $(m+2)$-constellations $(\gamma,\alpha_m,\ldots,\alpha_0)$ with $\gamma,\alpha_0$ obeying the constraints given above. 

\section{Genus expansions for antisymmetrised matrix product ensembles} \label{s3}
We now move on to the mixed moments and cumulants of the antisymmetrised matrix product ensembles of Definition \ref{def2}. The key differences here are that the matrices $X_1,\ldots,X_m$ are now real and of even dimension, and that extra symmetries are now enforced due to the (sparse) structure of $J$. As in Section \ref{s2}, we begin by developing ribbon graph characterisations of the mixed moments and cumulants. Then, in \S\ref{s3.2}, we give a reformulation of our results in terms of constellations.

\setcounter{equation}{0}
\subsection{Ribbon graphs for antisymmetrised matrix product ensembles} \label{s3.1}
We take $N_0,\ldots,N_m$ to be positive even integers with $N_1,\ldots,N_m\ge N_0$. Then, with $J,X_1,\ldots,X_m$ as in Definition~\ref{def2}, we study the mixed moments and cumulants of $\mathcal{J}_m$. In what follows, the only properties of $J$ that will be used is that $J^T=-J$ and $J^2=-I_{N_0}$. We begin with an analogue of Lemma \ref{Lemma1}.

\begin{lemma} \label{Lemma3}
Fix $m,n,k_1,\ldots,k_n\in\mathbb{N}$. The mixed moment $m_{k_1,\ldots,k_n}^{(\mathcal{J}_m)}$ vanishes if any of the $k_1,\ldots,k_n$ are odd. Otherwise, we have that
\begin{equation} \label{eq3.1}
m_{k_1,\ldots,k_n}^{(\mathcal{J}_m)}=\sum_{i_b^{(a;\pm c)}=1}^{N_{c-1}}\mean{\prod_{s=1}^n\prod_{t=1}^{k_s}J_{i_t^{(s;1)}i_t^{(s;-1)}}}\prod_{u=1}^m\mean{\prod_{s=1}^n\prod_{t=1}^{k_s}(X^T_u)_{i_t^{(s;u+1)}i_t^{(s;u)}}(X_u)_{i_t^{(s;-u)}i_t^{(s;-u-1)}}},
\end{equation}
where the sum is again taken over $c\in[m+1]$, $a\in[n]$, and $b\in[k_a]$ and we retain the identifications $i_{k_s}^{(s;-m-1)}=i_1^{(s;m+1)}$ and $i_t^{(s;-m-1)}=i_{t+1}^{(s;m+1)}$ for $s\in[n]$ and $t\in[k_s-1]$.
\end{lemma}
\begin{proof}
If $k_s$ is odd for some $s\in[n]$, then $\Tr\mathcal{J}_m^{k_s}$ vanishes due to $\mathcal{J}_m$ being antisymmetric. For $k_1,\ldots,k_n$ even, the proof is as for Lemma \ref{Lemma1}.
\end{proof}

Ribbon graph representations corresponding to the summand of equation \eqref{eq3.1} can be constructed in much the same way as in the Hermitised case, bar a few caveats: First, the matrices $X_u$ ($u\in[m]$) are now real Ginibre, so $(X_u^T)_{ij}=(X_u)_{ji}$ and, consequently, one must now take into account covariances of the form
\begin{equation} \label{eq3.2}
\mean{(X_u^T)_{ij}(X_u^T)_{lk}}=\mean{(X_u)_{ji}(X_u)_{kl}}=\frac{1}{2\sqrt{N_uN_{u-1}}}\chi_{i=l}\chi_{j=k}
\end{equation}
in addition to the replacement of \eqref{eq2.4},
\begin{equation*}
\mean{(X_u^T)_{ij}(X_u)_{kl}}=\frac{1}{2\sqrt{N_uN_{u-1}}}\chi_{k=j}\chi_{l=i};
\end{equation*}
covariances of the form \eqref{eq3.2} are to be represented by twisted ribbons and the resulting ribbon graphs may now be non-orientable. The second difference between the Hermitised and antisymmetrised case is that the matrices $H$ have been replaced by the deterministic matrix $J$. Thus, there is no need to consider covariances involving $J$ and one can even do away with edges representing entries of $J$, so long as care is taken to properly interpret the vertices replacing them. We therefore glue together $m$ ribbon graphs for the Laguerre orthogonal ensemble.

\begin{definition}
Fix $m,n,k_1,\ldots,k_n\in\mathbb{N}$ such that each of $k_1,\ldots,k_n$ are even and let $\phi_1,\ldots,\phi_{m+1}$ be unique colours. Define $\widetilde{\mathcal{P}}_{k_1,\ldots,k_n}^{(m)}$ to be the set of $(m+1)$-coloured locally orientable ribbon graphs that can be constructed as follows:
\begin{enumerate}
\item For each $s\in[n]$, draw an oriented $2mk_s$-gon with a marked vertex (to be labelled $i_1^{(s;m+1)}$).
\item Colour the vertices of each polygon by $\phi_{m+1},\phi_m,\ldots,\phi_2,\phi_1,\phi_2,\ldots,\phi_m$ repeated in cyclic order, starting at the marked vertex.
\item Glue ribbons to pairs of polygon edges such that the boundaries of the resulting ribbon graphs inherit well-defined colours from the vertices they pass through. A ribbon is said to be twisted if it represents the identification of polygon edges that have the same orientation.
\end{enumerate}
Let $\widetilde{\mathcal{P}}_{k_1,\ldots,k_n}^{(m),c}\subseteq\widetilde{\mathcal{P}}_{k_1,\ldots,k_n}^{(m)}$ denote the subset of connected ribbon graphs.
\end{definition}

Application of the Isserlis--Wick theorem to equation \eqref{eq3.1} follows in mostly the same way as in the proof of Lemma \ref{Lemma2}: The mixed moments and cumulants are given by normalised sums over the sets defined above and for each ribbon graph $\Gamma$ of the sum, boundaries of colour $\phi_u$ for each $u=2,\ldots,m+1$ corresponding to identification of vertex labels of type $i_b^{(a;\pm u)}$ contribute a weight of $N_{u-1}$. However, boundaries of colour $\phi_1$ must be treated with care, because the vertices they pass through each represent two summation indices designating an entry of $J$. Thus, rather than weighting each boundary of colour $\phi_1$ by a factor of $N_0$, we must instead sum over a product of entries of $J$.

\begin{definition} \label{def6}
Fix $m,n,k_1,\ldots,k_n$ with $k_1,\ldots,k_n$ even and let $\Gamma\in\widetilde{\mathcal{P}}_{k_1,\ldots,k_n}^{(m)}$. Recall the notation $V_u(\Gamma)$ for the number of boundaries of colour $\phi_{u+1}$ ($0\le u\le m$) from Lemma \ref{Lemma2}. Let $\gamma_1,\ldots,\gamma_{V_0(\Gamma)}$ be the boundaries of $\Gamma$ that are of colour $\phi_1$.

For $i\in[V_0(\Gamma)]$, define the $J$-weight $\omega_J(\gamma_i)$ of a boundary $\gamma_i$ as the trace of a word in $\{J,J^T\}$ constructed as follows: Travel once around the boundary $\gamma_i$ and each time a vertex of colour $\phi_1$ belonging to $\gamma_i$ is passed through, record a factor of $J$ if the traversal is in the same direction as the orientation of the polygon that the vertex belongs to, else record a factor of $J^T$.
\end{definition}

We give an example of the computation of $J$-weights in Figure \ref{fig2} below. The procedure outlined above can be understood as a special case of the decompositions of traces of matrix products studied in free probability theory \citep[Lecture~22]{NS06}. It should be noted that due to the antisymmetric nature of $J$, $\omega_J(\gamma_i)=0$ if $\gamma_i$ passes through an odd number of vertices. Moreover, the computation of $\omega_J(\gamma_i)$ does not depend on where in $\gamma_i$ the walk along $\gamma_i$ starts, nor its direction. We may now give the analogue of Proposition \ref{prop1} for antisymmetrised matrix product ensembles.

\begin{figure}[H]
        \centering
\captionsetup{width=.9\linewidth}
        \includegraphics[width=0.8\textwidth]{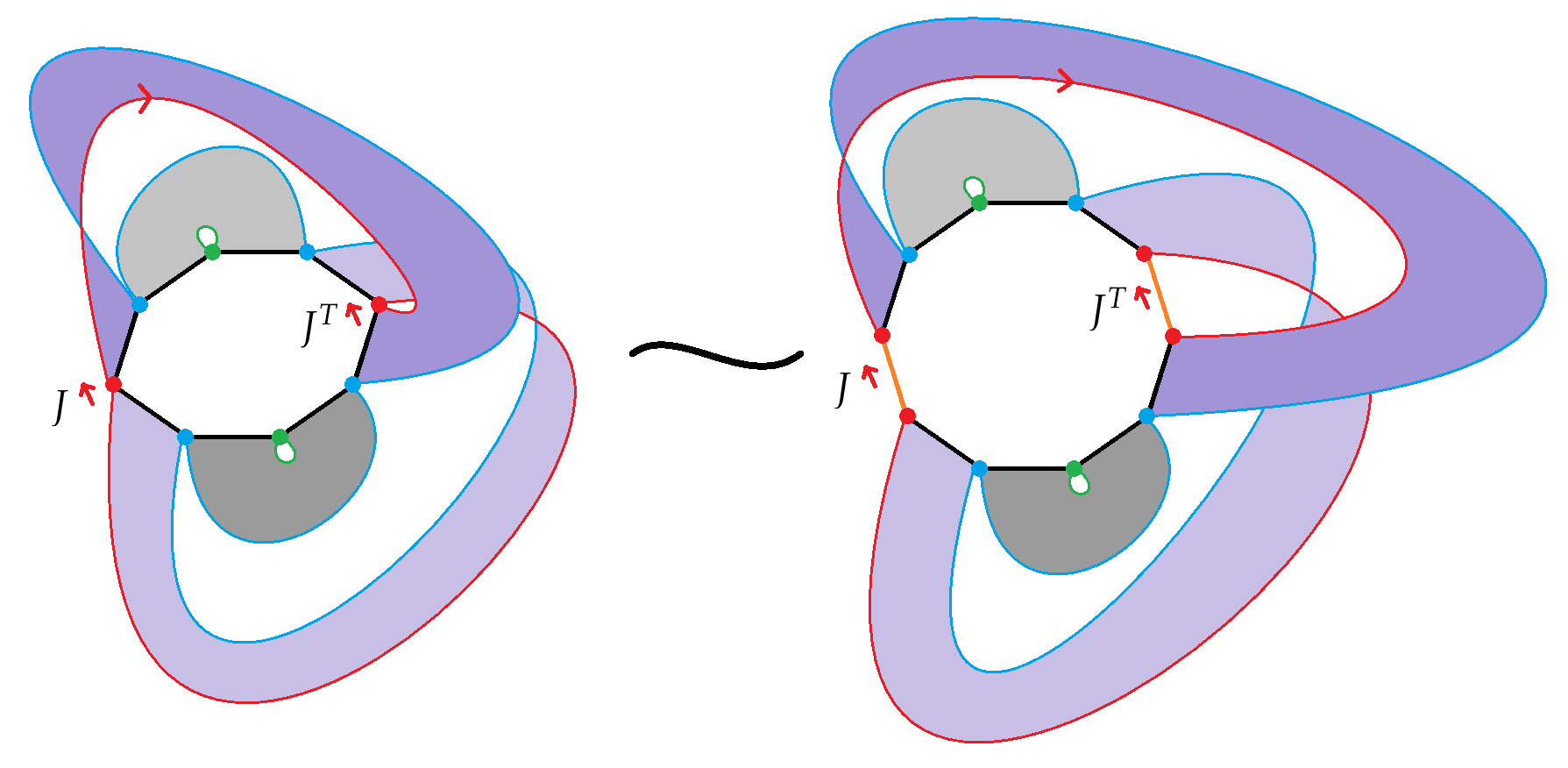}
        \caption{Illustrated on the left is a ribbon graph drawn from $\widetilde{\mathcal{P}}_2^{(2)}$ with colours $\phi_1,\phi_2,\phi_3$ being red, blue, and green, respectively. We orient the octagon clockwise. To calculate the $J$-weight of the single red boundary, we walk along said boundary starting at the top and moving to the right. The red vertex on the right is passed first, and it is passed through in an anti-clockwise direction, so it is marked as $J^T$. The second vertex is passed through in a clockwise direction, so it is marked as $J$. Thus, the $J$-weight of this boundary is $\Tr(J^TJ)$. On the right, we have expanded the red vertices into orange edges representing entries of $J$ when travelling clockwise around the polygon; the edges represent entries of $J^T$ in the reverse direction.} \label{fig2}
\end{figure}

\begin{proposition} \label{prop3}
Let $m,N_0,\ldots,N_m\in\mathbb{N}$ be as in Definition \ref{def2} and suppose further that for each $i\in[m]$, there exists a non-negative constant $\hat{\nu}_i:=N_i/N_0$. Retain the definitions of $V_0(\Gamma),\ldots,V_m(\Gamma)$, the boundaries $\gamma_i$, and $J$-weights from Lemma \ref{Lemma2} and Definition \ref{def6}. Then, for $n,k_1,\ldots,k_n\in\mathbb{N}$ such that each of $k_1,\ldots,k_n$ are even,
\begin{equation*}
c_{k_1,\ldots,k_n}^{(\mathcal{J}_m)}=(\hat{\nu}_1\cdots\hat{\nu}_{m-1}\sqrt{\hat{\nu}_m})^{-\overline{k}_n}N_0^{2-n}\sum_{h=0}^{m\overline{k}_n+1-n-m}\frac{1}{N_0^h}\sum_{p_1,\ldots,p_m=1}^{\overline{k}_n}\hat{\nu}_1^{p_1}\cdots\hat{\nu}_m^{p_m}\sum_{\Gamma\in\widetilde{\mathcal{T}}_{\mathbf{k},\mathbf{p},h}^{(m),c}}\prod_{i=1}^{V_0(\Gamma)}\tilde\omega_J(\gamma_i),
\end{equation*}
where $\tilde\omega_J(\gamma_i)=\frac{1}{N_0}\omega_J(\gamma_i)\in\{-1,0,1\}$ and $\widetilde{\mathcal{T}}_{\mathbf{k},\mathbf{p},h}^{(m),c}\subseteq\widetilde{\mathcal{P}}_{k_1,\ldots,k_n}^{(m),c}$ is the subset of ribbon graphs with $V_i(\Gamma)=p_i$ for all $i\in[m]$ and Euler genus $h$ (i.e., $h$ is related to the Euler characteristic by $h=2-\chi=m\overline{k}_n-\sum_{i=0}^mV_i(\Gamma)+2-n$; it is double the usual genus for orientable ribbon graphs, but takes odd values for non-orientable ribbon graphs; even $h$ can still correspond to non-orientable ribbon graphs, such as those homeomorphic to the Klein bottle).
\end{proposition}
\begin{proof}
Proceeding as in the proof of Lemma \ref{Lemma2} while recalling that arguments concerning mixed moments apply to mixed cumulants upon restricting to connected ribbon graphs, we have that
\begin{multline*}
c_{k_1,\ldots,k_n}^{(\mathcal{J}_m)}=(N_1\cdots N_{m-1}\sqrt{N_0N_m})^{-\overline{k}_n}\sum_{\Gamma\in\widetilde{\mathcal{P}}_{k_1,\ldots,k_n}^{(m),c}}N_1^{V_1(\Gamma)}\cdots N_m^{V_m(\Gamma)}
\\\times\left(\prod_{a=1}^n\prod_{b=1}^{k_a}\sum_{i_b^{(a;\pm1)}=1}^{N_0}\right)\prod_{s=1}^n\prod_{t=1}^{k_s}J_{i_t^{(s;1)}i_t^{(s;-1)}}\chi_{\gamma_1}\cdots\chi_{\gamma_{V_0(\Gamma)}},
\end{multline*}
where for each $i\in[V_0(\Gamma)]$, $\chi_{\gamma_i}$ is the product of indicator functions on indices of the form $i_b^{(a;\pm1)}$ represented by the boundary $\gamma_i$. Now, each entry $J_{i_t^{(s;1)}i_t^{(s;-1)}}$ belongs to exactly one of $\gamma_1,\ldots,\gamma_{V_0(\Gamma)}$ and traversing each boundary as in Definition \ref{def6} prescribes an order on the entries $J_{i_t^{(s;1)}i_t^{(s;-1)}}$. Writing out a product of these entries in the order just described for each boundary $\gamma_i$, $\chi_{\gamma_i}$ ensures that the second index of each term matches the first index of the next, so long as care is taken to sometimes replace $J_{i_t^{(s;1)}i_t^{(s;-1)}}$ with $(J^T)_{i_t^{(s;-1)}i_t^{(s;1)}}$. In this way, the second line of the above formula simplifies to a product of traces so that
\begin{align*}
c_{k_1,\ldots,k_n}^{(\mathcal{J}_m)}&=(N_1\cdots N_{m-1}\sqrt{N_0N_m})^{-\overline{k}_n}\sum_{\Gamma\in\widetilde{\mathcal{P}}_{k_1,\ldots,k_n}^{(m),c}}N_1^{V_1(\Gamma)}\cdots N_m^{V_m(\Gamma)}\prod_{i=1}^{V_0(\Gamma)}\omega_J(\gamma_i)
\\&=(N_1\cdots N_{m-1}\sqrt{N_0N_m})^{-\overline{k}_n}\sum_{\Gamma\in\widetilde{\mathcal{P}}_{k_1,\ldots,k_n}^{(m),c}}N_0^{V_0(\Gamma)}\cdots N_m^{V_m(\Gamma)}\prod_{i=1}^{V_0(\Gamma)}\tilde\omega_J(\gamma_i).
\end{align*}
The remainder of this proof proceeds exactly as the proof of Proposition \ref{prop1}.
\end{proof}

\subsection{Constellations for antisymmetrised matrix product ensembles} \label{s3.2}
In this subsection, we repeat the exercise of \S\ref{s2.2}, now needing to address twisted ribbons and taking into account the $J$-weights of Definition \ref{def6}. We proceed, as before, by prescribing an algebraic encoding of the ribbon graphs of the previous subsection:
\begin{enumerate}
\item Consider the oriented polygons from which the ribbon graphs of $\widetilde{\mathcal{P}}_{k_1,\ldots,k_n}^{(m),c}$ are built. In the notation of Lemma \ref{Lemma3}, we label the polygon edges representing matrix entries of the form $(X_u)_{i_t^{(s;-u)}i_t^{(s;-u-1)}}$ by $\overline{k}_{s-1}+t$ and those representing entries of the form $(X^T_u)_{i_t^{(s;u+1)}i_t^{(s;u)}}$ by $-(\overline{k}_{s-1}+t)$. Thus, traversing the boundary of the $s$-th polygon in the direction of its orientation, starting at the marked vertex, one sequentially encounters clusters of $m$ edges sharing the label $-(\overline{k}_{s-1}+1), \overline{k}_{s-1}+1, -(\overline{k}_{s-1}+2),\overline{k}_{s-1}+2,\ldots,-\overline{k}_s,\overline{k}_s$.
\item Let $\mathcal{S}_{\pm\overline{k}_n}$ denote the set of permutations on $-[\overline{k}_n]\cup[\overline{k}_n]$ and define $e,f\in\mathcal{S}_{\pm\overline{k}_n}$ by
\begin{align*}
e&=(-1,1)(-2,2)\cdots(-\overline{k}_n,\overline{k}_n),
\\f&=\prod_{s=1}^n(-1,\overline{k}_s)\prod_{i=1}^{k_s-1}(i,-i-1).
\end{align*}
Note that $ef=\gamma(-k_1,\ldots,-1)(-\overline{k}_2,\ldots,-k_1-1)\cdots(-\overline{k}_n,\ldots,-\overline{k}_{n-1}-1)\in\mathcal{S}_{\pm\overline{k}_n}$, where $\gamma$ is as in \S\ref{s2.2}.
\item For $i\in[m]$, define $\sigma_i\in\mathcal{S}_{\pm\overline{k}_n}$ to be a fixed-point free involution such that $\sigma_i(a)=b$ if there is a ribbon connecting a polygon edge of type $X_i$ or $X_i^T$ of label $a$ to one of type $X_i$ or $X_i^T$ of label $b$. Note that a ribbon is (un-)twisted if $a$ and $b$ (do not) have the same sign.
\end{enumerate}
As in \S\ref{s2.2}, we define $\widetilde{\mathcal{S}}_{k_1,\ldots,k_n}^{(m),c}$ to be the set of tuples $(\sigma_1,\ldots,\sigma_m)\in\mathcal{S}_{\pm\overline{k}_n}^m$ such that $\sigma_1,\ldots,\sigma_m$ are fixed-point free involutions and $\mean{e,f,\sigma_1,\ldots,\sigma_m}$ is transitive. We then have that
\begin{equation*}
\widetilde{\mathcal{P}}_{k_1,\ldots,k_n}^{(m),c}\simeq\widetilde{\mathcal{S}}_{k_1,\ldots,k_n}^{(m),c}.
\end{equation*}
Applying the same arguments as in \S\ref{s2.2} shows that for $\Gamma\in\widetilde{\mathcal{P}}_{k_1,\ldots,k_n}^{(m),c}$ denoting the ribbon graph corresponding to $(\sigma_1,\ldots,\sigma_m)\in\widetilde{\mathcal{S}}_{k_1,\ldots,k_n}^{(m),c}$ and $i\in[m-1]$, the cycles of $\sigma_{i+1}\sigma_i$ list every second vertex visited when touring the boundaries of colour $\phi_{i+1}$. On the other hand, the cycles of $\sigma_1e$ and $f\sigma_m$ list the vertices of the boundaries of colour $\phi_1$ and $\phi_{m+1}$ respectively, but each boundary is described by two cycles due to vertices of colour $\phi_1$ and $\phi_{m+1}$ each inheriting two labels from their incident polygon edges. Hence, we have $V_0(\Gamma)=\tfrac{1}{2}\#\sigma_1e$, $V_m(\Gamma)=\tfrac{1}{2}\#f\sigma_m$, and for $i\in[m-1]$, $V_i(\Gamma)=\tfrac{1}{2}\#\sigma_{i+1}\sigma_i$. Note the extra factor of half not seen in the cycle counts of Proposition \ref{prop2}. This stems from the fact that, unlike in the Hermitised case, we must now put entries of $X_i,X_i^T$, hence vertices of the same colour, on equal footing, instead of bipartitioning them. At a deeper level, the factor of half appears due to the fact that the non-orientable ribbon graphs of $\widetilde{\mathcal{P}}_{k_1,\ldots,k_n}^{(m),c}$ have ambiguous choices of local orientations (these ambiguities vanish when considering their double covers; see, e.g., \citep{RGM26}).

In the antisymmetrised setting, the novelty lies in computing the $J$-weights of the boundaries corresponding to the cycles of $\sigma_1e$. Without these weights, the present formalism produces formulae for the mixed cumulants of the real Wishart product ensembles of Definition \ref{def3}. We find that the antisymmetric structure of $J$ leads to an interesting algebraic formulation of Proposition \ref{prop3}.

\begin{proposition} \label{prop4}
Let $m,N_0,\ldots,N_m\in\mathbb{N}$ be as in Definition \ref{def2} and suppose that for each $i\in[m]$, there exists a non-negative constant $\hat{\nu}_i=N_i/N_0$. Define $\widetilde{\mathcal{S}}_{k_1,\ldots,k_n}^{(m),c,\pm}(h)\subseteq\widetilde{\mathcal{S}}_{k_1,\ldots,k_n}^{(m),c}$ such that for $(\sigma_1,\ldots,\sigma_m)\in\widetilde{\mathcal{S}}_{k_1,\ldots,k_n}^{(m),c,*}(h)$,
\begin{enumerate}
\item the Euler genus is $h$, equivalently $\# f\sigma_m+\#\sigma_1e+\sum_{i=1}^{m-1}\#\sigma_{i+1}\sigma_i=2m\overline{k}_n+4-2n-2h$;
\item the non-trivial cycles of $\varphi$ are all of even length, where $\varphi$ is the product of cycles of $\sigma_1 e$ satisfying the property that the elements of minimal absolute value are of positive sign;
\item the total number of negative elements in the non-trivial cycles of $\varphi$ is even (odd) when $*=+$ (when $*=-$).
\end{enumerate}
Then, for $n,k_1,\ldots,k_n\in\mathbb{N}$ such that each of $k_1,\ldots,k_n$ are even,
\begin{multline*}
c_{k_1,\ldots,k_n}^{(\mathcal{J}_m)}=(-1)^{\overline{k}_n/2}(\hat{\nu}_1\cdots\hat{\nu}_{m-1}\sqrt{\hat{\nu}_m})^{-\overline{k}_n}N_0^{2-n}
\\\times\sum_{h=0}^{m\overline{k}_n+1-n-m}\frac{1}{N_0^h}\Big(\sum_{\substack{(\sigma_0,\ldots,\sigma_m)\\ \in\widetilde{\mathcal{S}}_{k_1,\ldots,k_n}^{(m),c,+}(h)}}-\sum_{\substack{(\sigma_0,\ldots,\sigma_m)\\ \in\widetilde{\mathcal{S}}_{k_1,\ldots,k_n}^{(m),c,-}(h)}}\Big)\hat{\nu}_1^{\tfrac{1}{2}\#\sigma_2\sigma_1}\cdots\hat{\nu}_{m-1}^{\tfrac{1}{2}\#\sigma_m\sigma_{m-1}}\hat{\nu}_m^{\tfrac{1}{2}\# f\sigma_m}.
\end{multline*}
\end{proposition}
\begin{proof}
Fix $(\sigma_1,\ldots,\sigma_m)\in\widetilde{\mathcal{S}}_{k_1,\ldots,k_n}^{(m),c}$ and let $\Gamma\in\widetilde{\mathcal{P}}_{k_1,\ldots,k_n}^{(m),c}$ be its corresponding ribbon graph. The Euler genus of Proposition \ref{prop3} relates to the formula given in the first condition since we have from above that $\sum_{i=0}^mV_i(\Gamma)=\tfrac{1}{2}\# f\sigma_m+\#\sigma_1e+\tfrac{1}{2}\sum_{i=1}^{m-1}\#\sigma_{i+1}\sigma_i$.

Now, recall that the boundaries of $\Gamma$ of colour $\phi_1$ are in two-to-one correspondence with the cycles of $\sigma_1e$. These cycles occur in pairs, meaning that if $\tau$ is a cycle of $\sigma_1e$, then so too is $e\tau^{-1}e$, which is the inverse of $\tau$ with all elements negated. We define $\varphi$ such that it contains one cycle of every such pair, hence the cycles of $\varphi$ are in bijection with the boundaries of colour $\phi_1$. Then, in the notation of Definition \ref{def6}, and letting $\tau_i$ denote the cycle of $\varphi$ encoding the boundary $\gamma_i$ of $\Gamma$, we observe that $\omega_J(\gamma_i)$ is the trace of a word in $\{J,J^T\}$ such that
\begin{enumerate}
\item $\omega_J(\gamma_i)\ne0$ only when the word is of even length, which is equivalent to the second condition of the present proposition;
\item the number of occurrences of $J^T$ in the word is of the same parity as $r(\tau_i)$, which we define to be the number of negative elements in $\tau_i$.
\end{enumerate}
Thus, we may discard terms of weight zero in Proposition \ref{prop3} by restricting our sum over $p_1,\ldots,p_m$ and $\widetilde{\mathcal{T}}_{\mathbf{k},\mathbf{p},h}^{(m),c}$ to a sum over $\widetilde{\mathcal{S}}_{k_1,\ldots,k_n}^{(m),c,+}(h)\cup\widetilde{\mathcal{S}}_{k_1,\ldots,k_n}^{(m),c,-}(h)$. Moreover, we see that
\begin{equation*}
\omega_J(\gamma_i)=(-1)^{r(\tau_i)}\Tr J^{|\tau_i|}=(-1)^{r(\tau_i)+|\tau_i|/2}N_0,
\end{equation*}
where $|\tau_i|$ is the length of $\tau_i$. Then, since $\sum_{i=1}^{V_0(\Gamma)}|\tau_i|=\overline{k}_n$, we have that
\begin{equation*}
\prod_{i=1}^{V_0(\Gamma)}\tilde\omega_J(\gamma_i)=(-1)^{\overline{k}_n/2+r(\varphi)},
\end{equation*}
where we write $r(\varphi)=\sum_{i=1}^{V_0(\Gamma)}r(\tau_i)$. This latter quantity is exactly that which is described in the third condition defining $\widetilde{\mathcal{S}}_{k_1,\ldots,k_n}^{(m),c,*}(h)$, so $(-1)^{r(\varphi)}=\pm1$ when $(\sigma_1,\ldots,\sigma_m)\in\widetilde{\mathcal{S}}_{k_1,\ldots,k_n}^{(m),c,\pm}(h)$. This concludes the proof.
\end{proof}

\begin{example}
Let $(\sigma_1,\sigma_2)\in\widetilde{\mathcal{S}}_2^{(2),c}$ correspond to the ribbon graph $\Gamma$ of Figure \ref{fig2}. Fixing convention by taking the top-most polygon edge to be of label $-1$ and type $X_2^T$, we have that $\sigma_1=(-1,-2)(1,2)$ and $\sigma_2=(-1,2)(1,-2)$; these encode the purple and grey ribbons, respectively. We compute $f\sigma_2=(-1)(1)(-2)(2)$, $\sigma_2\sigma_1=(-1,1)(-2,2)$, and $\sigma_1e=(-1,2)(1,-2)$. Hence, $V_0(\Gamma)=\tfrac{1}{2}\#\sigma_1e=1$, $V_1(\Gamma)=\tfrac{1}{2}\#\sigma_2\sigma_1=1$, and $V_2(\Gamma)=\#f\sigma_2=2$, which agrees with $\Gamma$ having one red, one blue, and two green boundaries. Comparing to the first condition of Proposition \ref{prop4} shows that the Euler genus is $h=1$, while the second condition is confirmed upon noting that $\varphi=(1,-2)$. According to the third condition, we have that $(\sigma_1,\sigma_2)\in\widetilde{\mathcal{S}}_2^{(2),c,-1}(1)$, which is in keeping with the $J$-weight computed in the caption of Figure \ref{fig2}.
\end{example}

In the antisymmetrised case, the connection to standard constellations is more tenuous than in the Hermitised case. Nevertheless, we may define for $(\sigma_1,\ldots,\sigma_m)\in\widetilde{\mathcal{S}}_{k_1,\ldots,k_n}^{(m),c}$ the sequence of permutations $\alpha_0=\sigma_1e$, $\alpha_i=\sigma_{i+1}\sigma_i$ for $i\in[m-1]$, and $\alpha_m=f\sigma_m$ --- note that $(\sigma_1,\ldots,\sigma_m)$ can be recovered from $(\alpha_0,\ldots,\alpha_m)$. Since $e,f,\sigma_1,\ldots,\sigma_m$ are involutions such that $\mean{e,f,\sigma_1,\ldots,\sigma_m}$ is transitive, we have that $ef\alpha_m\cdots\alpha_0=\mathrm{id}$ and $\mean{e,f,\alpha_m,\ldots,\alpha_0}$ is transitive. Hence, $(e,f,\alpha_m,\ldots,\alpha_0)$ is an $(m+3)$-constellation.

\setcounter{equation}{0}
\section{Loop equations for the antisymmetrised Laguerre ensemble} \label{s4}
We now develop the loop equation hierarchy for the $(N,N)$ antisymmetrised Laguerre ensemble, which corresponds to specialising to $m=1$ and $N_1=N_0=N$. As mentioned in \S\ref{s1.2}, this amounts to using integration by parts to produce recursions on the mixed moments $m^{(\mathcal{J}_1)}_{k_1,\ldots,k_n}$, their generating functions $U_n^{(\mathcal{J}_1)}(x_1,\ldots,x_n)$, the connected analogues $W_n^{(\mathcal{J}_1)}(x_1,\ldots,x_n)$, and the correlator expansion coefficients $W_n^{(\mathcal{J}_1),l}(x_1,\ldots,x_n)$, in that order. For the remainder of this section, we drop the subscript $(\mathcal{J}_1)$ for clarity.

The calculations forthcoming serve as a first step towards understanding the loop equation formalism for antisymmetrised matrix product ensembles and, moreover, for random matrix products in general. A point of interest concerning loop equations for matrix products is that their spectral curves, i.e., the loop equations characterising $W_1^0(x)$, are polynomials in $W_1^0(x)$ of order higher than two. This is in contrast with the second order spectral curves seen for classical $\beta$ ensembles, such as those solved by the Wigner semicircle and Mar\v{c}enko--Pastur laws \citep{EM09}, \citep{WF14}, \citep{WF15}, \citep{FRW17}. Spectral curves of order higher than two have garnered interest in the abstract setting \citep{EO09}, \citep{BE13}, \citep{BHLMR14}, \citep{Ora15}, but there have not been many concrete examples of such higher order spectral curves and, more broadly, loop equations (here, higher order means higher than usual) arising from random matrix theory (see, however, the studies on so-called multi-matrix models \citep{EO09}). We will see that the antisymmetrised Laguerre ensemble has a third order spectral curve and that the loop equations have a third order structure in general.

To ease notation, we write $B=\mathcal{J}_1=X^TJX$ and adopt the Einstein summation convention so that repeated matrix indices in a product of matrix entries are summed over $1,\ldots,N$. For example, the expression $A_{ab}B_{bc}$ is hiding an implicit sum over $b=1,\ldots,N$ and is given by $[AB]_{ac}$. In particular, we must take care to remember that $J_{aa}=\Tr J=0$ and $\partial_{X_{ab}}X_{ab}=\sum_{a,b=1}^N1=N^2$. Moreover, we define $I_n:=(k_2,\ldots,k_n)$ and use the shorthand $\Tr B^{I_n}=\Tr B^{k_2}\cdots\Tr B^{k_n}$. Finally, we place braces around partial derivatives to indicate that the differential operators do not act on any terms outside the braces. This is especially important inside averages, where the absence of braces means that differential operators act on the product of all terms on their right, including the p.d.f.~present in the integral definition of the average.

The intuition behind the first calculation of the loop equation formalism is to represent the factor $\Tr B^{k_1+1}$ in the definition of $m_{k_1+1,I_n}$ as a polygon --- or loop --- of $k_1+1$ sides, remove one factor of $B$ to cut open the loop, and then reintroduce that same term to join the loop back up, except that the reintroduced term is now extracted from the p.d.f. \eqref{eq1.3} of the real Ginibre ensemble by way of using an appropriate differential operator. This intuition is why loop equations are also referred to as cut-and-join equations. We give an example as a warm-up.

\begin{lemma} \label{Lemma4}
Let us write $\Tr B^{k_1+1}=(B^{k_1})_{ad}B_{da}$ (using the Einstein summation convention) and then erase the last term so that we are left with $(B^{k_1})_{ad}$. We have that
\begin{equation} \label{eq4.1}
\partial_{(X^T)_{ab}}(J^T)_{bc}\partial_{X_{cd}}\exp\left(-\tfrac{N}{2}\Tr(X^TX)\right)=N^2B_{da}\exp\left(-\tfrac{N}{2}\Tr(X^TX)\right)
\end{equation}
and, consequently,
\begin{equation} \label{eq4.2}
\mean{(B^{k_1})_{ad}\Tr B^{I_n}\partial_{(X^T)_{ab}}(J^T)_{bc}\partial_{X_{cd}}}=N^2m_{k_1+1,I_n}.
\end{equation}
\end{lemma}
\begin{proof}
First note by the Leibniz product rule that
\begin{multline*}
\partial_{(X^T)_{ab}}\partial_{X_{cd}}\exp\left(-\tfrac{N}{2}\Tr(X^TX)\right)=-N\partial_{(X^T)_{ab}}X_{cd}\exp\left(-\tfrac{N}{2}\Tr(X^TX)\right)
\\=\left(N^2(X^T)_{ab}X_{cd}-N\chi_{a=d,b=c}\right)\exp\left(-\tfrac{N}{2}\Tr(X^TX)\right).
\end{multline*}
Multiplying this result by $(J^T)_{bc}$ gives
\begin{equation*}
\left(N^2(X^TJ^TX)_{ad}-NJ_{bb}\chi_{a=d,b=c}\right)\exp\left(-\tfrac{N}{2}\Tr(X^TX)\right),
\end{equation*}
which reduces to the right-hand side of equation \eqref{eq4.1} upon recalling that $J_{bb}=\Tr J=0$ and observing that $(X^TJ^TX)^T=X^TJX$, so $(X^TJ^TX)_{ad}=B_{da}$. Pre-multiplying both sides of equation \eqref{eq4.1} by $(N/2\pi)^{N^2/2}(B^{k_1})_{ad}\Tr B^{I_n}$ and then integrating over $\mathbb{M}_{N\times N}(\mathbb{R})$ while keeping track of notation yields equation \eqref{eq4.2}.
\end{proof}

Here, we have related $m_{k_1+1,I_n}$ to an average that looks like $m_{k_1+1,I_n}$, except that a factor of $B_{da}$ has been replaced by $\partial_{(X^T)_{ab}}(J^T)_{bc}\partial_{X_{cd}}$. As this average is non-zero, one might be led to considering the closely related average $\mean{\partial_{(X^T)_{ab}}(J^T)_{bc}\partial_{X_{cd}}(B^{k_1})_{ad}\Tr B^{I_n}}$, which vanishes due to the fundamental theorem of calculus. However, this results in averages of the form $\mean{\Tr(B^pX^TX)\Tr B^q}$, which are not immediately expressible in terms of mixed moments and one must develop a series of complicated identities that eventually combine to produce an identity purely on mixed moments. It turns out that a lot of guesswork can be avoided by beginning at the average
\begin{equation} \label{eq4.3}
\mean{\left(\partial_{(X^T)_{ab}}+NX_{ba}\right)(J^T)_{bc}\left(\partial_{X_{cd}}+N(X^T)_{dc}\right)(B^{k_1})_{ad}\Tr B^{I_n}}.
\end{equation}
In this way, we are cutting and joining the loop represented by $\Tr B^{k_1+1}$ by removing and reintroducing factors of $X^T$ and $X$ belonging to a single factor of $B$, rather than removing and reintroducing said factor of $B$ itself.

\subsection{Loop equations on the \texorpdfstring{$U_n^{(\mathcal{J}_1)}$}{Un(J1)}} \label{s4.1}
Before unpacking the average \eqref{eq4.3}, let us first give the companion to Lemma \ref{Lemma4}.
\begin{lemma} \label{Lemma5}
The partial derivatives of $(B^{k_1})_{ef}$ with respect to $(X^T)_{ab}$ and $X_{cd}$ (taking $e,f$ to be arbitrary and possibly equal to each other or one of $a,b,c,d$) are given by
\begin{align}
\partial_{(X^T)_{ab}}(B^{k_1})_{ef}&=\sum_{p_1+p_2=k_1-1}\left((B^{p_1})_{ea}(JXB^{p_2})_{bf}+(B^{p_1}X^TJ)_{eb}(B^{p_2})_{af}\right), \label{eq4.4}
\\ \partial_{X_{cd}}(B^{k_1})_{ef}&=\sum_{p_1+p_2=k_1-1}\left((B^{p_1})_{ed}(JXB^{p_2})_{cf}+(B^{p_1}X^TJ)_{ec}(B^{p_2})_{df}\right), \label{eq4.5}
\end{align}
where the sums over $p_1+p_2=k_1-1$ are over the range $p_1=0,1,\ldots,k_1-1$ with $0\leq p_2\leq k_1-1$ set to $k_1-p_1-1$; these sums are empty for $k_1=0$. For later convenience, let us also mention that
\begin{multline} \label{eq4.6}
\partial_{(X^T)_{ab}}(B^{k_1}X^T)_{ef}=\sum_{p_1+p_2=k_1-1}\left((B^{p_1})_{ea}(JXB^{p_2}X^T)_{bf}+(B^{p_1}X^TJ)_{eb}(B^{p_2}X^T)_{af}\right)
\\+\chi_{f=b}(B^{k_1})_{ea}.
\end{multline}
\end{lemma}
\begin{proof}
First note that by the Leibniz product rule,
\begin{multline*}
\partial_{(X^T)_{ab}}B_{gh}=\left\{\partial_{(X^T)_{ab}}(X^T)_{gl}\right\}(JX)_{lh}+(X^TJ)_{gl}\left\{\partial_{(X^T)_{ab}}X_{lh}\right\}
\\=\chi_{g=a,l=b}(JX)_{lh}+(X^TJ)_{gl}\chi_{l=b,h=a}=\chi_{g=a}(JX)_{bh}+\chi_{h=a}(X^TJ)_{gb}.
\end{multline*}
Thus, the left-hand side of equation \eqref{eq4.5} decomposes as
\begin{multline*}
\partial_{(X^T)_{ab}}(B^{k_1})_{ef}=\sum_{p_1=0}^{k_1-1}(B^{p_1})_{eg}\left\{\partial_{(X^T)_{ab}}B_{gh}\right\}(B^{k_1-p_1-1})_{hf}
\\=\sum_{p_1=0}^{k_1-1}\left((B^{p_1})_{ea}(JX)_{bh}(B^{k_1-p_1-1})_{hf}+(B^{p_1})_{eg}(X^TJ)_{gb}(B^{k_1-p_1-1})_{af}\right),
\end{multline*}
which simplifies to the right-hand side of equation \eqref{eq4.4} upon summing over the repeated indices $g,h$. To obtain equation \eqref{eq4.5} from equation \eqref{eq4.4}, simply observe that setting $(c,d)=(b,a)$ shows that $X_{cd}=(X^T)_{ab}$. For equation \eqref{eq4.6}, note that
\begin{align*}
\partial_{(X^T)_{ab}}(B^{k_1}X^T)_{ef}&=\left\{\partial_{(X^T)_{ab}}(B^{k_1})_{eg}\right\}(X^T)_{gf}+(B^{k_1})_{eg}\left\{\partial_{(X^T)_{ab}}(X^T)_{gf}\right\}
\\&=\left\{\partial_{(X^T)_{ab}}(B^{k_1})_{eg}\right\}(X^T)_{gf}+(B^{k_1})_{eg}\chi_{g=a,f=b}
\\&=\left\{\partial_{(X^T)_{ab}}(B^{k_1})_{eg}\right\}(X^T)_{gf}+(B^{k_1})_{ea}\chi_{f=b};
\end{align*}
substituting equation \eqref{eq4.4} with $f\mapsto g$ into the braces then completes the proof.
\end{proof}

With these identities in hand, we can now apply integration by parts to the average \eqref{eq4.3} in two different ways to obtain a precursor to the loop equation on the $U_n$.
\begin{proposition} \label{prop5}
The average \eqref{eq4.3} simplifies to both the left- and right-hand sides of the following equation:
\begin{equation} \label{eq4.7}
N^2m_{k_1+1,I_n}=\mean{\left\{\partial_{(X^T)_{ab}}(J^T)_{bc}\partial_{X_{cd}}(B^{k_1})_{ad}\Tr B^{I_n}\right\}}.
\end{equation}
\end{proposition}
\begin{proof}
For the left-hand side, expanding the argument of the average \eqref{eq4.3} shows that
\begin{align}
&\mean{\left(\partial_{(X^T)_{ab}}+NX_{ba}\right)(J^T)_{bc}\left(\partial_{X_{cd}}+N(X^T)_{dc}\right)(B^{k_1})_{ad}\Tr B^{I_n}} \nonumber
\\&=\mean{\partial_{(X^T)_{ab}}(J^T)_{bc}\partial_{X_{cd}}(B^{k_1})_{ad}\Tr B^{I_n}}+N\mean{\partial_{(X^T)_{ab}}(J^T)_{bc}(X^T)_{dc}(B^{k_1})_{ad}\Tr B^{I_n}} \nonumber
\\&\quad+N\mean{X_{ba}(J^T)_{bc}\partial_{X_{cd}}(B^{k_1})_{ad}\Tr B^{I_n}}+N^2\mean{X_{ba}(J^T)_{bc}(X^T)_{dc}(B^{k_1})_{ad}\Tr B^{I_n}}.  \label{eq4.8}
\end{align}
The first two terms on the right-hand side of this equation are manifestly zero by the fundamental theorem of calculus. Moreover, contracting indices in the fourth term on the right-hand side of equation \eqref{eq4.8} while rewriting $(J^T)_{bc}$ as $J_{cb}$ shows that this term is equal to the left-hand side of equation \eqref{eq4.7}. Finally, to see that the third term on the right-hand side of equation \eqref{eq4.8} vanishes, note that for $f(X)$ a function of the entries of $X$, integration by parts shows that
\begin{align*}
\mean{X_{ba}(J^T)_{bc}\partial_{X_{cd}}f(X)}&=\mean{\partial_{X_{cd}}X_{ba}(J^T)_{bc}f(X)}-\mean{\left\{\partial_{X_{cd}}X_{ba}\right\}(J^T)_{bc}f(X)}
\\&=\mean{\partial_{X_{cd}}X_{ba}(J^T)_{bc}f(X)}-\mean{\chi_{a=d}(J^T)_{bb} f(X)},
\end{align*}
which reduces to zero by application of the fundamental theorem of calculus to the first term and by making the substitution $(J^T)_{bb}=\Tr J^T=0$ in the second term.

Now, for the right-hand side of equation \eqref{eq4.7}, let $f(X)$ once again be a function of the entries of $X$ and observe that
\begin{multline} \label{eq4.9}
\left(\partial_{(X^T)_{ab}}+NX_{ba}\right)f(X)\exp\left(-\tfrac{N}{2}\Tr(X^TX)\right)
\\=\left\{\partial_{(X^T)_{ab}}f(X)\right\}\exp\left(-\tfrac{N}{2}\Tr(X^TX)\right)+f(X)\left\{\partial_{(X^T)_{ab}}\exp\left(-\tfrac{N}{2}\Tr(X^TX)\right)\right\}
\\+NX_{ba}f(X)\exp\left(-\tfrac{N}{2}\Tr(X^TX)\right)
\\=\left\{\partial_{(X^T)_{ab}}f(X)\right\}\exp\left(-\tfrac{N}{2}\Tr(X^TX)\right)
\end{multline}
since $\partial_{(X^T)_{ab}}\exp\left(-\tfrac{N}{2}\Tr(X^TX)\right)$ is equal to $-NX_{ba}\exp\left(-\tfrac{N}{2}\Tr(X^TX)\right)$. We similarly have
\begin{equation*}
\left(\partial_{X_{cd}}+N(X^T)_{dc}\right)f(X)\exp\left(-\tfrac{N}{2}\Tr(X^TX)\right)=\left\{\partial_{X_{cd}}f(X)\right\}\exp\left(-\tfrac{N}{2}\Tr(X^TX)\right).
\end{equation*}
Using this latter identity shows that the average \eqref{eq4.3} simplifies as
\begin{multline*}
\mean{\left(\partial_{(X^T)_{ab}}+NX_{ba}\right)(J^T)_{bc}\left(\partial_{X_{cd}}+N(X^T)_{dc}\right)(B^{k_1})_{ad}\Tr B^{I_n}}
\\=\mean{\left(\partial_{(X^T)_{ab}}+NX_{ba}\right)(J^T)_{bc}\left\{\partial_{X_{cd}}(B^{k_1})_{ad}\Tr B^{I_n}\right\}},
\end{multline*}
which can be seen to be precisely the right-hand side of equation \eqref{eq4.7} upon setting
\begin{equation*}
f(X)=(J^T)_{bc}\left\{\partial_{X_{cd}}(B^{k_1})_{ad}\Tr B^{I_n}\right\}
\end{equation*}
into the identity \eqref{eq4.9}.
\end{proof}

Our claim is that using the Leibniz product rule to expand the derivative in the right-hand side of equation \eqref{eq4.7} and then properly simplifying each of the averages in the resulting sum will produce our sought loop equation. To demonstrate this claim, let us now write the right-hand side of equation \eqref{eq4.7} as
\begin{equation} \label{eq4.10}
\mean{\left\{\partial_{(X^T)_{ab}}(J^T)_{bc}\partial_{X_{cd}}(B^{k_1})_{ad}\Tr B^{I_n}\right\}}=\mathcal{A}+\sum_{i=2}^n\left(\mathcal{A}_i+\overline{\mathcal{A}}_i+\mathcal{B}_i\right)+\sum_{\substack{2\leq i,j\leq n,\\i\neq j}}\mathcal{C}_{ij},
\end{equation}
where
\begin{align}
\mathcal{A}&=\mean{\left\{\partial_{(X^T)_{ab}}(J^T)_{bc}\partial_{X_{cd}}(B^{k_1})_{ad}\right\}\Tr B^{I_n}}, \label{eq4.11}
\\ \mathcal{A}_i&=\mean{\left\{\partial_{(X^T)_{ab}}(B^{k_1})_{ad}\right\}(J^T)_{bc}\left\{\partial_{X_{cd}}\Tr B^{k_i}\right\}\Tr B^{I_n\setminus\{k_i\}}},
\\ \overline{\mathcal{A}}_i&=\mean{\left\{\partial_{(X^T)_{ab}}\Tr B^{k_i}\right\}(J^T)_{bc}\left\{\partial_{X_{cd}}(B^{k_1})_{ad}\right\}\Tr B^{I_n\setminus\{k_i\}}},
\\ \mathcal{B}_i&=\mean{\left\{\partial_{(X^T)_{ab}}(J^T)_{bc}\partial_{X_{cd}}\Tr B^{k_i}\right\}(B^{k_1})_{ad}\Tr B^{I_n\setminus\{k_i\}}},
\\ \mathcal{C}_{ij}&=\mean{\left\{\partial_{(X^T)_{ab}}\Tr B^{k_i}\right\}(J^T)_{bc}\left\{\partial_{X_{cd}}\Tr B^{k_j}\right\}(B^{k_1})_{ad}\Tr B^{I_n\setminus\{k_i,k_j\}}}. \label{eq4.15}
\end{align}
We show in the following lemma that each of these terms can be expressed in terms of the mixed moments, so combining equations \eqref{eq4.7} and \eqref{eq4.10} will result in a loop equation on said moments, as claimed.
\begin{lemma} \label{Lemma6}
For $k\in\mathbb{N}$, let $[k]_{\mathrm{mod}\, 2}$ equal zero when $k$ is even and one when $k$ is odd. Then, the quantities \eqref{eq4.11}--\eqref{eq4.15} above are given by
\begin{align}
\mathcal{A}&=\sum_{p_1+p_2=k_1-1}m_{p_1,p_2,I_n}-\sum_{p_1+p_2+p_3=k_1-1}m_{p_1,p_2,p_3,I_n}+\frac{1-k_1^2}{2}m_{k_1-1,I_n}, \label{eq4.16}
\\ \mathcal{A}_i=\overline{\mathcal{A}}_i&=2k_i[k_i+1]_{\mathrm{mod}\,2}\left(m_{k_1+k_i-1,I_n\setminus\{k_i\}}-\sum_{p_1+p_2=k_1-1}m_{k_i+p_1,p_2,I_n\setminus\{k_i\}}\right), \label{eq4.17}
\\ \mathcal{B}_i&=-2k_i[k_i+1]_{\mathrm{mod}\,2}\left(m_{k_1+k_i-1,I_n\setminus\{k_i\}}+\sum_{p_1+p_2=k_i-1}m_{k_1+p_1,p_2,I_n\setminus\{k_i\}}\right), \label{eq4.18}
\\ \mathcal{C}_{ij}=\mathcal{C}_{ji}&=-4k_ik_j[k_i+1]_{\mathrm{mod}\,2}[k_j+1]_{\mathrm{mod}\,2}\,m_{k_1+k_i+k_j-1,I_n\setminus\{k_i,k_j\}}. \label{eq4.19}
\end{align}
\end{lemma}
\begin{proof}
We present only the computations of $\mathcal{A}$ and $\mathcal{C}_{ij}$ since their derivations contain all of the ideas needed to prove equations \eqref{eq4.17} and \eqref{eq4.18}. We first give the proof of equation \eqref{eq4.19}, as it is simpler than that of equation \eqref{eq4.16}. Thus, we begin by using equation \eqref{eq4.4} with $(k_1,f)\mapsto(k_i,e)$ to see that
\begin{multline} \label{eq4.20}
\partial_{(X^T)_{ab}}\Tr B^{k_i}=\partial_{(X^T)_{ab}}(B^{k_i})_{ee}=\sum_{p_1+p_2=k_i-1}\left((JXB^{k_i-1})_{ba}+(B^{k_i-1}X^TJ)_{ab}\right)
\\=(JXB^{k_i-1})_{ba}\sum_{p_1+p_2=k_i-1}\left(1+(-1)^{k_i}\right)=2k_i[k_i+1]_{\mathrm{mod}\,2}(JXB^{k_i-1})_{ba}.
\end{multline}
Here, we have used the fact that $(B^{k_i-1}X^TJ)^T=J^TX(B^T)^{k_i-1}=(-1)^{k_i}JXB^{k_i-1}$ and that $(1+(-1)^{k_i})$ equals two if $k_i$ is even and vanishes, otherwise. Repeating this argument with equation \eqref{eq4.5} likewise shows that
\begin{equation*}
\partial_{X_{cd}}\Tr B^{k_j}=2k_j[k_j+1]_{\mathrm{mod}\,2}(B^{k_j-1}X^TJ)_{dc}.
\end{equation*}
Multiplying these two results together with $(J^T)_{bc}$ then reveals the identity
\begin{equation*}
\left\{\partial_{(X^T)_{ab}}\Tr B^{k_i}\right\}(J^T)_{bc}\left\{\partial_{X_{cd}}\Tr B^{k_j}\right\}=-4k_ik_j[k_i+1]_{\mathrm{mod}\,2}[k_j+1]_{\mathrm{mod}\,2}(B^{k_i+k_j-1})_{da}.
\end{equation*}
Further multiplying this by $(B^{k_1})_{ad}\Tr B^{I_n\setminus\{k_i,k_j\}}$, summing over the repeated indices $a,d$, and then taking the average yields the right-hand side of equation \eqref{eq4.19}, as required.

Similar to above, we begin the proof of equation \eqref{eq4.16} by noting that equation \eqref{eq4.5} with $(e,f)=(a,d)$ tells us that
\begin{multline} \label{eq4.21}
(J^T)_{bc}\partial_{X_{cd}}(B^{k_1})_{ad}=\sum_{p_1+p_2=k_1-1}\left((-1)^{p_2}(B^{k_1-1}X^T)_{ab}-(B^{p_1}X^T)_{ab}\Tr B^{p_2}\right)
\\=[k_1]_{\mathrm{mod}\,2}(B^{k_1-1}X^T)_{ab}-\sum_{p_1+p_2=k_1-1}(B^{p_1}X^T)_{ab}\Tr B^{p_2},
\end{multline}
where we have used the fact that $\sum_{p_1+p_2=k_1-1}(-1)^{p_2}=\sum_{p_2=0}^{k_1-1}(-1)^{p_2}$ equals one if $k_1$ is odd and zero, otherwise. Now, taking $(k_1,e,f)\mapsto(k_1-1,a,b)$ in equation \eqref{eq4.6} shows that
\begin{align}
\partial_{(X^T)_{ab}}(B^{k_1-1}X^T)_{ab}&=\sum_{p_1+p_2=k_1-2}\left(\Tr B^{p_1}\Tr B^{p_2+1}+(-1)^{p_2}\Tr B^{k_1-1}\right)+N\Tr B^{k_1-1} \nonumber
\\&=\sum_{p_1+p_2=k_1-1}\Tr B^{p_1}\Tr B^{p_2}+[k_1+1]_{\mathrm{mod}\,2}\Tr B^{k_1-1} \nonumber
\\&=\sum_{p_1+p_2=k_1-1}\Tr B^{p_1}\Tr B^{p_2}; \label{eq4.22}
\end{align}
the second line follows from making the replacement $p_2\mapsto p_2-1$ in the sum and noting that the new $(p_1,p_2)=(k_1-1,0)$ term this introduces is equivalent to the term $N\Tr B^{k_1-1}$ that already appears at the end of the first line, while the final line follows from the fact that $[k_1+1]_{\mathrm{mod}\,2}$ vanishes whenever $k_1$ is odd, but $\Tr B^{k_1-1}$ vanishes whenever $k_1$ is even. Moving on to the second term in equation \eqref{eq4.21}, using equation \eqref{eq4.20} with the replacement $k_i\mapsto p_2$ and equation \eqref{eq4.22} with $k_1\mapsto p_1$ shows that
\begin{multline} \label{eq4.23}
\partial_{(X^T)_{ab}}(B^{p_1}X^T)_{ab}\Tr B^{p_2}=\left\{\partial_{(X^T)_{ab}}(B^{p_1}X^T)_{ab}\right\}\Tr B^{p_2}+(B^{p_1}X^T)_{ab}\left\{\partial_{(X^T)_{ab}}\Tr B^{p_2}\right\}
\\=\sum_{q_1+q_2=p_1}\Tr B^{q_1}\Tr B^{q_2}\Tr B^{p_2}+[p_1]_{\mathrm{mod}\,2}\Tr B^{p_1}\Tr B^{p_2}+2p_2[p_2+1]_{\mathrm{mod}\,2}\Tr B^{p_1+p_2}
\\=\sum_{q_1+q_2=p_1}\Tr B^{q_1}\Tr B^{q_2}\Tr B^{p_2}+2p_2[p_2+1]_{\mathrm{mod}\,2}\Tr B^{p_1+p_2},
\end{multline}
where the last line follows from the fact that $[p_1]_{\mathrm{mod}\,2}$ vanishes if $p_1$ is even, but $\Tr B^{p_1}$ vanishes if $p_1$ is odd. Applying the operator $\partial_{(X^T)_{ab}}$ to both sides of equation \eqref{eq4.21} and then substituting in equations \eqref{eq4.22} and \eqref{eq4.23} shows that
\begin{multline} \label{eq4.24}
\partial_{(X^T)_{ab}}(J^T)_{bc}\partial_{X_{cd}}(B^{k_1})_{ad}=[k_1]_{\mathrm{mod}\,2}\sum_{p_1+p_2=k_1-1}\Tr B^{p_1}\Tr B^{p_2}
\\-\sum_{p_1+p_2=k_1-1}\left(2p_2[p_2+1]_{\mathrm{mod}\,2}\Tr B^{p_1+p_2}+\sum_{q_1+q_2=p_1}\Tr B^{q_1}\Tr B^{q_2}\Tr B^{p_2}\right).
\end{multline}
The factor of $[k_1]_{\mathrm{mod}\,2}$ in the first term can be removed since it is superfluous: for $\Tr B^{p_1},\Tr B^{p_2}$ to be non-zero, each of $p_1,p_2$ must be even and, consequently, the sum over $p_1+p_2=k_1-1$ is non-zero only if $k_1$ is odd. Likewise, the second term is non-zero for $k_1$ odd and it simplifies as
\begin{equation*}
-\sum_{p_1+p_2=k_1-1}2p_2[p_2+1]_{\mathrm{mod}\,2}\Tr B^{p_1+p_2}=-2\Tr B^{k_1-1}(2+4+\cdots+k_1-1)=\frac{1-k_1^2}{2}\Tr B^{k_1-1},
\end{equation*}
while replacing $p_1$ in the outer sum of the third term by $q_1+q_2$ and then renaming the summation indices shows that
\begin{multline*}
-\sum_{p_1+p_2=k_1-1}\sum_{q_1+q_2=p_1}\Tr B^{q_1}\Tr B^{q_2}\Tr B^{p_2}=-\sum_{q_1+q_2+p_2=k_1-1}\Tr B^{q_1}\Tr B^{q_2}\Tr B^{p_2}
\\=-\sum_{p_1+p_2+p_3=k_1-1}\Tr B^{p_1}\Tr B^{p_2}\Tr B^{p_3}.
\end{multline*}
Combining these observations together, we thus see that equation \eqref{eq4.24} reduces to
\begin{multline*}
\partial_{(X^T)_{ab}}(J^T)_{bc}\partial_{X_{cd}}(B^{k_1})_{ad}=\sum_{p_1+p_2=k_1-1}\Tr B^{p_1}\Tr B^{p_2}+\frac{1-k_1^2}{2}\Tr B^{k_1-1}
\\-\sum_{p_1+p_2+p_3=k_1-1}\Tr B^{p_1}\Tr B^{p_2}\Tr B^{p_3}.
\end{multline*}
Finally, multiplying the right-hand side of this equation by $\Tr B^{I_n}$ and taking the average produces the right-hand side of equation \eqref{eq4.16}, as required.
\end{proof}

It is now a simple matter of substituting the expressions \eqref{eq4.16}--\eqref{eq4.19} into the right-hand side of equation \eqref{eq4.10} while replacing the left-hand side by $N^2m_{k_1+1,I_n}$, as prescribed by Propositon \ref{prop5}, to obtain the desired loop equation on the mixed moments $m_{k_1,\ldots,k_n}$. This loop equation is presented as equation \eqref{eqA.1} in Appendix \ref{appendixC.1}.

At this point, it is possible to proceed by substituting the moment-cumulants relation \eqref{eq1.7} into the loop equation on the moments to obtain an analogous equation on the cumulants $c_{k_1,\ldots,k_n}$. Simplifying this equation on the $c_{k_1,\ldots,k_n}$ in an appropriate manner (using a similar argument to that given in the proof of Proposition \ref{prop7} in \S\ref{s4.2} upcoming), multiplying the result by $x_1^{-k_1-1}\cdots x_n^{-k_n-1}$, and summing over $k_1,\ldots,k_n\geq0$ would then yield the corresponding loop equation on the connected correlators $W_n(x_1,\ldots,x_n)$. We do not go down this path, instead opting to first derive the loop equation on the unconnected correlators $U_n(x_1,\ldots,x_n)$ and then transforming it, through the identity \citep{Smi95}
\begin{equation} \label{eq4.25}
W_n(x_1,\ldots,x_n)=U_n(x_1,J_n)-\sum_{\emptyset\ne J\subseteq J_n}W_{n-\#J}(x_1,J_n\setminus J)U_{\#J}(J),
\end{equation}
into the loop equation on the $W_n(x_1,\ldots,x_n)$ --- here and henceforth, $J_n:=(x_2,\ldots,x_n)$ and $\#S$ denotes the size of the set $S$. If necessary, the loop equation on the mixed cumulants $c_{k_1,\ldots,k_n}$ can be extracted from the loop equation on the $W_n(x_1,\ldots,x_n)$ by equating coefficients of $x_1^{-k_1-1}\cdots x_n^{-k_n-1}$, possibly by taking suitable residues.

\begin{proposition} \label{prop6}
Set $U_0:=1$, and define $U_{n'}:=0$ for $n'<0$. Furthermore, for $i=2,\ldots,n$, define the auxiliary function
\begin{multline*}
A_i(x_1,J_n):=x_i^2\left(2U_n(x_1,x_1,J_n\setminus\{x_i\})-U_n(x_i,J_n)\right)-x_1x_iU_n(x_1,J_n)
\\+\frac{1}{x_1}\left(x_1x_iU_{n-1}(J_n)-x_i^2U_{n-1}(x_1,J_n\setminus\{x_i\})\right).
\end{multline*}
Then, for $n\geq1$, the unconnected correlators satisfy the loop equation
\begin{multline} \label{eq4.26}
0=x_1U_{n+2}(x_1,x_1,x_1,J_n)-U_{n+1}(x_1,x_1,J_n)+N^2x_1U_n(x_1,J_n)-N^3U_{n-1}(J_n)
\\+\frac{1}{2x_1}\left(x_1^2\frac{\partial^2}{\partial x_1^2}+x_1\frac{\partial}{\partial x_1}-1\right)U_n(x_1,J_n)+2\sum_{i=2}^n\frac{\partial}{\partial x_i}\left\{\frac{A_i(x_1,J_n)}{x_1^2-x_i^2}\right\}
\\+8\sum_{2\leq i<j\leq n}\frac{1}{x_1}\frac{\partial^2}{\partial x_i\partial x_j}\Bigg\{\frac{x_1x_i^2x_jU_{n-2}(J_n\setminus\{x_i\})}{(x_1^2-x_j^2)(x_i^2-x_j^2)}-\frac{x_1x_ix_j^2U_{n-2}(J_n\setminus\{x_j\})}{(x_1^2-x_i^2)(x_i^2-x_j^2)}
\\+\frac{x_i^2x_j^2U_{n-2}(x_1,J_n\setminus\{x_i,x_j\})}{(x_1^2-x_i^2)(x_1^2-x_j^2)}\Bigg\}.
\end{multline}
\end{proposition}
\begin{proof}
As mentioned above, inserting the equalities \eqref{eq4.16}--\eqref{eq4.19} into equation \eqref{eq4.10} and combining the result with equation \eqref{eq4.7} gives the loop equation \eqref{eqA.1} on the mixed moments $m_{k_1,\ldots,k_n}$. Multiplying both sides of this loop equation by $x_1^{-k_1-1}\cdots x_n^{-k_n-1}$ and summing over $k_1,\ldots,k_n\geq0$ then produces the loop equation \eqref{eq4.26} above. We detail this transformation term by term in Appendix \ref{appendixC.1}.
\end{proof}

\subsection{Loop equations on the \texorpdfstring{$W_n^{(\mathcal{J}_1)}$}{Wn(J1)} and \texorpdfstring{$W_n^{(\mathcal{J}_1),l}$}{Wn(J1),l}} \label{s4.2}
Our main goal in this subsection is to transform the loop equation \eqref{eq4.26} on the $U_n(x_1,\ldots,x_n)$ into a loop equation on the $W_n(x_1,\ldots,x_n)$ through the use of the identity \eqref{eq4.25}. Actually, we use higher order analogues of this identity, which we present below.
\begin{lemma} \label{Lemma7}
Recall that $J_n=(x_2,\ldots,x_n)$ and let $\mu\vdash(x_1,\ldots,x_n)$ indicate that $\mu$ is a partition of the ordered set $(x_1,\ldots,x_n)$, i.e., $\mu=\{\mu_t\}_{t=1}^m$ for some $1\leq m\leq n$ such that the disjoint union $\sqcup_{t=1}^m\mu_t=\{x_1,\ldots,x_n\}$. Writing $\#\mu$ for the size of $\mu$ ($m$ in the preceding sentence), we have that
\begin{align}
U_n(x_1,J_n)&=\sum_{K_1\sqcup K_2=J_n}U_{\#K_1}(K_1)W_{\#K_2+1}(x_1,K_2), \label{eq4.27}
\\ U_{n+1}(x_1,x_1,J_n)&=\sum_{\mu\vdash(x_1,x_1)}\sum_{\sqcup_{t=1}^{\#\mu+1}K_t=J_n}U_{\#K_1}(K_1)\prod_{\mu_t\in\mu}W_{\#K_{t+1}+\#\mu_t}(\mu_t,K_{t+1}), \label{eq4.28}
\\ U_{n+2}(x_1,x_1,x_1,J_n)&=\sum_{\mu\vdash(x_1,x_1,x_1)}\sum_{\sqcup_{t=1}^{\#\mu+1}K_t=J_n}U_{\#K_1}(K_1)\prod_{\mu_t\in\mu}W_{\#K_{t+1}+\#\mu_t}(\mu_t,K_{t+1}). \label{eq4.29}
\end{align}
\end{lemma}
\begin{proof}
Equation \eqref{eq4.27} is a simple rewriting of equation \eqref{eq4.25}. This equation tells us, upon increasing $n$ by one and replacing $J_n$ with $(x_1,J_n)$, that
\begin{equation*}
U_{n+1}(x_1,x_1,J_n)=\sum_{K_1\sqcup K_2=(x_1,J_n)}U_{\#K_1}(K_1)W_{\#K_2+1}(x_1,K_2).
\end{equation*}
Splitting this sum based on whether or not $K_1$ contains $x_1$ gives
\begin{equation*}
U_{n+1}(x_1,x_1,J_n)=\sum_{K_1\sqcup K_2=J_n}\left[U_{\#K_1+1}(x_1,K_1)W_{\#K_2+1}(x_1,K_2)+U_{\#K_1}(K_1)W_{\#K_2+2}(x_1,x_2,K_2)\right].
\end{equation*}
Substituting in an appropriate rewriting of equation \eqref{eq4.27} for $U_{\#K_1+1}(x_1,K_1)$ then simplifies this expression to equation \eqref{eq4.28}. Equation \eqref{eq4.29} follows likewise through a repetition of this exercise.
\end{proof}
Observe that since we are dealing with partitions of ordered sets, the partition $\mu=\{\{x_1\},\{x_1,x_1\}\}$ is counted three times: the $x_1$ in $\mu_1=\{x_1\}$ could have been the first, second, or third entry in $(x_1,x_1,x_1)$. Although the $\mu_t$ must be non-empty, the $K_t$ are may be empty.

With Lemma \ref{Lemma7} in hand, we can now start computing loop equations for the connected correlators $W_n(x_1,\ldots,x_n)$. Thus, setting $n=1$ in equation \eqref{eq4.26} and then replacing $U_0$, $U_1(x_1)$, $U_2(x_1,x_1)$, and $U_3(x_1,x_1,x_1)$ by $1$, $W_1(x_1)$, the right-hand side of equation \eqref{eq4.28} with $n=1$, and the right-hand side of equation \eqref{eq4.29} with $n=1$, respectively, yields the \textit{$n=1$ loop equation on the connected correlators},
\begin{multline} \label{eq4.30}
0=x_1\left[W_3(x_1,x_1,x_1)+3W_2(x_1,x_1)W_1(x_1)+W_1(x_1)^3\right]-W_2(x_1,x_1)-W_1(x_1)^2
\\+N^2x_1W_1(x_1)-N^3+\frac{1}{2x_1}\left(x_1^2\frac{\mathrm{d}^2}{\mathrm{d}x_1^2}+x_1\frac{\mathrm{d}}{\mathrm{d}x_1}-1\right)W_1(x_1).
\end{multline}

It can immediately be seen that this loop equation has too many unknowns ($W_1,W_2$, and $W_3$) to be solvable. However, substituting the large $N$ expansion \eqref{eq1.10} (which is valid due to Proposition \ref{prop3}) into equation \eqref{eq4.30} and equating terms of like order in $N$ produces, for each $l\geq0$, the $(1,l)$ loop equation used to compute $W_1^l(x_1)$,
\begin{multline} \label{eq4.31}
0=x_1\Big[W_3^{l-4}(x_1,x_1,x_1)+3\sum_{l_1+l_2=l-2}W_2^{l_1}(x_1,x_1)W_1^{l_2}(x_1)+\sum_{l_1+l_2+l_3=l}W_1^{l_1}(x_1)W_1^{l_2}(x_1)W_1^{l_3}(x_1)\Big]
\\-W_2^{l-3}(x_1,x_1)-\sum_{l_1+l_2=l-1}W_1^{l_1}(x_1)W_1^{l_2}(x_1)+x_1W_1^l(x_1)-\chi_{l=0}
\\+\frac{1}{2x_1}\left(x_1^2\frac{\mathrm{d}^2}{\mathrm{d}x_1^2}+x_1\frac{\mathrm{d}}{\mathrm{d}x_1}-1\right)W_1^{l-2}(x_1);
\end{multline}
we set $W_n^{l'}:=0$ for all $n\geq1$ and $l'<0$. In particular, setting $l=0$ in the above gives the spectral curve
\begin{equation} \label{eq4.32}
0=x_1\left(W_1^0(x_1)\right)^3+x_1W_1^0(x_1)-1.
\end{equation}
Being a cubic polynomial in $W_1^0(x_1)$, this equation can be solved to express $W_1^0(x_1)$ as a cube root (although equation \eqref{eq4.32} has three solutions, only one is consistent with the fact that $W_1^0(x)\sim1/x=m_0/x$ as $x\to\infty$). We do not make $W_1^0(x_1)$ explicit here, but refer to \citep{FL15}, \citep{DF20} and references therein for techniques on obtaining such an expression --- in \S\ref{s4.3}, we instead find a rational parametrisation for this spectral curve.

Setting $l=1$ in equation \eqref{eq4.31} enables the computation of $W_1^1(x_1)$ with knowledge of $W_1^0(x_1)$, but at $l=2$, this equation becomes unsolvable due to the introduction of both $W_1^2(x_1)$ and $W_2^0(x_1,x_2)$. Thus, beyond $l=1$, one must consider loop equations for $n>1$. Setting $n=2$ in equation \eqref{eq4.26} and using Lemma \ref{Lemma7} with $n=2$ shows that
\begin{multline} \label{eq4.33}
0=x_1\sum_{\mu\vdash(x_1,x_1,x_1)}\sum_{\sqcup_{t=1}^{\#\mu+1}K_t=\{x_2\}}U_{\#K_1}(K_1)\prod_{\mu_t\in\mu}W_{\#K_{t+1}+\#\mu_t}(\mu_t,K_{t+1})
\\-\sum_{\mu\vdash(x_1,x_1)}\sum_{\sqcup_{t=1}^{\#\mu+1}K_t=\{x_2\}}U_{\#K_1}(K_1)\prod_{\mu_t\in\mu}W_{\#K_{t+1}+\#\mu_t}(\mu_t,K_{t+1})-N^3U_1(x_2)
\\+\left[N^2x_1+\frac{1}{2x_1}\left(x_1^2\frac{\partial^2}{\partial x_1^2}+x_1\frac{\partial}{\partial x_1}-1\right)\right]\left(W_2(x_1,x_2)+W_1(x_1)U_1(x_2)\right)
\\+\frac{2}{x_1}\frac{\partial}{\partial x_2}\frac{1}{x_1^2-x_2^2}\Big[x_1x_2^2\left(2W_2(x_1,x_1)+2W_1(x_1)^2-W_2(x_2,x_2)-W_1(x_2)^2\right)
\\\hspace{6em}-x_1^2x_2\left(W_2(x_1,x_2)+W_1(x_1)W_1(x_2)\right)
\\+x_1x_2W_1(x_2)-x_2^2W_1(x_1)\Big].
\end{multline}
This is not yet our sought loop equation, as we can simplify it one step further. Indeed, subtracting the product of the right-hand side of equation \eqref{eq4.30} by $U_1(x_2)$, which is zero, from equation \eqref{eq4.33} reveals the \textit{$n=2$ loop equation on the connected correlators},
\begin{multline} \label{eq4.34}
0=x_1\Big[W_4(x_1,x_1,x_1,x_2)+3W_3(x_1,x_1,x_2)W_1(x_1)
\\+3W_2(x_1,x_1)W_2(x_1,x_2)+3W_2(x_1,x_2)W_1(x_1)^2\Big]-W_3(x_1,x_1,x_2)
\\-2W_2(x_1,x_2)W_1(x_1)+N^2x_1W_2(x_1,x_2)+\frac{1}{2x_1}\left(x_1^2\frac{\partial^2}{\partial x_1^2}+x_1\frac{\partial}{\partial x_1}-1\right)W_2(x_1,x_2)
\\+\frac{2}{x_1}\frac{\partial}{\partial x_2}\frac{1}{x_1^2-x_2^2}\Big[x_1x_2^2\left(2W_2(x_1,x_1)+2W_1(x_1)^2-W_2(x_2,x_2)-W_1(x_2)^2\right)
\\-x_1^2x_2\left(W_2(x_1,x_2)+W_1(x_1)W_1(x_2)\right)+x_1x_2W_1(x_2)-x_2^2W_1(x_1)\Big].
\end{multline}

As with equation \eqref{eq4.30}, equation \eqref{eq4.34} is unsolvable, but substituting in the large $N$ expansion \eqref{eq1.10} results in a set of solvable loop equations on the correlator expansion coefficients. Thus, we have for each $l\geq0$, the $(2,l)$ loop equation used to compute the expansion coefficient $W_2^l(x_1,x_2)$,
\begin{align}
0&=x_1\Big[W_4^{l-4}(x_1,x_1,x_2,x_2)+3\sum_{l_1+l_2=l-2}W_3^{l_1}(x_1,x_1,x_2)W_1^{l_2}(x_1) \nonumber
\\&\hspace{6em}+3\sum_{l_1+l_2=l-2}W_2^{l_1}(x_1,x_1)W_2^{l_2}(x_1,x_2)+3\sum_{l_1+l_2+l_3=l}W_2^{l_1}(x_1,x_2)W_1^{l_2}(x_1)W_1^{l_3}(x_1)\Big] \nonumber
\\&\quad-W_3^{l-3}(x_1,x_1,x_2)-2\sum_{l_1+l_2=l-1}W_2^{l_1}(x_1,x_2)W_1^{l_2}(x_1)+x_1W_2^l(x_1,x_2) \nonumber
\\&\quad+\frac{1}{2x_1}\left(x_1^2\frac{\partial^2}{\partial x_1^2}+x_1\frac{\partial}{\partial x_1}-1\right)W_2^{l-2}(x_1,x_2) \nonumber
\\&\quad+\frac{2}{x_1}\frac{\partial}{\partial x_2}\frac{1}{x_1^2-x_2^2}\Big[2x_1x_2^2W_2^{l-2}(x_1,x_1)-x_1x_2^2W_2^{l-2}(x_2,x_2)-x_1^2x_2W_2^{l-2}(x_1,x_2) \nonumber
\\&\hspace{10em}+x_1x_2^2\sum_{l_1+l_2=l}\left(2W_1^{l_1}(x_1)W_1^{l_2}(x_1)-W_1^{l_1}(x_2)W_1^{l_2}(x_2)\right)+x_1x_2W_1^{l-1}(x_2) \nonumber
\\&\hspace{12em}-x_1^2x_2\sum_{l_1+l_2=l}W_1^{l_1}(x_1)W_1^{l_2}(x_2)-x_2^2W_1^{l-1}(x_1)\Big]. \label{eq4.35}
\end{align}
Setting $l=0$ then specifies $W_2^0(x_1,x_2)$ as a rational function of $W_1^0$ and its derivative:
\begin{multline} \label{eq4.36}
W_2^0(x_1,x_2)=\frac{2}{x_1}\frac{1}{3(W_1^0(x_1))^2+1}\frac{\partial}{\partial x_2}\frac{x_2}{x_2^2-x_1^2}\Big[2x_2(W_1^0(x_1))^2-x_2(W_1^0(x_2))^2
\\-x_1W_1^0(x_1)W_1^0(x_2)\Big].
\end{multline}
We solve this equation using the aforementioned rational parametrisation in \S\ref{s4.3}.

Let us now highlight that in obtaining equation \eqref{eq4.34} from equation \eqref{eq4.33}, we have simply subtracted off terms from the first three lines of equation \eqref{eq4.33} that contain a factor of $U_1(x_2)$ while using equation \eqref{eq4.30} to argue that this does not change the left-hand side. This idea extends to the general $n\geq3$ case:
\begin{proposition} \label{prop7}
Let us retain the notation used in Lemma \ref{Lemma7}. For $n\geq3$, the connected correlators $W_n(x_1,\ldots,x_n)$ of the $(N,N)$ antisymmetrised Laguerre ensemble satisfy the loop equation
\begin{align}
0&=\left(x_1\sum_{\mu\vdash(x_1,x_1,x_1)}-\sum_{\mu\vdash(x_1,x_1)}\right)\sum_{\sqcup_{t=1}^{\#\mu}K_t=J_n}\prod_{\mu_t\in\mu}W_{\#K_t+\#\mu_t}(\mu_t,K_t) \nonumber
\\&\quad+N^2x_1W_n(x_1,J_n)+\frac{1}{2x_1}\left(x_1^2\frac{\partial^2}{\partial x_1^2}+x_1\frac{\partial}{\partial x_1}-1\right)W_n(x_1,J_n) \nonumber
\\&\quad+\frac{2}{x_1}\sum_{i=2}^n\frac{\partial}{\partial x_i}\frac{1}{x_1^2-x_i^2}\Bigg[\bigg(2x_1x_i^2\sum_{\mu\vdash(x_1,x_1)}-x_1x_i^2\sum_{\mu\vdash(x_i,x_i)}-x_1^2x_i\sum_{\mu\vdash(x_1,x_i)}\bigg) \nonumber
\\&\hspace{16em}\times\sum_{\sqcup_{t=1}^{\#\mu}K_t=J_n\setminus\{x_i\}}\prod_{\mu_t\in\mu}W_{\#K_t+\#\mu_t}(\mu_t,K_t) \nonumber
\\&\hspace{20em}+x_1x_iW_{n-1}(J_n)-x_i^2W_{n-1}(x_1,J_n\setminus\{x_i\})\Bigg] \nonumber
\\&\quad+8\sum_{2\leq i<j\leq n}\frac{1}{x_1}\frac{\partial^2}{\partial x_i\partial x_j}\Bigg\{\frac{x_1x_i^2x_jW_{n-2}(J_n\setminus\{x_i\})}{(x_1^2-x_j^2)(x_i^2-x_j^2)}-\frac{x_1x_ix_j^2W_{n-2}(J_n\setminus\{x_j\})}{(x_1^2-x_i^2)(x_i^2-x_j^2)} \nonumber
\\&\hspace{20em}+\frac{x_i^2x_j^2W_{n-2}(x_1,J_n\setminus\{x_i,x_j\})}{(x_1^2-x_i^2)(x_1^2-x_j^2)}\Bigg\}. \label{eq4.37}
\end{align}
\end{proposition}
\begin{proof}
We give a modification of the inductive argument detailed in \citep[App.~A]{FRW17} (see also \citep{DF20}). To proceed, let us define $E_1(x_1)$ to be the right-hand side of the $n=1$ loop equation \eqref{eq4.30}, $E_2(x_1,x_2)$ to be the right-hand side of the $n=2$ loop equation \eqref{eq4.34}, and for $n\geq3$, define $E_n(x_1,\ldots,x_n)$ to be the right-hand side of the alleged loop equation \eqref{eq4.37} above. Our induction hypothesis is that $E_m(x_1,\ldots,x_m)=0$ for all $1\leq m\leq n$ --- this has already been shown through equations \eqref{eq4.30}, \eqref{eq4.34} to be true for the base cases $m=1,2$. In a similar manner to the derivation of equation \eqref{eq4.33}, use Lemma \ref{Lemma7} to write the $U_n$ loop equation \eqref{eq4.26} in terms of the $W_n$. Then, check term by term that the right-hand side of this equation agrees with the sum
\begin{equation*}
\sum_{K_1\sqcup K_2=J_n}U_{\#K_1}(K_1)E_{\#K_2+1}(x_1,K_2)=E_n(x_1,J_n)+\sum_{\substack{K_1\sqcup K_2=J_n,\\K_1\neq\emptyset}}U_{\#K_1}(K_1)E_{\#K_2+1}(x_1,K_2).
\end{equation*}
Since the left-hand side of the equation obtained from using Lemma \ref{Lemma7} to rewrite the $U_n$ loop equation \eqref{eq4.26} is zero, the above sum must also vanish. Hence, we have
\begin{equation*}
0=E_n(x_1,J_n)+\sum_{\substack{K_1\sqcup K_2=J_n,\\K_1\neq\emptyset}}U_{\#K_1}(K_1)E_{\#K_2+1}(x_1,K_2)=E_n(x_1,\ldots,x_n),
\end{equation*}
where the second equality follows from our induction hypothesis. As $E_n(x_1,\ldots,x_n)$ is the right-hand side of equation \eqref{eq4.37}, the proposition is proved.
\end{proof}

Like with the loop equations given earlier in this subsection, to extract information from the above loop equation, one must substitute the large $N$ expansion \eqref{eq1.10} into said equation and collect terms of equal order in $N$ to obtain the $(n,l)$ loop equations on the correlator expansion coefficients $W_n^l$ --- it is the $(n,l)$ loop equations that can be solved. We do not display the $W_n^l$ loop equations here since it is straightforward to extract them from equation \eqref{eq4.27}: one need only replace products of the form $W_{n_1}(x_1,\ldots,x_{n_1})\cdots W_{n_k}(x_1,\ldots,x_{n_k})$ with sums of the form
\begin{equation*}
\sum_{l_1+\cdots+l_k=l-r}W_{n_1}^{l_1}(x_1,\ldots,x_{n_1})\cdots W_{n_k}^{l_k}(x_1,\ldots,x_{n_k}),
\end{equation*}
where $r$ depends on what power of $N$ the original product is multiplied by; we refer to the derivations of equations \eqref{eq4.31} and \eqref{eq4.35} for examples.

We make some brief comparisons to loop equations seen in the existing literature. First, for $n\ge3$, the loop equation \eqref{eq4.37} is higher order than the equivalent loop equations characterising the classical $\beta$ ensembles. By this, we mean that the first line of equation \eqref{eq4.37} involves partitions of $(x_1,x_1,x_1)$ and the term $W_{n+2}(x_1,x_1,x_1,J_n)$, whereas these do not appear in the equivalent equations for the classical $\beta$ ensembles. Similarly for the sums over $\mu\vdash(x_1,x_1),(x_i,x_i)$ seen in the third line. Moving past comparisons to loop equations for classical $\beta$ ensembles, let us recall that even though the above observations show that the loop equation \eqref{eq4.37} is structurally quite different from most of the loop equations studied in the literature, there still exist some works studying such higher order loop equations. Namely, multi-matrix models and chains of matrices have been shown to be characterised by such higher order loop equations \citep{EO09}, \citep{Ora15}, while loop equations involving large sums of the form $\sum_{\mu\vdash(x_1,\ldots,x_1)}$ have been studied in the abstract setting for their own sake \citep{BE13}, \citep{BHLMR14}. We also remind the reader that loop equations for the $(N,N,N)$ complex Wishart product ensemble were derived in \citep{DF20}, where they were shown to be of the same order as equation \eqref{eq4.37}. Finally, we highlight the fact that the antisymmetric nature of the matrix $B=\mathcal{J}_1$ shows itself in the loop equation \eqref{eq4.37} in a way that has not been seen in previously studied loop equations: There are denominators of the form $x_i^2-x_j^2$ throughout this equation which are unchanged when mapping $x_i\mapsto -x_i$ or $x_j\mapsto-x_j$, while previously studied loop equations have instead exhibited denominators of the form $x_i-x_j$. It can be surmised from the computations of Appendix \ref{appendixC.1} that these differences of squares arise from the fact that the mixed moments $m_{k_1,\ldots,k_n}$ vanish whenever any of the $k_i$ are odd.

\subsection{Discussion on \texorpdfstring{$W_1^{(\mathcal{J}_1),0}$}{W1(J1),0} and \texorpdfstring{$W_2^{(\mathcal{J}_1),0}$}{W2(J1),0}} \label{s4.3}
Let us now focus our discussion on the $(1,0)$ and $(2,0)$ loop equations \eqref{eq4.32}, \eqref{eq4.36} so that we may comment on the large $N$ limiting behaviour of the $(N,N)$ antisymmetrised Laguerre ensemble. We commence our discussion by noting that $W_1^0(x_1)$ has the large $x_1$ expansion
\begin{equation*}
W_1^0(x_1)=W_1^{(\mathcal{J}_1),0}(x_1)=\sum_{k=0}^{\infty}\frac{m_k^{(\mathcal{J}_1),0}}{x_1^{k+1}},
\end{equation*}
where we define $m_k^{(\mathcal{J}_1),0}:=\lim_{N\to\infty} m_k^{(\mathcal{J}_1)}$. Letting $m_k^{(\mathrm{i}\mathcal{J}_1),0}$ denote the analogous limiting moments of the matrix $\mathrm{i}\mathcal{J}_1$, we have that
\begin{equation} \label{eq4.38}
m_k^{(\mathcal{J}_1),0}=\lim_{N\to\infty}\mean{\Tr\mathcal{J}_1^k}=(-\mathrm{i})^k\lim_{N\to\infty}\mean{\Tr (\mathrm{i}\mathcal{J}_1)^k}=(-\mathrm{i})^km_k^{(\mathrm{i}\mathcal{J}_1),0}.
\end{equation}

It is known from \citep[Cor.~1.2]{FILZ19} that for $\mathcal{J}_m$ drawn from the $(N_0,N_1,\ldots,N_m)$ antisymmetrised matrix product ensemble (recall from Definition \ref{def2} that each of $N_0,N_1,\ldots,N_m$ are even integers), the positive eigenvalues $\lambda_1,\ldots,\lambda_{N_0/2}$ of $\mathrm{i}\mathcal{J}_m$ are such that $\lambda_j'=\lambda_j^2/4^m$ ($j\in[N_0/2]$) are statistically equivalent to the eigenvalues of the $(\tfrac{N_0}{2},\tfrac{N_1}{2},\tfrac{N_1-1}{2},\ldots,\tfrac{N_m}{2},\tfrac{N_m-1}{2})$ complex Wishart product ensemble, with the caveat that the latter be generalised using induced Ginibre matrices \citep{FBKSZ12}, \citep{AIK13}, \citep{IK14} to allow for half-integer dimensions $\tfrac{N_1-1}{2},\ldots,\tfrac{N_m-1}{2}$. Hence, by a result of \citep{PZ11} (see also \citep{FL15} for our specific $m=1$ case) relating the limiting eigenvalue density of the complex Wishart product ensemble to the Fuss--Catalan distribution \citep[App.~A]{BPB99}, \citep{PZ11}, \citep{Neu14}, \citep{Ips15}, we have that
\begin{equation} \label{eq4.39}
m_{2k}^{(\mathcal{J}_m),0}=(-1)^km_k^{(\mathrm{FC}_{2m+1})},
\end{equation}
where $m_{2k}^{(\mathcal{J}_m),0}$ is defined in analogy with equation \eqref{eq4.38} and
\begin{equation} \label{eq4.40}
m_k^{(\mathrm{FC}_d)}=\frac{1}{(d-1)k+1}\binom{dk}{k}
\end{equation}
is the order $d$ Fuss--Catalan number. In the present $m=1$ setting, one can check using computer algebra that the solution of equation \eqref{eq4.32} has large $x_1$ expansion
\begin{equation} \label{eq4.41}
W_1^0(x_1)=\frac{1}{x_1}-\frac{1}{x_1^3}+\frac{3}{x_1^5}-\frac{12}{x_1^7}+\frac{55}{x_1^9}-\frac{273}{x_1^{11}}+{\rm O}(x_1^{-13});
\end{equation}
observe that for integer $k$, the coefficient of $x_1^{-2k-1}$ is precisely $(-1)^km_k^{(\mathrm{FC}_3)}$, while there are no terms of the form $x_1^{-2k}$.

Rewriting the above equivalence of moments and Fuss--Catalan numbers in terms of generating functions shows that
\begin{equation} \label{eq4.42}
W_1^{(\mathcal{J}_1),0}(x_1)=-x_1W_1^{(\mathrm{FC}_3)}(-x_1^2),
\end{equation}
where $W_1^{(\mathrm{FC}_3)}(x):=\sum_{k=0}^\infty m_k^{(\mathrm{FC}_3)}/x^{k+1}$. In combination with equation \eqref{eq4.32}, this suggests that upon writing $u=-x_1^2$,
\begin{equation} \label{eq4.43}
0=u^2\left(W_1^{(\mathrm{FC}_3)}(u)\right)^3-uW_1^{(\mathrm{FC}_3)}(u)+1.
\end{equation}
This is known to be true \citep{FL15}, and thus serves as confirmation that the loop equation analysis presented in \S\ref{s4.2} has produced the correct spectral curve for the $(N,N)$ antisymmetrised Laguerre ensemble.

We now move on to the rational parametrisation of our cubic spectral curve \eqref{eq4.32}. Though one might expect this to be more complicated than for classical $\beta$ ensembles, where there is a standard procedure involving an application of the Joukowsky transform to the quadratic spectral curve, we are fortuitous in that our spectral curve is a linear polynomial in $x_1$. Thus, rearranging equation \eqref{eq4.32} to give
\begin{equation*}
x_1=\frac{1}{\left(W_1^0(x_1)\right)^3+W_1^0(x_1)}
\end{equation*}
manifests a natural parametrisation: Simply define
\begin{equation} \label{eq4.44}
x(z)=\frac{1}{z^3+z},\quad y(z)=z.
\end{equation}
Using the Newton polygon method (see \citep{BP00} and references therein) on the complex algebraic curve $0=xy^3+xy-1$ shows that our spectral curve is the genus zero Riemann sphere with viable coordinates $x(z),y(z)$ as given above. Hence, $W_1^0(x(z_1))$ is a single-valued function of $z_1\in\mathbb{CP}^1$ whose value agrees with $W_1^0(x_1)$ when $z_1\in x^{-1}(x_1)$ is chosen to be close to $z_1=0$ (this ensures that we choose the solution of equation \eqref{eq4.32} obeying the asymptotic requirement that $W_1^0(x_1)\sim1/x_1$ as $x_1\to\infty$).

Now, taking $z_1,z_2\in\mathbb{CP}^1$ and substituting $x_1=x(z_1)$ and $x_2=x(z_2)$ \eqref{eq4.44} into equation \eqref{eq4.36} shows that
\begin{align}
W_2^0(x(z_1),x(z_2))&=\frac{1}{x'(z_1)x'(z_2)}\frac{1}{(z_1-z_2)^2}-\frac{1}{(x(z_1)-x(z_2))^2} \nonumber
\\&\quad-\frac{1}{x'(z_1)x'(z_2)}\frac{1}{(z_1+z_2)^2}+\frac{1}{(x(z_1)+x(z_2))^2}. \label{eq4.45}
\end{align}
The first line of the right-hand side is the universal form of $W_2^0(x(z_1),x(z_2))$ obtained in the case of the one-cut 1-Hermitian matrix model \citep{AJM90}, \citep{ACKM93} and classical $\beta$ ensembles \citep{WF14}, \citep{FRW17}; the first term is referred to as the Bergman kernel or fundamental differential of the second kind (see, e.g., \citep{EO09}, \citep{EKR18}). The second line of equation \eqref{eq4.45} is a novel term which ensures that $W_2^0(x(z_1),x(z_2))$ is antisymmetric in $z_1,z_2$. Thus, we have a double pole at $z_1=-z_2$ in addition to the standard double pole at $z_1=z_2$ (cf.~the discussion on antisymmetric denominators at the end of \S\ref{s4.2}). 

It is a simple matter to determine $W_1^1(x_1)$ as a function of $z_1$: Set $l=1$ in equation \eqref{eq4.31} with $x_1$ replaced by $x(z_1)$ and $W_1^0(x_1)$ replaced by $y(z_1)=z_1$ \eqref{eq4.44} to see that
\begin{equation} \label{eq4.46}
W_1^1(x(z_1))=\frac{z_1^2}{x(z_1)(3z_1^2+1)}=\frac{z_1^5+z_1^3}{3z_1^2+1}.
\end{equation}
Though we do not do so here, one may proceed in this manner to compute all of the $W_n^l$ since the loop equations are closed and form a triangular recursive system.

From Proposition \ref{prop3}, we know that the mixed cumulants $c_{k_1,\ldots,k_n}^{(\mathcal{J}_1)}$ have a genus expansion of the form
\begin{equation} \label{eq4.47}
c_{k_1,\ldots,k_n}^{(\mathcal{J}_1)}=N^{2-n}\sum_{l=0}^{\overline{k}_n-n}\frac{c_{k_1,\ldots,k_n}^{(\mathcal{J}_1),l}}{N^l}.
\end{equation}
Combining this with the large $N$ expansion \eqref{eq1.10} and large $x_1,\ldots,x_n$ expansion \eqref{eq1.9} of the connected correlators $W_n(x_1,\ldots,x_n)$ then shows that
\begin{equation*}
W_n^l(x_1,\ldots,x_n)=\sum_{k_1,\ldots,k_n=0}^{\infty}\frac{c_{k_1,\ldots,k_n}^{(\mathcal{J}_1),l}}{x_1^{k_1+1}\cdots x_n^{k_n+1}}.
\end{equation*}
Hence, noting that these interpretations as generating functions are valid only in the large $x_1,\ldots,x_n$ regime, we have that the coefficients of the cumulant genus expansions can be extracted from the correlator expansion coefficients using the residue formula
\begin{equation} \label{eq4.48}
c_{k_1,\ldots,k_n}^{(\mathcal{J}_1),l}=(-1)^n\underset{x_1,\ldots,x_n=\infty}{\mathrm{Res}}x_1^{k_1}\cdots x_n^{k_n}W_n^l(x_1,\ldots,x_n)\,\mathrm{d}x_1\cdots\mathrm{d}x_n.
\end{equation}
As our formalism produces expressions for the correlator expansion coefficients in terms of the variables $z_1,\ldots,z_n\in\mathbb{CP}^1$ through the map \eqref{eq4.44}, this residue formula is only useful to us upon reformulating it in terms of the $z_1,\ldots,z_n$ variables, as well. Thus, observe that $x_i=x(z_i)=\infty$ corresponds to $z_i=0,\pm\mathrm{i}$. Of these three solutions, we are interested in $z_i=0$, as it coincides with the requirement that $y(z_i)=z_i\to0$ as $x(z_i)\to\infty$, which follows from the fact that the branch of $W_1^0(x_1)$ that behaves as a moment generating function is the one that exhibits the asymptotic $W_1^0(x_1)\sim1/x_1$ in the $x_1\to\infty$ limit. Noting also that setting $x_i=x(z_i)$ means that $\mathrm{d}x_i=x'(z_i)\mathrm{d}z_i$, equation \eqref{eq4.48} therefore implies that
\begin{equation} \label{eq4.49}
c_{k_1,\ldots,k_n}^{(\mathcal{J}_1),l}=(-1)^n\underset{z_1,\ldots,z_n=0}{\mathrm{Res}}W_n^l(x(z_1),\ldots,x(z_n))\prod_{i=1}^nx(z_i)^{k_i}x'(z_i)\,\mathrm{d}z_i.
\end{equation}
For example, inserting $W_1^0(x(z_1))=z_1$ into the above shows that $c_4^{(\mathcal{J}_1),0}=3$, which is exactly the coefficient of $1/x_1^5$ in the large $x_1$ expansion \eqref{eq4.41}. Likewise, substituting equations \eqref{eq4.45} and \eqref{eq4.46} into the formula \eqref{eq4.49} shows that $c_2^{(\mathcal{J}_1),1}=1$ and $c_{2,2}^{(\mathcal{J}_1),0}=12$ --- see Appendix \ref{appendixD.1} for complementary data. Finally, we remark that applying residue calculus of this form to the loop equations themselves is precisely the path to reducing said loop equations to simpler topological recursion formulae (simpler in that the irrelevance of terms with vanishing residue is made clear). We expect that doing so would lead to topological recursion formulae similar to those of \citep{BE13}, \citep{BHLMR14}, but we leave this endeavour to future work.

\setcounter{equation}{0}
\section{Loop equations for the Hermitised Laguerre ensemble} \label{s5}
We now repeat the exercise of Section \ref{s4} for the $(N,N)$ Hermitised Laguerre ensemble. Our discussion will be relatively succinct, providing technical details only when they differ to those shown earlier. We now take $X,H$ to be as in Definition \ref{def1} and redefine $B=\mathcal{H}_1=X^\dagger HX$. Moreover, we continue to use the Einstein summation convince and retain the notations $J_n=(x_2,\ldots,x_n)$, $I_n=(k_2,\ldots,k_n)$, and $\Tr B^{I_n}=\Tr B^{k_2}\cdots\Tr B^{k_n}$. Within this section, we do not display the implicit superscript $(\mathcal{H}_1)$ for the quantities $m_{k_1,\ldots,k_n}$, $c_{k_1,\ldots,k_n}$, and their respective generating functions $U_n(x_1,\ldots,x_n)$, $W_n(x_1,\ldots,x_n)$. By Proposition \ref{prop1}, $W_n(x_1,\ldots,x_n)$ admits the genus expansion \eqref{eq1.10} with $W_n^l(x_1,\ldots,x_n)=0$ when $l$ is odd.

Emulating the key idea of Section \ref{s4}, we begin by considering the average
\begin{equation} \label{eq5.1}
\mean{\left(\partial_{(X^\dagger)_{ab}}+NX_{ba}\right)\left(\partial_{(H^\dagger)_{bc}}+NH_{cb}\right)\left(\partial_{X_{cd}}+N(X^\dagger)_{dc}\right)(B^{k_1})_{ad}\Tr B^{I_n}}.
\end{equation}
This choice of starting point is informed by the following computations:
\begin{align}
\partial_{(X^\dagger)_{ab}}\exp\left(-N\Tr(X^\dagger X)\right)&=-N\left\{\partial_{(X^\dagger)_{ab}}(X^\dagger)_{ef}X_{fe}\right\}\exp\left(-N\Tr(X^\dagger X)\right) \nonumber
\\&=-N\Big[\left\{\partial_{(X^\dagger)_{ab}}(X^\dagger)_{ef}\right\}X_{fe}+(X^\dagger)_{ef}\left\{\partial_{(X^\dagger)_{ab}}X_{fe}\right\}\Big] \nonumber
\\&\qquad\times\exp\left(-N\Tr(X^\dagger X)\right) \nonumber
\\&=-N\left[\chi_{e=a,f=b}X_{fe}+0\right]\exp\left(-N\Tr(X^\dagger X)\right) \nonumber
\\&=-NX_{ba}\exp\left(-N\Tr(X^\dagger X)\right), \label{eq5.2}
\\\partial_{(H^\dagger)_{bc}}\exp\left(-\tfrac{N}{2}\Tr(H^2)\right)&=-\frac{N}{2}\Big[\left\{\partial_{(H^\dagger)_{bc}}H_{ef}\right\}H_{fe}+H_{ef}\left\{\partial_{(H^\dagger)_{bc}}H_{fe}\right\}\Big] \nonumber
\\&\qquad\times\exp\left(-\tfrac{N}{2}\Tr(H^2)\right) \nonumber
\\&=-NH_{cb}\exp\left(-\tfrac{N}{2}\Tr(H^2)\right), \label{eq5.3}
\\ \partial_{X_{cd}}\exp\left(-N\Tr(X^\dagger X)\right)&=-N(X^\dagger)_{dc}\exp\left(-N\Tr(X^\dagger X)\right). \label{eq5.4}
\end{align}
The first two equalities in the above follow from the chain and Leibniz product rules and the third equality is due to the fact that, as a complex variable, $X_{fe}$ is independent of $(X^\dagger)_{ab}=\overline{X_{ba}}$ for any choice of $e,f$. Similar reasoning yields equations \eqref{eq5.3} and \eqref{eq5.4} upon recalling that $H$ being Hermitian implies that $H^\dagger=H$. 

As in Section \ref{s4}, we proceed by computing the average \eqref{eq5.1} in two different ways, thereby deriving loop equations on the mixed moments $m_{k_1,\ldots,k_n}$ and, consequently, the unconnected correlators $U_n(x_1,\ldots,x_n)$.

\subsection{Loop equations on the \texorpdfstring{$U_n^{(\mathcal{H}_1)}$}{Un(H1)}} \label{s5.1}
Since $(X^\dagger)_{ab}=\overline{X_{ba}}$ and $X_{cd}$ are independent of each other for all choices of $a,b,c,d$, the analogue of Lemma \ref{Lemma5} relevant to the present setting is comparatively simple:
\begin{lemma} \label{Lemma8}
Having set $B=X^\dagger HX$ as above, the partial derivatives of $(B^{k_1})_{ef}$ with respect to $(X^\dagger)_{ab}$, $ H_{bc}$, and $X_{cd}$ are respectively given by
\begin{align*}
\partial_{(X^\dagger)_{ab}}(B^{k_1})_{ef}&=\sum_{p_1+p_2=k_1-1}(B^{p_1})_{ea}( HXB^{p_2})_{bf},
\\ \partial_{ H_{bc}}(B^{k_1})_{ef}&=\sum_{p_1+p_2=k_1-1}(B^{p_1}X^\dagger)_{eb}(XB^{p_2})_{cf},
\\ \partial_{X_{cd}}(B^{k_1})_{ef}&=\sum_{p_1+p_2=k_1-1}(B^{p_1}X^\dagger H)_{ec}(B^{p_2})_{df}.
\end{align*}
Similar to Lemma \ref{Lemma5}, these sums are empty when $k_1=0$ and are otherwise taken over the range $p_1=0,1,\ldots,k_1-1$. The indices $e,f$ are arbitrary and the above equations hold true upon setting $e,f$ equal to each other or to one of $a,b,c,d$.
\end{lemma}

Moving on, the analogue of Proposition \ref{prop5} corresponding to the $(N,N)$ Hermitised Laguerre ensemble is as follows:
\begin{proposition} \label{prop8}
The average \eqref{eq5.1} simplifies to both the left- and right-hand sides of the equation
\begin{multline} \label{eq5.5}
\mean{\left\{\partial_{(X^\dagger)_{ab}}\partial_{( H^\dagger)_{bc}}\partial_{X_{cd}}(B^{k_1})_{ad}\Tr B^{I_n}\right\}}
\\=N^3m_{k_1+1,I_n}-k_1\sum_{p_1+p_2=k_1-1}m_{p_1,p_2,I_n}-2\sum_{2\leq i<j\leq n}k_ik_jm_{k_1,k_i+k_j-1,I_n\setminus\{k_i,k_j\}}
\\-\sum_{i=2}^nk_i\left[2k_1m_{k_1+k_i-1,I_n\setminus\{k_i\}}+\sum_{p_1+p_2=k_i-1}m_{k_1,p_1,p_2,I_n\setminus\{k_i\}}\right].
\end{multline}
\end{proposition}
\begin{proof}
Expanding the first two brackets in the argument of the average \eqref{eq5.1} and using the independence of $ H$ from $X^\dagger,X$ shows that
\begin{multline*}
\mean{\left(\partial_{(X^\dagger)_{ab}}+NX_{ba}\right)\left(\partial_{( H^\dagger)_{bc}}+N H_{cb}\right)\left(\partial_{X_{cd}}+N(X^\dagger)_{dc}\right)(B^{k_1})_{ad}\Tr B^{I_n}}
\\=\mean{\partial_{(X^\dagger)_{ab}}\left(\partial_{( H^\dagger)_{bc}}+N H_{cb}\right)\left(\partial_{X_{cd}}+N(X^\dagger)_{dc}\right)(B^{k_1})_{ad}\Tr B^{I_n}}
\\+N\mean{\partial_{( H^\dagger)_{bc}}X_{ba}\left(\partial_{X_{cd}}+N(X^\dagger)_{dc}\right)(B^{k_1})_{ad}\Tr B^{I_n}}
\\+N^2\mean{X_{ba} H_{cb}\left(\partial_{X_{cd}}+N(X^\dagger)_{dc}\right)(B^{k_1})_{ad}\Tr B^{I_n}}.
\end{multline*}
The first two averages on the right-hand side vanish by the fundamental theorem of calculus, while the third average can be simplified using integration by parts:
\begin{align}
N^2&\mean{X_{ba} H_{cb}\left(\partial_{X_{cd}}+N(X^\dagger)_{dc}\right)(B^{k_1})_{ad}\Tr B^{I_n}} \nonumber
\\&=N^3\mean{\Tr B^{k_1+1}\Tr B^{I_n}}+N^2\mean{X_{ba} H_{cb}\partial_{X_{cd}}(B^{k_1})_{ad}\Tr B^{I_n}} \nonumber
\\&=N^3\mean{\Tr B^{k_1+1}\Tr B^{I_n}}+N^2\mean{\partial_{X_{cd}}X_{ba} H_{cb}(B^{k_1})_{ad}\Tr B^{I_n}} \nonumber
\\&\quad-N^2\mean{\left\{\partial_{X_{cd}}X_{ba}\right\} H_{cb}(B^{k_1})_{ad}\Tr B^{I_n}} \nonumber
\\&=N^3\mean{\Tr B^{k_1+1}\Tr B^{I_n}}-N^2\mean{\chi_{c=b,d=a} H_{cb}(B^{k_1})_{ad}\Tr B^{I_n}} \nonumber
\\&=N^3m_{k_1+1,I_n}-N^2\mean{\Tr  H\Tr B^{k_1}\Tr B^{I_n}}; \label{eq5.6}
\end{align}
the average $\mean{\partial_{X_{cd}}X_{ba} H_{cb}(B^{k_1})_{ad}\Tr B^{I_n}}$ in the third line vanishes due to the fundamental theorem of calculus, yet again. Now, the average $\mean{\Tr  H\Tr B^{k_1}\Tr B^{I_n}}$ poses a challenge not seen in the proof of Proposition \ref{prop5}, wherein the analogous average $\mean{\Tr J\Tr B^{k_1}\Tr B^{I_n}}$ was zero due to the antisymmetry of $J$. Here, we must derive auxiliary identities to deal with this average. The fundamental theorem of calculus (left-hand side) and the Leibniz product rule combined with Lemma \ref{Lemma8} (right-hand side) shows that
\begin{multline} \label{eq5.7}
0=\mean{\partial_{( H^\dagger)_{aa}}\Tr B^{k_1}\Tr B^{I_n}}
\\=k_1\mean{\Tr(B^{k_1-1}X^\dagger X)\Tr B^{I_n}}+\sum_{i=2}^nk_i\mean{\Tr(B^{k_i-1}X^\dagger X)\Tr B^{k_1}\Tr B^{I_n\setminus\{k_i\}}}
\\-N\mean{\Tr  H\Tr B^{k_1}\Tr B^{I_n}}.
\end{multline}
The averages involving traces of the form $\Tr(B^{k_j-1}X^\dagger X)$ can be expressed in terms of the mixed moments by similarly showing that for $i=1,\ldots,n$,
\begin{multline*}
0=\mean{\partial_{(X^\dagger)_{ab}}(B^{k_i-1}X^\dagger)_{ab}\Tr B^{\{k_1\}\cup I_n\setminus\{k_i\}}}
\\=-N\mean{\Tr(B^{k_i-1}X^\dagger X)\Tr B^{\{k_1\}\cup I_n\setminus\{k_i\}}}+\sum_{p_1+p_2=k_i-1}m_{p_1,p_2,\{k_1\}\cup I_n\setminus\{k_i\}}
\\+\sum_{\substack{1\leq j\leq n,\\j\neq i}}k_jm_{k_i+k_j-1,\{k_1\}\cup I_n\setminus\{k_i,k_j\}}.
\end{multline*}
Setting $i=1,\ldots,n$ into this equation and substituting the resulting set of identities into equation \eqref{eq5.7} multiplied by $N$ shows that
\begin{align*}
-N^2&\mean{\Tr  H\Tr B^{k_1}\Tr B^{I_n}}
\\&=-Nk_1\mean{\Tr(B^{k_1-1}X^\dagger X)\Tr B^{I_n}}-N\sum_{i=2}^nk_i\mean{\Tr(B^{k_i-1}X^\dagger X)\Tr B^{k_1}\Tr B^{I_n\setminus\{k_i\}}}
\\&=-k_1\sum_{p_1+p_2=k_1-1}m_{p_1,p_2,I_n}-k_1\sum_{j=2}^nk_jm_{k_1+k_j-1,I_n\setminus\{k_j\}}-\sum_{i=2}^nk_i\sum_{p_1+p_2=k_i-1}m_{k_1,p_1,p_2,I_n\setminus\{k_i\}}
\\&\qquad-\sum_{i=2}^nk_i\left(k_1m_{k_1+k_i-1,I_n\setminus\{k_i\}}+\sum_{\substack{2\leq j\leq n,\\j\neq i}}k_jm_{k_1,k_i+k_j-1,I_n\setminus\{k_i,k_j\}}\right).
\end{align*}
Substituting this result into equation \eqref{eq5.6} then gives the right-hand side of equation \eqref{eq5.5}, as required.

To show that the average \eqref{eq5.1} simplifies to the left-hand side of equation \eqref{eq5.5}, one need only repeat the relevant arguments given in the proof of Proposition \ref{prop5}, with the identities \eqref{eq5.2}--\eqref{eq5.4} in mind.
\end{proof}

The next step in obtaining loop equations on the mixed moments is to expand the right-hand side of equation \eqref{eq5.5} using the Leibniz product rule. Letting
\begin{equation*}
f_i(B)=\begin{cases}(B^{k_1})_{ad},&i=1,\\ \Tr B^{k_i},&i=2,\ldots,n\end{cases}
\end{equation*}
and writing $\mathcal{A}_{i,j,h}$ for the average obtained by replacing $f_i(B)$ by $\{\partial_{(X^\dagger)_{ab}}f_i(B)\}$, then $f_j(B)$ by $\{\partial_{( H^\dagger)_{bc}}f_j(B)\}$, and finally $f_h(B)$ by $\{\partial_{X_{cd}}f_h(B)\}$ in $\langle f_1(B)\cdots f_n(B)\rangle$ so that, e.g.,
\begin{equation*}
\mathcal{A}_{3,3,1}=\mean{\left\{\partial_{(X^\dagger)_{ab}}\partial_{( H^\dagger)_{bc}}\Tr B^{k_3}\right\}\left\{\partial_{X_{cd}}(B^{k_1})_{ad}\right\}\Tr B^{I_n\setminus\{k_3\}}},
\end{equation*}
we have that the left-hand side of equation \eqref{eq5.5} is given by
\begin{multline} \label{eq5.8}
\mean{\left\{\partial_{(X^\dagger)_{ab}}\partial_{( H^\dagger)_{bc}}\partial_{X_{cd}}(B^{k_1})_{ad}\Tr B^{I_n}\right\}}=\sum_{i,j,h=1}^n\mathcal{A}_{i,j,h}
\\\hspace{-2em}=\mathcal{A}_{1,1,1}+\sum_{i=2}^n\left(\mathcal{A}_{1,1,i}+\mathcal{A}_{1,i,1}+\mathcal{A}_{i,1,1}+\mathcal{A}_{1,i,i}+\mathcal{A}_{i,1,i}+\mathcal{A}_{i,i,1}+\mathcal{A}_{i,i,i}\right)
\\+\sum_{\substack{2\leq i,j\leq n,\\i\neq j}}\left(\mathcal{A}_{1,i,j}+\mathcal{A}_{i,1,j}+\mathcal{A}_{i,j,1}+\mathcal{A}_{i,i,j}+\mathcal{A}_{i,j,i}+\mathcal{A}_{j,i,i}\right)+\sum_{\substack{2\leq i,j,h\leq n,\\i\neq j\neq h\neq i}}\mathcal{A}_{i,j,h}.
\end{multline}
Every term on the right-hand side of this equation can be expressed in terms of the mixed moments using similar methods to those detailed in the proof of Lemma~\ref{Lemma6}, with Lemma~\ref{Lemma8} supplanting the role of Lemma \ref{Lemma5}. Hence, substituting this equation into equation \eqref{eq5.5} produces the aforementioned loop equation on the mixed moments --- this loop equation and the required expressions for the $\mathcal{A}_{i',j',h'}$ are displayed in Appendix \ref{appendixC.2}. Multiplying both sides of said loop equation by $x_1^{-k_1-1}\cdots x_n^{-k_n-1}$ and summing over $k_1,\ldots,k_n\geq0$ then gives the following loop equation on the unconnected correlators --- this transformation is carried out term by term in Appendix \ref{appendixC.2}, as well.

\begin{proposition} \label{prop9}
Recall that $J_n=(x_2,\ldots,x_n)$, set $U_0:=1$, and define $U_{n'}:=0$ for $n'<0$. Furthermore, define the auxiliary functions (not to be confused with that given in Proposition \ref{prop6})
\begin{multline} \label{eq5.9}
A_i(x_1,J_n)=\frac{3x_1x_iU_{n+1}(x_1,x_1,x_1,J_n\setminus\{x_i\})-x_1x_iU_{n+1}(x_1,x_1,J_n)}{x_1-x_i}
\\-\frac{x_1x_iU_{n+1}(x_1,x_i,J_n)+x_i^2U_{n+1}(x_i,x_i,J_n)}{x_1-x_i}
\\+\left(\frac{3x_1}{2}\frac{\partial^2}{\partial x_1^2}+\frac{\partial}{\partial x_1}+x_i\frac{\partial^2}{\partial x_1\partial x_i}+\frac{x_i}{2x_1}\frac{\partial^2}{\partial x_i^2}x_i\right)
\\ \times\left\{\frac{x_iU_{n-1}(x_1,J_n\setminus\{x_i\})-x_1U_{n-1}(J_n)}{x_1-x_i}\right\},
\end{multline}
\begin{multline} \label{eq5.10}
A_{i,j}(x_1,J_n)=\frac{2x_iU_{n-1}(x_1,J_n\setminus\{x_i\})-2x_jU_{n-1}(x_1,J_n\setminus\{x_j\})}{x_i-x_j}
\\\hspace{-6em}+x_ix_j\Bigg\{\frac{6U_{n-1}(x_1,x_1,J_n\setminus\{x_i,x_j\})+U_{n-1}(J_n)}{(x_1-x_i)(x_1-x_j)}
\\\hspace{4em}-\frac{3U_{n-1}(x_1,J_n\setminus\{x_j\})+3U_{n-1}(x_1,,J_n\setminus\{x_i\})}{(x_1-x_i)(x_1-x_j)}
\\\hspace{4em}-\frac{3U_{n-1}(x_i,J_n\setminus\{x_j\})-2U_{n-1}(x_1,J_n\setminus\{x_i\})}{(x_1-x_i)(x_i-x_j)}
\\+\frac{3U_{n-1}(x_j,J_n\setminus\{x_i\})-2U_{n-1}(x_1,J_n\setminus\{x_j\})}{(x_1-x_j)(x_i-x_j)}\Bigg\},
\end{multline}
\begin{multline} \label{eq5.11}
A_{i,j,h}(x_1,J_n)=\frac{x_ix_jx_hU_{n-3}(x_1,J_n\setminus\{x_i,x_j,x_h\})}{(x_1-x_i)(x_1-x_j)(x_1-x_h)}-\frac{x_1x_jx_hU_{n-3}(J_n\setminus\{x_j,x_h\})}{(x_1-x_i)(x_i-x_j)(x_i-x_h)}
\\+\frac{x_1x_ix_hU_{n-3}(J_n\setminus\{x_i,x_h\})}{(x_1-x_j)(x_i-x_j)(x_j-x_h)}-\frac{x_1x_ix_jU_{n-3}(J_n\setminus\{x_i,x_j\})}{(x_1-x_h)(x_i-x_h)(x_j-x_h)}.
\end{multline}
Then, for $n\geq1$, the unconnected correlators satisfy the loop equation
\begin{multline} \label{eq5.12}
0=x_1^2U_{n+3}(x_1,x_1,x_1,x_1,J_n)-N^3x_1U_n(x_1,J_n)+N^4U_{n-1}(J_n)
\\\hspace{2em}+\lim_{\xi\to x_1}\left(x_1^2\frac{\partial^2}{\partial x_1^2}+2x_1\frac{\partial}{\partial x_1}\right)\left\{U_{n+1}(\xi,x_1,J_n)+\frac{1}{2}U_{n+1}(x_1,x_1,J_n)\right\}+\sum_{i=2}^n\frac{\partial}{\partial x_i}A_i(x_1,J_n)
\\+\sum_{2\leq i<j\leq n}\frac{\partial^2}{\partial x_i\partial x_j}A_{i,j}(x_1,J_n)+\frac{6}{x_1}\sum_{2\leq i<j<h\leq n}\frac{\partial^3}{\partial x_i\partial x_j\partial x_h}A_{i,j,h}(x_1,J_n).
\end{multline}
\end{proposition}

\subsection{Loop equations on the \texorpdfstring{$W_n^{(\mathcal{H}_1)}$}{Wn(H1)} and \texorpdfstring{$W_n^{(\mathcal{H}_1),l}$}{Wn(H1),l}} \label{s5.2}
Setting $n=1$ in equation \eqref{eq5.12} above and using the canonical extension of Lemma \ref{Lemma7} specifying $U_{n+3}(x_1,x_1,x_1,x_1,J_n)$ as a suitable sum over $\mu\vdash(x_1,x_1,x_1,x_1)$ produces the $n=1$ \textit{loop equation on the connected correlators} $W_n(x_1,\ldots,x_n)$,
\begin{multline} \label{eq5.13}
0=x_1^2\sum_{\mu\vdash(x_1,x_1,x_1,x_1)}\prod_{\mu_t\in\mu}W_{\#\mu_t}(\mu_t)-N^3x_1W_1(x_1)+N^4
\\+\lim_{\xi\to x_1}\left(x_1^2\frac{\partial^2}{\partial x_1^2}+2x_1\frac{\partial}{\partial x_1}\right)\left\{W_2(\xi,x_1)+\frac{1}{2}W_2(x_1,x_1)\right.
\\+\left.W_1(\xi)W_1(x_1)+\frac{1}{2}\left(W_1(x_1)\right)^2\right\}.
\end{multline}
Inserting the $1/N$ expansion \eqref{eq1.10} into this equation and collecting terms of order $N^{4-l}$ ($l\geq0$) yields the $(1,l)$ loop equation
\begin{multline} \label{eq5.14}
0=x_1^2\sum_{\mu\vdash(x_1,x_1,x_1,x_1)}\sum_{\substack{l_1+\cdots+l_{\#\mu}\\=l+2\#\mu-8}}\prod_{\mu_t\in\mu}W_{\#\mu_t}^{l_t}(\mu_t)-x_1W_1^l(x_1)+\chi_{l=0}
\\\hspace{-8em}+\lim_{\xi\to x_1}\left(x_1^2\frac{\partial^2}{\partial x_1^2}+2x_1\frac{\partial}{\partial x_1}\right)\Bigg\{W_2^{l-4}(\xi,x_1)+\frac{1}{2}W_2^{l-4}(x_1,x_1)
\\+\sum_{l_1+l_2=l-2}W_1^{l_1}(x_1)\left(W_1^{l_2}(\xi)+\frac{1}{2}W_1^{l_2}(x_1)\right)\Bigg\}.
\end{multline}
(Recall our convention of setting $W_n^{l'}:=0$ for all $n\geq1$ and $l'<0$.) Setting $l=0$ in the above gives the spectral curve
\begin{equation} \label{eq5.15}
0=x_1^2\left(W_1^0(x_1)\right)^4-x_1W_1^0(x_1)+1,
\end{equation}
while setting $l=1$ shows that
\begin{equation} \label{eq5.16}
0=\left(4x_1\left(W_1^0(x_1)\right)^3-1\right)W_1^1(x_1).
\end{equation}
Equation \eqref{eq5.15} has been given in \citep{FIL18}, where it was derived through techniques from free probability theory and also through comparison with the Fuss--Catalan distribution --- we have in analogy with equation \eqref{eq4.42} that $W_1^{(\mathcal{H}_m),0}(x_1)=x_1W_1^{(\mathrm{FC}_{2m+2})}(x_1^2)$ \citep{FIL18} and that equation \eqref{eq4.43} is known \citep{FL15} to have the generalisation
\begin{equation} \label{eq5.17}
0=u^{m-1}\left(W_1^{(\mathrm{FC}_m)}(u)\right)^m-uW_1^{(\mathrm{FC}_m)}(u)+1,\quad m\in\mathbb{N}.
\end{equation}
Since $4x_1\left(W_1^0(x_1)\right)^3=1$ is in contradiction with equation \eqref{eq5.15}, equation \eqref{eq5.16} tells us that $W_1^1(x_1)$ must be identically zero --- this is perfectly in line with our earlier observation that $W_n^l=0$ whenever $l$ is odd.

\begin{definition} \label{def7}
For notational convenience, define $\overline{A}_i(x_1,J_n)$ to be $A_i(x_1,J_n)$, as specified by equation \eqref{eq5.9}, but with $U_{n-1}(x_1,J_n\setminus\{x_i\})$ replaced by $W_{n-1}(x_1,J_n\setminus\{x_i\})$, $U_{n-1}(J_n)$ replaced by $W_{n-1}(J_n)$, and each instance of $U_{n+1}(K)$ ($K$ an ordered set of variables drawn from $\{x_1,\ldots,x_n\}$) replaced by
\begin{equation*}
\sum_{\mu\vdash K'}\sum_{\sqcup_{t=1}^{\#\mu}K_t=J_n\setminus\{x_i\}}\prod_{\mu_t\in\mu}W_{\#K_t+\#\mu_t}(\mu_t,K_t),
\end{equation*}
where $K'$ is the ordered set obtained by removing $J_n\setminus\{x_i\}$ from $K$. Likewise, define $\overline{A}_{i,j}(x_1,J_n)$ to be $A_{i,j}(x_1,J_n)$ \eqref{eq5.10}, but with each instance of $U_{n-1}(K)$ replaced by
\begin{equation*}
\sum_{\mu\vdash K''}\sum_{\sqcup_{t=1}^{\#\mu}K_t=J_n\setminus\{x_i,x_j\}}\prod_{\mu_t\in\mu}W_{\#K_t+\#\mu_t}(\mu_t,K_t),
\end{equation*}
where $K''$ is now the result of removing $J_n\setminus\{x_i,x_j\}$ from $K$. Finally, let $\overline{A}_{i,j,h}(x_1,J_n)$ be the function $A_{i,j,h}(x_1,J_n)$ \eqref{eq5.11} with each instance of $U_{n-3}(K)$ replaced by $W_{n-3}(K)$.
\end{definition}

Setting $n=2$ in equation \eqref{eq5.12} and subtracting equation \eqref{eq5.14} multiplied by $U_1(x_2)$ from the result yields the $n=2$ \textit{loop equation on the connected correlators},
\begin{multline} \label{eq5.18}
0=x_1^2\sum_{\mu\vdash(x_1,x_1,x_1,x_1)}\sum_{\sqcup_{t=1}^{\#\mu}K_t=\{x_2\}}\prod_{\mu_t\in\mu}W_{\#K_t+\#\mu_t}(\mu_t,K_t)-N^3x_1W_2(x_1,x_2)+\frac{\partial}{\partial x_2}\overline{A}_2(x_1,x_2)
\\\hspace{-6em}+\lim_{\xi\to x_1}\left(x_1^2\frac{\partial^2}{\partial x_1^2}+2x_1\frac{\partial}{\partial x_1}\right)\Bigg\{\sum_{\mu\vdash(\xi,x_1)}\sum_{\sqcup_{t=1}^{\#\mu}K_t=\{x_2\}}\prod_{\mu_t\in\mu}W_{\#K_t+\#\mu_t}(\mu_t,K_t)
\\+\frac{1}{2}\sum_{\mu\vdash(x_1,x_1)}\sum_{\sqcup_{t=1}^{\#\mu}K_t=\{x_2\}}\prod_{\mu_t\in\mu}W_{\#K_t+\#\mu_t}(\mu_t,K_t)\Bigg\};
\end{multline}
the auxiliary function $\overline{A}_2(x_1,x_2)$ is as in Definition \ref{def7} above. Substituting the large $N$ expansion \eqref{eq1.10} into this equation, multiplying the result by $N^{-3}$, and taking the limit $N\to\infty$ produces the $(n,l)=(2,0)$ loop equation
\begin{multline} \label{eq5.19}
0=4x_1^2\left(W_1^0(x_1)\right)^3W_2^0(x_1,x_2)-x_1W_2^0(x_1,x_2)-\frac{\partial}{\partial x_2}\frac{x_2^2\left(W_1^0(x_2)\right)^3}{x_1-x_2}
\\+\frac{\partial}{\partial x_2}\frac{x_1x_2\left[3\left(W_1^0(x_1)\right)^3-\left(W_1^0(x_1)\right)^2W_1^0(x_2)-W_1^0(x_1)\left(W_1^0(x_2)\right)^2\right]}{x_1-x_2}.
\end{multline}

As seen in \S\ref{s4.2}, the above method for deriving equation \eqref{eq5.18} extends to the general $n\geq3$ case. Thus, let us now give the analogue of Proposition \ref{prop7} for the $(N,N)$ Hermitised Laguerre ensemble --- we omit the proof, as it is exactly the same as that given for Proposition \ref{prop7}, with equation \eqref{eq5.12} now assuming the role of equation \eqref{eq4.26}.
\begin{proposition} \label{prop10}
Let $\overline{A}_i(x_1,J_n)$, $\overline{A}_{i,j}(x_1,J_n)$, and $\overline{A}_{i,j,h}(x_1,J_n)$ be as introduced in Definition \ref{def7} above. For $n\geq3$, the connected correlators of the $(N,N)$ Hermitised Laguerre ensemble satisfy the loop equation
\begin{multline} \label{eq5.20}
0=x_1^2\sum_{\mu\vdash(x_1,x_1,x_1,x_1)}\sum_{\sqcup_{t=1}^{\#\mu}K_t=J_n}\prod_{\mu_t\in\mu}W_{\#K_t+\#\mu_t}(\mu_t,K_t)-N^3x_1W_n(x_1,J_n)+\sum_{i=2}^n\frac{\partial}{\partial x_2}\overline{A}_i(x_1,J_n)
\\+\sum_{2\leq i<j\leq n}\frac{\partial^2}{\partial x_i\partial x_j}\overline{A}_{i,j}(x_1,J_n)+\frac{6}{x_1}\sum_{2\leq i<j<h\leq n}\frac{\partial^3}{\partial x_i\partial x_j\partial x_h}\overline{A}_{i,j,h}(x_1,J_n)
\\\hspace{-6em}+\lim_{\xi\to x_1}\left(x_1^2\frac{\partial^2}{\partial x_1^2}+2x_1\frac{\partial}{\partial x_1}\right)\Bigg\{\sum_{\mu\vdash(\xi,x_1)}\sum_{\sqcup_{t=1}^{\#\mu}K_t=J_n}\prod_{\mu_t\in\mu}W_{\#K_t+\#\mu_t}(\mu_t,K_t)
\\+\frac{1}{2}\sum_{\mu\vdash(x_1,x_1)}\sum_{\sqcup_{t=1}^{\#\mu}K_t=J_n}\prod_{\mu_t\in\mu}W_{\#K_t+\#\mu_t}(\mu_t,K_t)\Bigg\}.
\end{multline}
\end{proposition}

The above loop equation can be used to derive loop equations characterising the correlator expansion coefficients $W_n^l$ using the method described below Proposition \ref{prop7}. Like in \S\ref{s4.2}, we do not display these loop equations here, but note that they are straightforward rewritings of equation \eqref{eq5.20} in a similar fashion to how equation \eqref{eq5.14} is just equation \eqref{eq5.13} with products of connected correlators $W_{n'}$ replaced by suitable sums of products of correlator expansion coefficients $W_{n'}^{l'}$.

The loop equation \eqref{eq5.20} is of higher order than equation \eqref{eq4.37} (and also the loop equation of \citep{DF20}), this being indicated by the spectral curve \eqref{eq5.15} being a quartic polynomial in $W_1^0(x_1)$ instead of cubic, the first sum in the right-hand side of equation \eqref{eq5.20} containing a product of four correlators instead of at most three, and there being a sum over $2\leq i<j<h\leq n$ of third order partial derivatives $\partial_{x_i}\partial_{x_j}\partial_{x_h}$ of the simple functions $\overline{A}_{i,j,h}(x_1,J_n)$ involving $W_{n-3}$ in said loop equation; cf.~the last two lines of equation \eqref{eq4.37}. Further points of difference include the occurrence of the limiting term in the last two lines of equation \eqref{eq5.20} and the lack of denominators involving differences of squares, which is a special feature of the antisymmetrised Laguerre ensemble.

\subsection{Discussion on \texorpdfstring{$W_1^{(\mathcal{H}_1),0}$}{W1(H1),0} and \texorpdfstring{$W_2^{(\mathcal{H}_1),0}$}{W2(H1),0}} \label{s5.3}
In parallel with \S\ref{s4.3}, we now find a rational parametrisation for the spectral curve \eqref{eq5.15} of the $(N,N)$ Hermitised Laguerre ensemble, evaluate $W_1^0(x_1)$, $W_2^0(x_1,x_2)$, and $W_1^2(x_1)$ in terms of this parametrisation, and then use residue calculus to compute some coefficients of the genus expansion
\begin{equation} \label{eq5.21}
c_{k_1,\ldots,k_n}^{(\mathcal{H}_1)}=N^{2-n}\sum_{l=0}^{\tfrac{3}{2}\overline{k}_n-n}\frac{c_{k_1,\ldots,k_n}^{(\mathcal{H}_1),l}}{N^l}
\end{equation}
for particular values of $n,k_1,\ldots,k_n$; similar to the expansion \eqref{eq4.47}, the above is known to exist due to Proposition \ref{prop1}. Thus, our immediate goal is to find rational functions $x(z),y(z)$ such that equation \eqref{eq5.15} holds true for all $z$ upon replacing $W_1^0(x_1)$ by $y(z)$ and $x_1$ by $x(z)$. By using the Newton polygon method \citep{BP00}, we know that the complex algebraic curve $0=x^2y^4-xy+1$ is of genus zero and is thus the Riemann sphere $\mathbb{CP}^1$. Hence, our desired parametrisation must exist (see, e.g., \citep[Thrm.~4.63]{SWP08}; an algorithm guaranteed to produce such a parametrisation is given in \citep[Sec.~4.7]{SWP08}). Writing $u(z)=x(z)y(z)^2$ shows that our spectral curve (written in the $x(z),y(z)$ variables) is equivalent to
\begin{equation*}
0=u(z)^2-\frac{u(z)}{y(z)}+1.
\end{equation*}
As this is a quadratic equation in $u(z)$, as considered in the one-cut 1-Hermitian matrix model, we may adapt the Joukowsky transform to let
\begin{equation} \label{eq5.22}
\frac{1}{y(z)}=z+\frac{1}{z}\implies y(z)=\frac{z}{z^2+1}.
\end{equation}
Then, $u(z)=z$ (or $u(z)=1/z$, but we ignore this choice without loss of generality) and, consequently,
\begin{equation} \label{eq5.23}
x(z)=\frac{u(z)}{y(z)^2}=\frac{(z^2+1)^2}{z}.
\end{equation}
It can immediately be seen that the parametrisations \eqref{eq5.22}, \eqref{eq5.23} satisfy the equation
\begin{equation*}
0=x(z)^2y(z)^4-x(z)y(z)+1,
\end{equation*}
so $y(z)$ is the single-valued analytic continuation of $W_1^0(x(z))$ to $\mathbb{CP}^1$. That is, $W_1^0(x_1)$ is equal to $y(z)$ when $z\in x^{-1}(x_1)$ is chosen to be close to $z=0$. To see why we must make this choice, note that our desired solution of equation \eqref{eq5.15} has the property that $W_1^0(x_1)\sim1/x_1$ as $x_1\to\infty$, which translates to the requirement that $y(z)\sim x(z)y(z)^2=u(z)=z\to0$ as $x(z)\to\infty$ (cf.~the discussion below the parametrisation \eqref{eq4.44}).

Moving on, taking $z_1,z_2\in\mathbb{CP}^1$ and making the substitutions $x_1=x(z_1)$ and $x_2=x(z_2)$ (with $W_1^0(x_i)=y(z_i)$) into equation \eqref{eq5.19} shows that
\begin{equation} \label{eq5.24}
W_2^0(x(z_1),x(z_2))=\frac{1}{x'(z_1)x'(z_2)}\frac{1}{(z_1-z_2)^2}-\frac{1}{(x(z_1)-x(z_2))^2}.
\end{equation}
This is exactly the expected universal form for $W_2^0$ when working with a genus zero spectral curve; cf.~the first line of equation \eqref{eq4.45} and the discussion below said equation.

Per the comment below equation \eqref{eq5.17}, $W_1^1(x_1)$ is identically zero, so the next-to-leading order term in the large $N$ expansion of $W_1(x_1)$ is $W_1^2(x_1)$, which is equivalently the $1/N^2$ correction to $W_1^0(x_1)$. With the evaluation \eqref{eq5.24} in hand, we may set $l=2$ in the loop equation \eqref{eq5.14} with $x_1$ replaced by $x(z_1)$ \eqref{eq5.23}, $W_1^0(x_1)$ replaced by $y(z_1)$ \eqref{eq5.22}, and $W_2^0(x_1,x_1)$ replaced by the limit
\begin{equation*}
\lim_{z_2\to z_1}W_2^0(x(z_1),x(z_2))=\frac{6z_1^6-2z_1^4+10z_1^2+2}{z_1^4\,x'(z_1)^4}
\end{equation*}
to compute
\begin{equation} \label{eq5.25}
W_1^2(x(z_1))=-\frac{z_1^3(9z_1^8+38z_1^4+16z_1^2+1)}{(3z_1^2-1)^5(z_1^2+1)^3}.
\end{equation}

As in the case of the antisymmetrised Laguerre ensemble, we observe that
\begin{equation} \label{eq5.26}
W_n^l(x_1,\ldots,x_n)=\sum_{k_1,\ldots,k_n=0}^{\infty}\frac{c_{k_1,\ldots,k_n}^{(\mathcal{H}_1),l}}{x_1^{k_1+1}\cdots x_n^{k_n+1}}.
\end{equation}
Hence, in analogy with equations \eqref{eq4.48}, \eqref{eq4.49} and recalling that $x(z_1),\ldots,x(z_n)\to\infty$ corresponds to $z_1,\ldots,z_n\to0$ in the regime where equation \eqref{eq5.26} is valid, we may use the residue formulae
\begin{align}
c_{k_1,\ldots,k_n}^{(\mathcal{H}_1),l}&=(-1)^n\underset{x_1,\ldots,x_n=\infty}{\mathrm{Res}}x_1^{k_1}\cdots x_n^{k_n}W_n^l(x_1,\ldots,x_n)\,\mathrm{d}x_1\cdots\mathrm{d}x_n \nonumber
\\&=(-1)^n\underset{z_1,\ldots,z_n=0}{\mathrm{Res}}W_n^l(x(z_1),\ldots,x(z_n))\prod_{i=1}^nx(z_i)^{k_i}x'(z_i)\,\mathrm{d}z_i \label{eq5.27}
\end{align}
to compute the coefficients of the genus expansion \eqref{eq5.21}. Substituting equations \eqref{eq5.22}, \eqref{eq5.24}, and \eqref{eq5.25} into the latter residue formula, we present the evaluation of $c_{k_1,\ldots,k_n}^{(\mathcal{H}_1),l}$ for $(n,l)=(1,0),(1,2),(2,0)$ and some low values of $k_1,k_2$ in Appendix \ref{appendixD.2}.

\section*{Acknowledgements}
The majority of the work of AAR was supported by a grant of the Australian Research Council, Discovery Project DP210102887. While finalising this work, AAR was supported by the Hong Kong RGC grants GRF 16304724 and GRF 17304225.

\appendix
\setcounter{equation}{0}
\section{Loop Equations on the \texorpdfstring{$m_{k_1,\ldots,k_n}^{(\mathcal{J}_1)}$}{m.k1,...,kn(J1)}} \label{appendixC.1}
Let $m_{k_1,\ldots,k_n}=m_{k_1,\ldots,k_n}^{(\mathcal{J}_1)}$ be as defined in equation \eqref{eq1.6} with $m=1$. Substituting equations \eqref{eq4.16}--\eqref{eq4.19} into equation \eqref{eq4.10} and then substituting the result into equation \eqref{eq4.7} yields the following loop equation (recall that $I_n=(k_2,\ldots,k_n)$ and $\Tr B^{I_n}=\Tr B^{k_2}\cdots\Tr B^{k_n}$):
\begin{multline} \label{eqA.1}
0=N^2m_{k_1+1,I_n}+\sum_{p_1+p_2+p_3=k_1-1}m_{p_1,p_2,p_3,I_n}-\sum_{p_1+p_2=k_1-1}m_{p_1,p_2,I_n}+\frac{k_1^2-1}{2}m_{k_1-1,I_n}
\\+8\sum_{2\leq i<j\leq n}k_ik_j[k_i+1]_{\mathrm{mod}\,2}[k_j+1]_{\mathrm{mod}\,2}m_{k_1+k_i+k_j-1,I_n\setminus\{k_i,k_j\}}
\\\hspace{-2em}+2\sum_{i=2}^nk_i[k_i+1]_{\mathrm{mod}\,2}\Bigg(2\sum_{p_1+p_2=k_1-1}m_{k_i+p_1,p_2,I_n\setminus\{k_i\}}+\sum_{p_1+p_2=k_i-1}m_{k_1+p_1,p_2,I_n\setminus\{k_i\}}
\\-m_{k_1+k_i-1,I_n\setminus\{k_i\}}\Bigg).
\end{multline}

To convert this equation into the loop equation \eqref{eq4.26} on the unconnected correlators $U_n(x_1,\ldots,x_n)=U^{(\mathcal{J}_1)}_n(x_1,\ldots,x_n)$, we multiply both sides by $x_1^{-k_1-1}\cdots x_n^{-k_n-1}$ and sum over $k_1,\ldots,k_n\geq0$. Applying this prescription term by term requires the following glossary (recall that $J_n=(x_2,\ldots,x_n)$):
\begin{equation} \label{eqA.2}
\sum_{k_1,\ldots,k_n\geq0}\sum_{p_1+p_2+p_3=k_1-1}\frac{m_{p_1,p_2,p_3,I_n}}{x_1^{k_1+1}\cdots x_n^{k_n+1}}=x_1U_{n+2}(x_1,x_1,x_1,J_n),
\end{equation}
\begin{equation} \label{eqA.3}
\sum_{k_1,\ldots,k_n\geq0}\sum_{p_1+p_2=k_1-1}\frac{m_{p_1,p_2,I_n}}{x_1^{k_1+1}\cdots x_n^{k_n+1}}=U_{n+1}(x_1,x_1,J_n),
\end{equation}
\begin{equation} \label{eqA.4}
\sum_{k_1,\ldots,k_n\geq0}\frac{m_{k_1+1,I_n}}{x_1^{k_1+1}\cdots x_n^{k_n+1}}=x_1U_n(x_1,J_n)-NU_{n-1}(J_n),
\end{equation}
\begin{equation} \label{eqA.5}
\sum_{k_1,\ldots,k_n\geq0}\frac{k_1^2-1}{2}\frac{m_{k_1-1,I_n}}{x_1^{k_1+1}\cdots x_n^{k_n+1}}=\frac{1}{2x_1}\left(x_1^2\frac{\partial^2}{\partial x_1^2}+x_1\frac{\partial}{\partial x_1}-1\right)U_n(x_1,J_n),
\end{equation}
\begin{multline} \label{eqA.6}
\sum_{k_1,\ldots,k_n\geq0}k_i\left[k_i+1\right]_{\textrm{mod 2}}\frac{m_{k_1+k_i-1,I_n\setminus\{k_i\}}}{x_1^{k_1+1}\cdots x_n^{k_n+1}}
\\=\frac{1}{x_1}\frac{\partial}{\partial x_i}\left\{\frac{x_i^2U_{n-1}(x_1,J_n\setminus\{x_i\})-x_1x_iU_{n-1}(J_n)}{x_1^2-x_i^2}\right\},
\end{multline}
\begin{multline} \label{eqA.7}
\sum_{k_1,\ldots,k_n\geq0}\sum_{p_1+p_2=k_i-1}k_i\left[k_i+1\right]_{\textrm{mod 2}}\frac{m_{k_1+p_1,p_2,I_n\setminus\{k_i\}}}{x_1^{k_1+1}\cdots x_n^{k_n+1}}
\\=\frac{\partial}{\partial x_i}\left\{\frac{x_1x_iU_{n}(x_1,J_n)-x_i^2U_{n}(x_i,J_n)}{x_1^2-x_i^2}\right\},
\end{multline}
\begin{multline} \label{eqA.8}
\sum_{k_1,\ldots,k_n\geq0}\sum_{p_1+p_2=k_1-1}k_i\left[k_i+1\right]_{\textrm{mod 2}}\frac{m_{k_i+p_1,p_2,I_n\setminus\{k_i\}}}{x_1^{k_1+1}\cdots x_n^{k_n+1}}
\\=\frac{\partial}{\partial x_i}\left\{\frac{x_i^2U_n(x_1,x_1,J_n\setminus\{x_i\})-x_1x_iU_n(x_1,J_n)}{x_1^2-x_i^2}\right\},
\end{multline}
\begin{multline} \label{eqA.9}
\sum_{k_1,\ldots,k_n\geq0}k_ik_j\left[k_i+1\right]_{\textrm{mod 2}}\left[k_j+1\right]_{\textrm{mod 2}}\frac{m_{k_1+k_i+k_j-1,I_n\setminus\{k_i,k_j\}}}{x_1^{k_1+1}\cdots x_n^{k_n+1}}
\\=\frac{1}{x_1}\frac{\partial^2}{\partial x_i \partial x_j}\Bigg\{\frac{x_1x_i^2x_jU_{n-2}(J_n\setminus\{x_i\})}{(x_1^2-x_j^2)(x_i^2-x_j^2)}-\frac{x_1x_ix_j^2U_{n-2}(J_n\setminus\{x_j\})}{(x_1^2-x_i^2)(x_i^2-x_j^2)}
\\+\frac{x_i^2x_j^2U_{n-2}(x_1,J_n\setminus\{x_i,x_j\})}{(x_1^2-x_i^2)(x_1^2-x_j^2)}\Bigg\}.
\end{multline}

The proofs of equations \eqref{eqA.2}--\eqref{eqA.9} are a little more involved than one might expect, so we compute a few illustrative examples for the reader's convenience:

To prove equation \eqref{eqA.2}, one needs to appropriately split up the factor $x_1^{k_1+1}$ and interchange the order of summation to see that
\begin{align*}
\sum_{k_1\geq0}\sum_{p_1+p_2+p_3=k_1-1}\frac{m_{p_1,p_2,p_3,I_n}}{x_1^{k_1+1}}&=x_1\sum_{k_1=0}^\infty\sum_{p_1=0}^{k_1-1}\sum_{p_2=0}^{k_1-p_1-1}\frac{m_{p_1,p_2,k_1-p_1-p_2-1,I_n}}{x_1^{p_1+1}x_1^{p_2+1}x_1^{k_1-p_1-p_2}}
\\&=x_1\sum_{p_1=0}^\infty\sum_{k_1=p_1+1}^\infty\sum_{p_2=0}^{k_1-p_1-1}\frac{m_{p_1,p_2,k_1-p_1-p_2-1,I_n}}{x_1^{p_1+1}x_1^{p_2+1}x_1^{k_1-p_1-p_2}}
\\&=x_1\sum_{p_1=0}^\infty\sum_{k_1=0}^\infty\sum_{p_2=0}^{k_1}\frac{m_{p_1,p_2,k_1-p_2,I_n}}{x_1^{p_1+1}x_1^{p_2+1}x_1^{k_1-p_2+1}}
\\&=x_1\sum_{p_1=0}^\infty\sum_{p_2=0}^\infty\sum_{k_1=p_2}^\infty\frac{m_{p_1,p_2,k_1-p_2,I_n}}{x_1^{p_1+1}x_1^{p_2+1}x_1^{k_1-p_2+1}}
\\&=x_1\sum_{p_1=0}^\infty\sum_{p_2=0}^\infty\sum_{k_1=0}^\infty\frac{m_{p_1,p_2,k_1,I_n}}{x_1^{p_1+1}x_1^{p_2+1}x_1^{k_1+1}}.
\end{align*}
Multiplying this result by $x_2^{-k_2-1}\cdots x_n^{-k_n-1}$ and further summing over $k_2,\ldots,k_n\geq0$ then gives the right-hand side of equation \eqref{eqA.2}.

Equation \eqref{eqA.4} has a simple but comparatively unique proof:
\begin{multline*}
\sum_{k_1=0}^{\infty}\frac{m_{k_1+1,I_n}}{x_1^{k_1+1}}=\sum_{k_1=-1}^{\infty}\frac{m_{k_1+1,I_n}}{x_1^{k_1+1}}-m_{0,I_n}=x_1\sum_{k_1=-1}^{\infty}\frac{m_{k_1+1,I_n}}{x_1^{k_1+2}}-m_{0,I_n}
\\=x_1\sum_{k_1=0}^{\infty}\frac{m_{k_1,I_n}}{x_1^{k_1+1}}-m_{0,I_n}=x_1\sum_{k_1=0}^{\infty}\frac{m_{k_1,I_n}}{x_1^{k_1+1}}-Nm_{I_n},
\end{multline*}
using the fact that $m_{0,I_n}=\mean{\Tr B^0\Tr B^{I_n}}=N\mean{\Tr B^{I_n}}=Nm_{I_n}$. Multiplying again by $x_2^{-k_2-1}\cdots x_n^{-k_n-1}$ and summing over $k_2,\ldots,k_n\geq0$ produces the desired result.

To prove equation \eqref{eqA.6}, we first note that since $[k_i+1]_{\mathrm{mod}\,2}$ vanishes if $k_i$ is odd, we can replace $k_i$ by $2k_i$ in the summand of the left-hand side of equation \eqref{eqA.6}. However, since $m_{k_1+2k_i-1,I_n\setminus\{k_i\}}=\mean{\Tr B^{k_1+2k_i-1}\Tr B^{I_n\setminus\{k_i\}}}$ vanishes whenever $k_1$ is even, we must also replace $k_1$ by $2k_1+1$. Thus, the left-hand side of equation \eqref{eqA.6} simplifies as
\begin{align*}
\sum_{k_1,\ldots,k_n\geq0}&k_i\left[k_i+1\right]_{\textrm{mod 2}}\frac{m_{k_1+k_i-1,I_n\setminus\{k_i\}}}{x_1^{k_1+1}\cdots x_n^{k_n+1}}
\\&=-\frac{\partial}{\partial x_i}x_i\sum_{k_1,\ldots,k_n\geq0}\frac{m_{2k_1+2k_i,I_n\setminus\{k_i\}}}{x_1^{2k_1+2}x_2^{k_2+1}\cdots x_{i-1}^{k_{i-1}+1}x_i^{2k_i+1}x_{i+1}^{k_{i+1}+1}\cdots x_n^{k_n+1}}
\\&=-\frac{1}{x_1}\frac{\partial}{\partial x_i}\sum_{k_1,\ldots,k_n\geq0}\frac{m_{2k_1+2k_i,I_n\setminus\{k_i\}}}{x_1^{2k_1+2k_i+1}x_2^{k_2+1}\cdots x_{i-1}^{k_{i-1}+1}x_{i+1}^{k_{i+1}+1}\cdots x_n^{k_n+1}}\left(\frac{x_1}{x_i}\right)^{2k_i}
\\&=-\frac{1}{x_1}\frac{\partial}{\partial x_i}\sum_{k_2,\ldots,k_n\geq0}\sum_{k_1\geq k_i}\frac{m_{2k_1,I_n\setminus\{k_i\}}}{x_1^{2k_1+1}x_2^{k_2+1}\cdots x_{i-1}^{k_{i-1}+1}x_{i+1}^{k_{i+1}+1}\cdots x_n^{k_n+1}}\left(\frac{x_1}{x_i}\right)^{2k_i}
\\&=-\frac{1}{x_1}\frac{\partial}{\partial x_i}\sum_{k_1,\ldots,k_{i-1},k_{i+1},\ldots,k_n\geq0}\sum_{k_i=0}^{k_1}\frac{m_{2k_1,I_n\setminus\{k_i\}}}{x_1^{2k_1+1}x_2^{k_2+1}\cdots x_{i-1}^{k_{i-1}+1}x_{i+1}^{k_{i+1}+1}\cdots x_n^{k_n+1}}\left(\frac{x_1}{x_i}\right)^{2k_i}
\\&=-\frac{1}{x_1}\frac{\partial}{\partial x_i}\sum_{k_1,\ldots,k_{i-1},k_{i+1},\ldots,k_n\geq0}\frac{m_{2k_1,I_n\setminus\{k_i\}}}{x_1^{2k_1+1}x_2^{k_2+1}\cdots x_{i-1}^{k_{i-1}+1}x_{i+1}^{k_{i+1}+1}\cdots x_n^{k_n+1}}\frac{1-(x_1/x_i)^{2k_1+2}}{1-x_1^2/x_i^2}
\\&=\frac{1}{x_1}\frac{\partial}{\partial x_i}\frac{x_i^2}{x_1^2-x_i^2}\Bigg(\sum_{k_1,\ldots,k_{i-1},k_{i+1},\ldots,k_n\geq0}\frac{m_{2k_1,I_n\setminus\{k_i\}}}{x_1^{2k_1+1}x_2^{k_2+1}\cdots x_{i-1}^{k_{i-1}+1}x_{i+1}^{k_{i+1}+1}\cdots x_n^{k_n+1}}
\\&\hspace{10em}-\sum_{k_1,\ldots,k_{i-1},k_{i+1},\ldots,k_n\geq0}\frac{x_1}{x_i}\frac{m_{2k_1,I_n\setminus\{k_i\}}}{x_i^{2k_1+1}x_2^{k_2+1}\cdots x_{i-1}^{k_{i-1}+1}x_{i+1}^{k_{i+1}+1}\cdots x_n^{k_n+1}}\Bigg).
\end{align*}
We may now replace $2k_1$ by $k_1$ in the summands without any problem since $m_{k_1,I_n\setminus\{k_i\}}=0$ whenever $k_1$ is odd. Rewriting the sums in terms of the unconnected correlators according to equation \eqref{eq1.8} and performing some basic algebra then yields the right-hand side of equation \eqref{eqA.6}, as required.

We do not display the proofs of equations \eqref{eqA.3}, \eqref{eqA.5}, \eqref{eqA.7}--\eqref{eqA.9} since they can be proven using extensions of the ideas given above.

\setcounter{equation}{0}
\section{Loop Equations on the \texorpdfstring{$m_{k_1,\ldots,k_n}^{(\mathcal{H}_1)}$}{m.k1,...,kn(H1)}} \label{appendixC.2}
Let $m_{k_1,\ldots,k_n}=m^{(\mathcal{H}_1)}_{k_1,\ldots,k_n}$ and $U_n(x_1,\ldots,x_n)=U^{(\mathcal{H}_1)}_n(x_1,\ldots,x_n)$. Taking $i,j,h$ to be pairwise distinct, the averages $\mathcal{A}_{i',j',h'}$ of equation \eqref{eq5.8} are given by
\begin{multline} \label{eqB.1}
\mathcal{A}_{1,1,1}=\sum_{p_1+\cdots+p_4=k_1-1}m_{p_1,\ldots,p_4,I_n}+\frac{k_1(k_1-1)}{2}\sum_{p_1+p_2=k_1-1}m_{p_1,p_2,I_n}
\\+2\sum_{p_1+p_2+p_3=k_1-1}p_1m_{p_1+p_2,p_3,I_n},
\end{multline}
\begin{multline}
\mathcal{A}_{1,1,i}=\mathcal{A}_{1,i,1}=\mathcal{A}_{i,1,1}
\\=\frac{k_ik_1(k_1-1)}{2}m_{k_1+k_i-1,I_n\setminus\{k_i\}}+k_i\sum_{p_1+p_2+p_3=k_1-1}m_{p_1+k_i,p_2,p_3,I_n\setminus\{k_i\}},
\end{multline}
\begin{align}
\mathcal{A}_{1,i,i}&=\mathcal{A}_{i,i,1}=k_i\sum_{p_1+p_2=k_1-1}\sum_{q_1+q_2=k_i-1}m_{p_1+q_1+1,p_2,q_2,I_n\setminus\{k_i\}},
\\\mathcal{A}_{i,1,i}&=k_1k_i^2m_{k_1+k_i-1,I_n\setminus\{k_i\}},
\\\mathcal{A}_{i,i,i}&=\frac{k_i^2(k_i-1)}{2}m_{k_1+k_i-1,I_n\setminus\{k_i\}}+k_i\sum_{p_1+p_2+p_3=k_i-1}m_{k_1+p_1,p_2,p_3,I_n\setminus\{k_i\}},
\\\mathcal{A}_{1,i,j}&=\mathcal{A}_{i,j,1}=k_ik_j\sum_{p_1+p_2=k_1-1}m_{p_1+k_i+k_j,p_2,I_n\setminus\{k_i,k_j\}},
\\\mathcal{A}_{i,1,j}&=k_ik_j\sum_{p_1+p_2=k_1-1}m_{p_1+k_i,p_2+k_j,I_n\setminus\{k_i,k_j\}},
\\\mathcal{A}_{i,i,j}&=\mathcal{A}_{j,i,i}=k_ik_j\sum_{p_1+p_2=k_i-1}m_{k_1+p_1+k_j,p_2,I_n\setminus\{k_i,k_j\}},
\\\mathcal{A}_{i,j,i}&=k_ik_j\sum_{p_1+p_2=k_i-1}m_{k_1+p_1,p_2+k_j,I_n\setminus\{k_i,k_j\}},
\\\mathcal{A}_{i,j,h}&=k_ik_jk_hm_{k_1+k_i+k_j+k_h-1,I_n\setminus\{k_i,k_j,k_h\}}. \label{eqB.10}
\end{align}

We do not provide proofs for equations \eqref{eqB.1}--\eqref{eqB.10}, as they can be computed using Lemma \ref{Lemma8} via the same ideas as in the proof of Lemma \ref{Lemma6}.

Now, substituting equation \eqref{eq5.8} into equation \eqref{eq5.5} while observing that some of the above expressions are unchanged upon reordering indices shows that
\begin{align}
0&=-N^3m_{k_1+1,I_n}+k_1\sum_{p_1+p_2=k_1-1}m_{p_1,p_2,I_n}+2\sum_{2\leq i<j\leq n}k_ik_jm_{k_1,k_i+k_j-1,I_n\setminus\{k_i,k_j\}} \nonumber
\\&\quad+\sum_{i=2}^nk_i\left(2k_1m_{k_1+k_i-1,I_n\setminus\{k_i\}}+\sum_{p_1+p_2=k_i-1}m_{k_1,p_1,p_2,I_n\setminus\{k_i\}}\right)+\mathcal{A}_{1,1,1} \nonumber
\\&\quad+\sum_{i=2}^n\left(3\mathcal{A}_{1,1,i}+2\mathcal{A}_{1,i,i}+\mathcal{A}_{i,1,i}+\mathcal{A}_{i,i,i}\right)+2\sum_{2\leq i<j\leq n}\left(2\mathcal{A}_{1,i,j}+\mathcal{A}_{i,1,j}\right) \nonumber
\\&\quad+\sum_{\substack{2\leq i,j\leq n,\\i\neq j}}\left(2\mathcal{A}_{i,i,j}+\mathcal{A}_{i,j,i}\right)+6\sum_{2\leq i<j<h\leq n}\mathcal{A}_{i,j,h}. \label{eqB.11}
\end{align}
Inserting the expressions \eqref{eqB.1}--\eqref{eqB.10} into this equation then produces the loop equation on the mixed moments,
\begin{align}
0&=-N^3m_{k_1+1,I_n}+k_1\sum_{p_1+p_2=k_1-1}m_{p_1,p_2,I_n}+2\sum_{2\leq i<j\leq n}k_ik_jm_{k_1,k_i+k_j-1,I_n\setminus\{k_i,k_j\}} \nonumber
\\&\quad+\sum_{i=2}^nk_i\left(2k_1m_{k_1+k_i-1,I_n\setminus\{k_i\}}+\sum_{p_1+p_2=k_i-1}m_{k_1,p_1,p_2,I_n\setminus\{k_i\}}\right)+\sum_{p_1+\cdots+p_4=k_1-1}m_{p_1,\ldots,p_4,I_n} \nonumber
\\&\quad+\frac{k_1(k_1-1)}{2}\sum_{p_1+p_2=k_1-1}m_{p_1,p_2,I_n}+2\sum_{p_1+p_2+p_3=k_1-1}p_1m_{p_1+p_2,p_3,I_n} \nonumber
\\&\quad+3\sum_{i=2}^n\left(\frac{k_ik_1(k_1-1)}{2}m_{k_1+k_i-1,I_n\setminus\{k_i\}}+k_i\sum_{p_1+p_2+p_3=k_1-1}m_{p_1+k_i,p_2,p_3,I_n\setminus\{k_i\}}\right) \nonumber
\\&\quad+\sum_{i=2}^n\left(2k_i\sum_{p_1+p_2=k_1-1}\sum_{q_1+q_2=k_i-1}m_{p_1+q_1+1,p_2,q_2,I_n\setminus\{k_i\}}+k_1k_i^2m_{k_1+k_i-1,I_n\setminus\{k_i\}}\right) \nonumber
\\&\quad+\sum_{i=2}^n\left(\frac{k_i^2(k_i-1)}{2}m_{k_1+k_i-1,I_n\setminus\{k_i\}}+k_i\sum_{p_1+p_2+p_3=k_i-1}m_{k_1+p_1,p_2,p_3,I_n\setminus\{k_i\}}\right) \nonumber
\\&\quad+2\sum_{2\leq i<j\leq n}k_ik_j\left(2\sum_{p_1+p_2=k_1-1}m_{p_1+k_i+k_j,p_2,I_n\setminus\{k_i,k_j\}}+\sum_{p_1+p_2=k_1-1}m_{p_1+k_i,p_2+k_j,I_n\setminus\{k_i,k_j\}}\right) \nonumber
\\&\quad+\sum_{\substack{2\leq i,j\leq n,\\i\neq j}}k_ik_j\left(2\sum_{p_1+p_2=k_i-1}m_{k_1+p_1+k_j,p_2,I_n\setminus\{k_i,k_j\}}+\sum_{p_1+p_2=k_i-1}m_{k_1+p_1,p_2+k_j,I_n\setminus\{k_i,k_j\}}\right) \nonumber
\\&\quad+6\sum_{2\leq i<j<h\leq n}k_ik_jk_hm_{k_1+k_i+k_j+k_h-1,I_n\setminus\{k_i,k_j,k_h\}}. \label{eqB.12}
\end{align}

Multiplying both sides of this equation by $x_1^{-k_1-1}\cdots x_n^{-k_n-1}$ and then taking the sum over $k_1,\ldots,k_n\geq0$ yields the loop equation on the unconnected correlators $U_n(x_1,\ldots,x_n)$ given in Proposition~\ref{prop9}. We apply this prescription term by term using extensions of the ideas underlying the proofs presented in Appendix \ref{appendixC.1} above (recall again that $J_n=(x_2,\ldots,x_n)$):
\begin{equation} \label{eqB.13}
\sum_{k_1,\ldots,k_n\geq0}\frac{m_{k_1+1,I_n}}{x_1^{k_1+1}\cdots x_n^{k_n+1}}=x_1U_n(x_1,J_n)-NU_{n-1}(J_n),
\end{equation}
\begin{equation} \label{eqB.14}
\sum_{k_1,\ldots,k_n\geq0}k_1\sum_{p_1+p_2=k_1-1}\frac{m_{p_1,p_2,I_n}}{x_1^{k_1+1}\cdots x_n^{k_n+1}}=-\frac{\partial}{\partial x_1}x_1U_{n+1}(x_1,x_1,J_n),
\end{equation}
\begin{multline} \label{eqB.15}
\sum_{k_1,\ldots,k_n\geq0}k_ik_j\frac{m_{k_1,k_i+k_j-1,I_n\setminus\{k_i,k_j\}}}{x_1^{k_1+1}\cdots x_n^{k_n+1}}
\\=\frac{\partial^2}{\partial x_i\partial x_j}\left\{\frac{x_iU_{n-1}(x_1,J_n\setminus\{x_i\})-x_jU_{n-1}(x_1,J_n\setminus\{x_j\})}{x_i-x_j}\right\},
\end{multline}
\begin{equation} \label{eqB.16}
\sum_{k_1,\ldots,k_n\geq0}k_1k_i\frac{m_{k_1+k_i-1,I_n\setminus\{k_i\}}}{x_1^{k_1+1}\cdots x_n^{k_n+1}}=\frac{\partial^2}{\partial x_1\partial x_i}\left\{\frac{x_1U_{n-1}(J_n)-x_iU_{n-1}(x_1,J_n\setminus\{x_i\})}{x_1-x_i}\right\},
\end{equation}
\begin{equation} \label{eqB.17}
\sum_{k_1,\ldots,k_n\geq0}k_i\sum_{p_1+p_2=k_i-1}\frac{m_{k_1,p_1,p_2,I_n\setminus\{k_i\}}}{x_1^{k_1+1}\cdots x_n^{k_n+1}}=-\frac{\partial}{\partial x_i}x_iU_{n+1}(x_1,x_i,J_n),
\end{equation}
\begin{multline} \label{eqB.18}
\sum_{k_1,\ldots,k_n\geq0}\frac{\mathcal{A}_{1,1,1}}{x_1^{k_1+1}\cdots x_n^{k_n+1}}=x_1^2U_{n+3}(x_1,x_1,x_1,x_1,J_n)
\\\hspace{4em}+\left(\frac{1}{2}x_1^2\frac{\partial^2}{\partial x_1^2}+2x_1\frac{\partial}{\partial x_1}+1\right)U_{n+1}(x_1,x_1,J_n)
\\+\lim_{\xi\to x_1}\left(x_1^2\frac{\partial^2}{\partial x_1^2}+2x_1\frac{\partial}{\partial x_1}\right)U_{n+1}(\xi,x_1,J_n),
\end{multline}
\begin{multline} \label{eqB.19}
\sum_{k_1,\ldots,k_n\geq0}\frac{\mathcal{A}_{1,1,i}}{x_1^{k_1+1}\cdots x_n^{k_n+1}}=\frac{\partial}{\partial x_i}\frac{x_1x_i}{x_1-x_i}\big\{U_{n+1}(x_1,x_1,x_1,J_n\setminus\{x_i\})-U_{n+1}(x_1,x_1,J_n)\big\}
\\+\frac{1}{2}\frac{\partial^3}{\partial x_1^2\partial x_i}\left\{\frac{x_1x_iU_{n-1}(x_1,J_n\setminus\{x_i\})-x_1^2U_{n-1}(J_n)}{x_1-x_i}\right\},
\end{multline}
\begin{equation} \label{eqB.20}
\sum_{k_1,\ldots,k_n\geq0}\frac{\mathcal{A}_{1,i,i}}{x_1^{k_1+1}\cdots x_n^{k_n+1}}=\frac{\partial}{\partial x_i}\left\{\frac{x_1x_iU_{n+1}(x_1,x_1,J_n)-x_i^2U_{n+1}(x_1,x_i,J_n)}{x_1-x_i}\right\},
\end{equation}
\begin{multline} \label{eqB.21}
\sum_{k_1,\ldots,k_n\geq0}\frac{\mathcal{A}_{i,1,i}}{x_1^{k_1+1}\cdots x_n^{k_n+1}}
\\=\left(x_i\frac{\partial}{\partial x_i}+1\right)\frac{\partial^2}{\partial x_1\partial x_i}\left\{\frac{x_iU_{n-1}(x_1,J_n\setminus\{x_i\})-x_1U_{n-1}(J_n)}{x_1-x_i}\right\},
\end{multline}
\begin{multline} \label{eqB.22}
\sum_{k_1,\ldots,k_n\geq0}\frac{\mathcal{A}_{i,i,i}}{x_1^{k_1+1}\cdots x_n^{k_n+1}}=\frac{1}{2x_1}\left(x_i\frac{\partial}{\partial x_i}+1\right)\frac{\partial^2}{\partial x_i^2}\left\{\frac{x_i^2U_{n-1}(x_1,J_n\setminus\{x_i\})-x_1x_iU_{n-1}(J_n)}{x_1-x_i}\right\}
\\+\frac{\partial}{\partial x_i}\frac{x_i^2}{x_1-x_i}\big\{U_{n+1}(x_1,x_i,J_n)-U_{n+1}(x_i,x_i,J_n)\big\},
\end{multline}
\begin{multline} \label{eqB.23}
\sum_{k_1,\ldots,k_n\geq0}\frac{\mathcal{A}_{1,i,j}}{x_1^{k_1+1}\cdots x_n^{k_n+1}}=\frac{\partial^2}{\partial x_i\partial x_j}x_ix_j\left\{\frac{U_{n-1}(x_1,J_n\setminus\{x_i\})}{(x_1-x_j)(x_i-x_j)}-\frac{U_{n-1}(x_1,J_n\setminus\{x_j\})}{(x_1-x_i)(x_i-x_j)}\right.
\\+\left.\frac{U_{n-1}(x_1,x_1,J_n\setminus\{x_i,x_j\})}{(x_1-x_i)(x_1-x_j)}\right\},
\end{multline}
\begin{multline} \label{eqB.24}
\sum_{k_1,\ldots,k_n\geq0}\frac{\mathcal{A}_{i,1,j}}{x_1^{k_1+1}\cdots x_n^{k_n+1}}=\frac{\partial^2}{\partial x_i\partial x_j}\frac{x_ix_j}{(x_1-x_i)(x_1-x_j)}
\\\hspace{9em}\times\big\{U_{n-1}(x_1,x_1,J_n\setminus\{x_i,x_j\})-U_{n-1}(x_1,J_n\setminus\{x_j\})
\\+U_{n-1}(J_n)-U_{n-1}(x_1,J_n\setminus\{x_i\})\big\},
\end{multline}
\begin{multline} \label{eqB.25}
\sum_{k_1,\ldots,k_n\geq0}\frac{\mathcal{A}_{i,i,j}}{x_1^{k_1+1}\cdots x_n^{k_n+1}}=\frac{\partial^2}{\partial x_i\partial x_j}x_ix_j\left\{\frac{U_{n-1}(J_n)}{(x_1-x_j)(x_i-x_j)}-\frac{U_{n-1}(x_i,J_n\setminus\{x_j\})}{(x_1-x_i)(x_i-x_j)}\right.
\\+\left.\frac{U_{n-1}(x_1,J_n\setminus\{x_j\})}{(x_1-x_i)(x_1-x_j)}\right\},
\end{multline}
\begin{multline} \label{eqB.26}
\sum_{k_1,\ldots,k_n\geq0}\frac{\mathcal{A}_{i,j,i}}{x_1^{k_1+1}\cdots x_n^{k_n+1}}=\frac{\partial^2}{\partial x_i\partial x_j}\frac{x_ix_j}{(x_1-x_i)(x_i-x_j)}
\\\hspace{6em}\times\big\{U_{n-1}(x_1,J_n\setminus\{x_j\})-U_{n-1}(x_i,J_n\setminus\{x_j\})
\\+U_{n-1}(J_n)-U_{n-1}(x_1,J_n\setminus\{x_i\})\big\},
\end{multline}
\begin{multline} \label{eqB.27}
\sum_{k_1,\ldots,k_n\geq0}\frac{\mathcal{A}_{i,j,h}}{x_1^{k_1+1}\cdots x_n^{k_n+1}}
\\\hspace{-4em}=\frac{1}{x_1}\frac{\partial^3}{\partial x_i\partial x_j\partial x_h}\left\{\frac{x_ix_jx_hU_{n-3}(x_1,J_n\setminus\{x_i,x_j,x_h\})}{(x_1-x_i)(x_1-x_j)(x_1-x_h)}-\frac{x_1x_jx_hU_{n-3}(J_n\setminus\{x_j,x_h\})}{(x_1-x_i)(x_i-x_j)(x_i-x_h)}\right.
\\+\left.\frac{x_1x_ix_hU_{n-3}(J_n\setminus\{x_i,x_h\})}{(x_1-x_j)(x_i-x_j)(x_j-x_h)}-\frac{x_1x_ix_jU_{n-3}(J_n\setminus\{x_i,x_j\})}{(x_1-x_h)(x_i-x_h)(x_j-x_h)}\right\}.
\end{multline}
Equation \eqref{eqB.13} is simply equation \eqref{eqA.4}, while equation \eqref{eqB.14} can be obtained by applying the operator $-\frac{\partial}{\partial x_1}x_1$ to both sides of equation \eqref{eqA.3}. Likewise, equation \eqref{eqB.17} can be obtained by first interchanging $x_1\leftrightarrow x_i$ in both sides of equation \eqref{eqA.3} before applying $-\frac{\partial}{\partial x_i}x_i$ to both sides of said equation. Equation \eqref{eqB.16} is proven in a similar way to equation \eqref{eqA.6}, with steps relating to the factor $[k_i+1]_{\mathrm{mod}\,2}$ being skipped. All other computations \eqref{eqB.15}--\eqref{eqB.27}, bar a caveat on equation \eqref{eqB.18}, can be proven using the methods demonstrated in Appendix \ref{appendixC.1} above and the reasonings just given. We now conclude this appendix with a partial proof of equation \eqref{eqB.18}:

The last term on the right-hand side of equation \eqref{eqB.1} stands out in that it cannot be treated in the same manner as the other expressions considered in this appendix. Multiplying this term by $x_1^{-k_1-1}$, summing over $k_1\geq0$, and simplifying appropriately shows that 
\begin{align*}
\sum_{k_1\geq0}2\sum_{p_1+p_2+p_3=k_1-1}p_1\frac{m_{p_1+p_2,p_3,I_n}}{x_1^{k_1+1}}&=2\sum_{k_1\geq0}\sum_{p_1=0}^{k_1-1}\sum_{p_2=0}^{k_1-p_1-1}p_1\frac{m_{p_1+p_2,k_1-p_1-p_2-1,I_n}}{x_1^{k_1+1}}
\\&=2\sum_{p_1=0}^{\infty}\sum_{k_1=p_1+1}^{\infty}\sum_{p_2=0}^{k_1-p_1-1}p_1\frac{m_{p_1+p_2,k_1-p_1-p_2-1,I_n}}{x_1^{k_1+1}}
\\&=2\sum_{p_1=0}^{\infty}\sum_{k_1=0}^{\infty}\sum_{p_2=0}^{k_1}p_1\frac{m_{p_1+p_2,k_1-p_2,I_n}}{x_1^{k_1+p_1+2}}
\\&=2\sum_{p_1,p_2=0}^{\infty}\sum_{k_1=p_2}^{\infty}p_1\frac{m_{p_1+p_2,k_1-p_2,I_n}}{x_1^{k_1+p_1+2}}
\\&=2\sum_{p_1,p_2,k_1=0}^{\infty}p_1\frac{m_{p_1+p_2,k_1,I_n}}{x_1^{k_1+p_1+p_2+2}}
\\&=2\sum_{q,k_1=0}^{\infty}\sum_{p_1+p_2=q}p_1\frac{m_{q,k_1,I_n}}{x_1^{k_1+1}x_1^{q+1}}
\\&=\sum_{q,k_1=0}^{\infty}q(q+1)\frac{m_{q,k_1,I_n}}{x_1^{k_1+1}x_1^{q+1}}.
\end{align*}
Further multiplying this result by $x_2^{-k_2-1}\cdots x_n^{-k_n-1}$ and summing over $k_2,\ldots,k_n\geq0$ then shows that
\begin{multline*}
\sum_{k_1,\ldots,k_n\geq0}2\sum_{p_1+p_2+p_3=k_1-1}p_1\frac{m_{p_1+p_2,p_3,I_n}}{x_1^{k_1+1}\cdots x_n^{k_n+1}}=\sum_{q,k_1,\ldots,k_n\geq0}q(q+1)\frac{m_{q,k_1,I_n}}{x_1^{k_1+1}x_1^{q+1}x_2^{k_2+1}\cdots x_n^{k_n+1}}
\\=\lim_{\xi\to x_1}\left(x_1^2\frac{\partial^2}{\partial x_1^2}+2x_1\frac{\partial}{\partial x_1}\right)\sum_{k_1,\ldots,k_n\geq0}\frac{m_{q,k_1,I_n}}{\xi^{k_1+1}x_1^{q+1}x_2^{k_2+1}\cdots x_n^{k_n+1}}.
\end{multline*}
The final expression here is manifestly the third line of equation \eqref{eqB.18}.

\setcounter{equation}{0}
\section{Evaluation of \texorpdfstring{$c_{k_1}^{(\mathcal{J}_1),0}$}{c.k1(J1),0}, \texorpdfstring{$c_{k_1}^{(\mathcal{J}_1),1}$}{c.k1(J1),1}, and \texorpdfstring{$c_{k_1,k_2}^{(\mathcal{J}_1),0}$}{c.k1,k2(J1),0}} \label{appendixD.1}
As discussed in \S\ref{s4.3}, setting $x(z_i)=1/(z_i^3+z_i)$ \eqref{eq4.44} in the residue formula \eqref{eq4.49}
\begin{equation} \label{eqC.1}
c_{k_1,\ldots,k_n}^{(\mathcal{J}_1),l}=(-1)^n\underset{z_1,\ldots,z_n=0}{\mathrm{Res}}W_n^{(\mathcal{J}_1),l}(x(z_1),\ldots,x(z_n))\prod_{i=1}^nx(z_i)^{k_i}x'(z_i)\,\mathrm{d}z_i
\end{equation}
and substituting in $W_1^{(\mathcal{J}_1),0}(x(z_1))=z_1$ along with the specifications \eqref{eq4.46} and \eqref{eq4.45} of, respectively, $W_1^{(\mathcal{J}_1),1}(x(z_1))$ and $W_2^{(\mathcal{J}_1),0}(x(z_1),x(z_2))$ enables the computation of the cumulant expansion coefficients $c_{k_1}^{(\mathcal{J}_1),0}$, $c_{k_1}^{(\mathcal{J}_1),1}$, and $c_{k_1,k_2}^{(\mathcal{J}_1),0}$ defined implicitly through equation \eqref{eq4.47}.

Letting $\mathcal{J}_1$ be drawn from the $(N,N)$ antisymmetrised Laguerre ensemble, we have that
\begin{align}
c_{k_1}^{(\mathcal{J}_1),0}&=\lim_{N\to\infty}\frac{1}{N}\mean{\Tr \mathcal{J}_1^{k_1}}, \label{eqC.2}
\\ c_{k_1}^{(\mathcal{J}_1),1}&=\lim_{N\to\infty}\left(\mean{\Tr \mathcal{J}_1^{k_1}}-Nc_{k_1}^{(\mathcal{J}_1),0}\right). \label{eqC.3}
\end{align}
Due to the antisymmetric nature of $\mathcal{J}_1$, the cumulant expansion coefficients $c_{k_1}^{(\mathcal{J}_1),0}$ and $c_{k_1}^{(\mathcal{J}_1),1}$ vanish for odd values of $k_1$, while for $k_1$ even, the former relates to the order 3 Fuss--Catalan numbers \eqref{eq4.40} through the relation
\begin{equation} \label{eqC.4}
c_{k_1}^{(\mathcal{J}_1),0}=m_{k_1}^{(\mathcal{J}_1),0}=\mathrm{i}^{k_1}m_{k_1/2}^{(\mathrm{FC}_3)}=\frac{\mathrm{i}^{k_1}}{k_1+1}\binom{3k_1/2}{k_1/2}.
\end{equation}
We now present the values of $c_{k_1}^{(\mathcal{J}_1),0}$ and $c_{k_1}^{(\mathcal{J}_1),1}$ computed through the residue formula \eqref{eqC.1} with $0\leq k_1\leq 18$ even --- our calculations are consistent with equation \eqref{eqC.4} above.

\begin{center}
\small
\begin{tabular}{c|c|c|c|c|c|c|c|c|c|c} 
$k_1$&$0$&$2$&$4$&$6$&$8$&$10$&$12$&$14$&$16$&$18$ \\ \hline
$c_{k_1}^{(\mathcal{J}_1),0}$&$1$&$-1$&$3$&$-12$&$55$&$-273$&$1\,428$&$-7\,752$&$43\,263$&$-246\,675$ \\ \hline
$c_{k_1}^{(\mathcal{J}_1),1}$&$0$&$1$&$-5$&$28$&$-165$&$1\,001$&$-6\,188$&$38\,760$&$-245\,157$&$1\,562\,275$
\end{tabular}
\end{center}
Comparing the second row of the above table to equation \eqref{eqC.4} suggests that for $k_1\in2\mathbb{N}$,
\begin{multline*}
c_{k_1}^{(\mathcal{J}_1),1}=-\mathrm{i}^{k_1}\binom{3k_1/2-1}{k_1/2-1}+\chi_{k_1=0}=\frac{\chi_{k_1=0}}{3}-\frac{(k_1+1)}{3}c_{k_1}^{(\mathcal{J}_1),0}
\\\iff W_1^{(\mathcal{J}_1),1}(x(z_1))=\frac{x(z_1)}{3x'(z_1)}\frac{\partial}{\partial z_1}W_1^{(\mathcal{J}_1),0}(x(z_1))+\frac{1}{3x(z_1)}=\frac{x(z_1)}{3x'(z_1)}+\frac{1}{3x(z_1)},
\end{multline*}
which is precisely in keeping with equation \eqref{eq4.46}.

For $n=2$, the analogue of equations \eqref{eqC.2} and \eqref{eqC.3} is
\begin{equation} \label{eqC.5}
c_{k_1,k_2}^{(\mathcal{J}_1),0}=\lim_{N\to\infty}\left(\mean{\Tr \mathcal{J}_1^{k_1}\Tr \mathcal{J}_1^{k_2}}-\mean{\Tr \mathcal{J}_1^{k_1}}\mean{\Tr \mathcal{J}_1^{k_2}}\right).
\end{equation}
Note that the right-hand side of equation \eqref{eqC.5} is identically zero whenever either of $k_1,k_2$ are odd, since $\mathcal{J}_1$ is antisymmetric, and also whenever either of $k_1,k_2$ equal zero, since $\Tr \mathcal{J}_1^0=N$ factors out of the first covariance therein. Thus, we present the values of $c_{k_1,k_2}^{(\mathcal{J}_1),0}$ computed through equation \eqref{eqC.1} with $2\leq k_1,k_2\leq 14$ even --- the asterisks represent redundant data, as $c_{k_2,k_1}^{(\mathcal{J}_1),0}=c_{k_1,k_2}^{(\mathcal{J}_1),0}$.

\begin{center}
\small
\begin{tabular}{c c|c|c|c|c|c|c|c}
&$k_1$&$2$&$4$&$6$&$8$&$10$&$12$&$14$
\\$k_2$&$c_{k_1,k_2}^{(\mathcal{J}_1),0}$&&&&&&&
\\ \hline $2$&&$12$&$-80$&$504$&$-3\,168$&$20\,020$&$-127\,296$&$813\,960$
\\ \hline $4$&&$*$&$600$&$-4\,032$&$26\,400$&$-171\,600$&$1\,113\,840$&$-7\,235\,200$
\\ \hline $6$&&$*$&$*$&$28\,224$&$-190\,080$&$1\,261\,260$&$-8\,316\,672$&$54\,698\,112$
\\ \hline $8$&&$*$&$*$&$*$&$1\,306\,800$&$-8\,808\,800$&$58\,810\,752$&$-390\,700\,800$
\\ \hline $10$&&$*$&$*$&$*$&$*$&$60\,120\,060$&$-405\,437\,760$&$2\,715\,913\,200$
\\ \hline $12$&&$*$&$*$&$*$&$*$&$*$&$2\,756\,976\,768$&$-18\,597\,358\,080$
\\ \hline $14$&&$*$&$*$&$*$&$*$&$*$&$*$&$126\,196\,358\,400$
\end{tabular}
\end{center}
Comparing this table to Table 1 of \citep{DF20} suggests that $c_{k_1,k_2}^{(\mathcal{J}_1),0}$ relates to the analogous quantity $c_{k_1,k_2}^{(c\mathcal{W}_2),0}$ of the $(N,N,N)$ complex Wishart product ensemble according to
\begin{equation*}
c_{2k_1,2k_2}^{(\mathcal{J}_1),0}=4\,(-1)^{k_1+k_2}\,c_{k_1,k_2}^{(c\mathcal{W}_2),0},\qquad k_1,k_2\in\mathbb{N},
\end{equation*}
which is in agreement with the discussion leading up to equation \eqref{eq4.39}.

\setcounter{equation}{0}
\section{Evaluation of \texorpdfstring{$c_{k_1}^{(\mathcal{H}_1),0}$}{c.k1(H1),0}, \texorpdfstring{$c_{k_1}^{(\mathcal{H}_1),2}$}{c.k1(H1),2}, and \texorpdfstring{$c_{k_1,k_2}^{(\mathcal{H}_1),0}$}{c.k1,k2(H1),0}} \label{appendixD.2}
In parallel to Appendix \ref{appendixD.1}, let us recall from \S\ref{s5.3} that setting $x(z_i)=(z_i^2+1)^2/z_i$ \eqref{eq5.23} in the residue formula \eqref{eq5.27}
\begin{equation} \label{eqD.1}
c_{k_1,\ldots,k_n}^{(\mathcal{H}_1),l}=(-1)^n\underset{z_1,\ldots,z_n=0}{\mathrm{Res}}W_n^{(\mathcal{H}_1),l}(x(z_1),\ldots,x(z_n))\prod_{i=1}^nx(z_i)^{k_i}x'(z_i)\,\mathrm{d}z_i
\end{equation}
and then substituting in $W_1^{(\mathcal{H}_1),0}(x(z_1))=y(z_1)$ \eqref{eq5.22} along with the expressions for $W_1^{(\mathcal{H}_1),2}(x(z_1))$ and $W_2^{(\mathcal{H}_1),0}(x(z_1),x(z_2))$ given respectively in equations \eqref{eq5.25} and \eqref{eq5.24} gives us a means to compute the coefficients $c_{k_1}^{(\mathcal{H}_1),0}$, $c_{k_1}^{(\mathcal{H}_1),2}$, and $c_{k_1,k_2}^{(\mathcal{H}_1),0}$ of the genus expansions \eqref{eq5.21} of the associated mixed cumulants $c_{k_1,\ldots,k_n}^{(\mathcal{H}_1)}$.

With $\mathcal{H}_1$ drawn from the $(N,N)$ Hermitised Laguerre ensemble, the relevant analogues of equations \eqref{eqC.2}, \eqref{eqC.3}, \eqref{eqC.5} are
\begin{align*}
c_{k_1}^{(\mathcal{H}_1),0}&=\lim_{N\to\infty}\frac{1}{N}\mean{\Tr  \mathcal{H}_1^{k_1}},
\\ c_{k_1}^{(\mathcal{H}_1),2}&=\lim_{N\to\infty}\left(N\mean{\Tr  \mathcal{H}_1^{k_1}}-N^2c_{k_1}^{(\mathcal{H}_1),0}\right),
\\ c_{k_1,k_2}^{(\mathcal{H}_1),0}&=\lim_{N\to\infty}\left(\mean{\Tr  \mathcal{H}_1^{k_1}\Tr  \mathcal{H}_1^{k_2}}-\mean{\Tr  \mathcal{H}_1^{k_1}}\mean{\Tr  \mathcal{H}_1^{k_2}}\right).
\end{align*}
Here, we have used the fact that the $l=1$ coefficient $c_{k_1}^{(\mathcal{H}_1),1}$ within the expansion \eqref{eq5.21} is zero.

Recall from Lemma \ref{Lemma2} that $c_{k_1}^{(\mathcal{H}_1)}$, hence $c_{k_1}^{(\mathcal{H}_1),0}$ and $c_{k_1}^{(\mathcal{H}_1),2}$, vanish whenever $k_1$ is odd. Thus, we display the values of these expansion coefficients obtained through equation \eqref{eqD.1} for $0\leq k_1\leq18$ even.

\begin{center}
\small
\begin{tabular}{c|c|c|c|c|c|c|c|c|c|c} 
$k_1$&$0$&$2$&$4$&$6$&$8$&$10$&$12$&$14$&$16$&$18$ \\ \hline
$c_{k_1}^{(\mathcal{H}_1),0}$&$1$&$1$&$4$&$22$&$140$&$969$&$7\,084$&$53\,820$&$420\,732$&$3\,362\,260$ \\ \hline
$c_{k_1}^{(\mathcal{H}_1),2}$&$0$&$1$&$34$&$645$&$9\,828$&$133\,620$&$1\,694\,154$&$20\,490\,470$&$239\,545\,800$&$2\,729\,482\,668$
\end{tabular}
\end{center}
In keeping with the discussion below equation \eqref{eq5.16}, it is expected from the relation
\begin{equation*}
W_1^{(\mathcal{H}_1),0}(x_1)=x_1W_1^{(\mathrm{FC}_4)}(x_1^2)
\end{equation*}
that for even integers $k_1$, $c_{k_1}^{(\mathcal{H}_1),0}$ relates to the order 4 Fuss--Catalan numbers \eqref{eq4.40} according to
\begin{equation*}
c_{k_1}^{(\mathcal{H}_1),0}=m_{k_1/2}^{(\mathrm{FC}_4)}=\frac{1}{3k_1/2+1}\binom{2k_1}{k_1/2}.
\end{equation*}
It can readily be checked that this is in line with the values of $c_{k_1}^{(\mathcal{H}_1),0}$ given in the first row of the above table.

For the same reason as given below equation \eqref{eqC.5}, $c_{k_1,k_2}^{(\mathcal{H}_1),0}$ vanishes when either of $k_1,k_2$ are zero. Thus, we give evaluations of this cumulant expansion coefficient for $1\leq k_1,k_2\leq 9$ --- the asterisks have the same meaning as in the table for $c_{k_1,k_2}^{(\mathcal{J}_1),0}$ displayed earlier.

\begin{center}
\small
\begin{tabular}{c c|c|c|c|c|c|c|c|c|c}
&$k_1$&$1$&$2$&$3$&$4$&$5$&$6$&$7$&$8$&$9$
\\$k_2$&$c_{k_1,k_2}^{(\mathcal{H}_1),0}$&&&&&&&&&
\\ \hline $1$&&$2$&$0$&$15$&$0$&$120$&$0$&$1\,001$&$0$&$8\,568$
\\ \hline $2$&&$*$&$12$&$0$&$112$&$0$&$990$&$0$&$8\,736$&$0$
\\ \hline $3$&&$*$&$*$&$150$&$0$&$1\,350$&$0$&$12\,012$&$0$&$107\,100$
\\ \hline $4$&&$*$&$*$&$*$&$1\,176$&$0$&$11\,088$&$0$&$101\,920$&$0$
\\ \hline $5$&&$*$&$*$&$*$&$*$&$12\,960$&$0$&$120\,120$&$0$&$1\,101\,600$
\\ \hline $6$&&$*$&$*$&$*$&$*$&$*$&$108\,900$&$0$&$1\,029\,600$&$0$
\\ \hline $7$&&$*$&$*$&$*$&$*$&$*$&$*$&$1\,145\,144$&$0$&$10\,720\,710$
\\ \hline $8$&&$*$&$*$&$*$&$*$&$*$&$*$&$*$&$9\,937\,200$&$0$
\\ \hline $9$&&$*$&$*$&$*$&$*$&$*$&$*$&$*$&$*$&$101\,959\,200$
\end{tabular}
\end{center}
Observe that $c_{k_1,k_2}^{(\mathcal{H}_1),0}=0$ whenever $k_1+k_2$ is odd, which agrees precisely with what is known from Lemma \ref{Lemma2}.

\newpage
\small
\bibliographystyle{plainnat}

\normalsize
\end{document}